\documentclass[a4paper,11pt]{article}
\usepackage[margin=.7in, headheight=13.6pt]{geometry}
\usepackage{layout}
\usepackage[english]{babel}
\usepackage[utf8]{inputenc}
\usepackage{fancyhdr}
\usepackage{xspace}
\usepackage{tabularx}
\usepackage{algorithm}
\usepackage{algorithmicx}
\usepackage{algcompatible}
\usepackage{algpseudocode}
\usepackage[toc,page]{appendix}
\usepackage{graphicx} 
\usepackage{listings}
\usepackage{color}
\usepackage{caption}
\usepackage{ragged2e}
\usepackage{url}
\usepackage{hyperref}
\usepackage{cleveref}
\usepackage{amsmath}
\usepackage{mathtools}
\usepackage{amsmath,amssymb,latexsym} 
\usepackage{fancyvrb}
\usepackage{float}
\usepackage{cite}
\usepackage{epstopdf}
\usepackage{epsfig}
\usepackage{lipsum}
\usepackage{booktabs,chemformula}
\usepackage{tabularx,booktabs} 
\usepackage{etoolbox}
\usepackage{breqn}
\usepackage{relsize}
\usepackage{microtype}
\usepackage{subfiles}
\usepackage{authblk}
\usepackage{soul}
\usepackage{pdfpages}
\usepackage{calligra}
\usepackage{calrsfs}
\usepackage{ifpdf}
\DeclareMathAlphabet{\mathcalligra}{T1}{calligra}{m}{n}
\usepackage{enumitem}
\usepackage{tcolorbox}

\hypersetup{
    colorlinks=true,
    linkcolor=blue,
    filecolor=magenta,      
    urlcolor=cyan,
}

\usepackage{xparse}
\let\oldState\State
\RenewDocumentCommand{\State}{o}{
  \IfValueTF{#1}{\makeatletter\setcounter{ALG@line}{#1}\addtocounter{ALG@line}{-1}\def\verbatim@font{\linespread{1}\normalfont\ttfamily}
\preto{\@verbatim}{\topsep=0pt \partopsep=0pt }\patchcmd{\@verbatim}
  {\verbatim@font}
  {\verbatim@font\small}
  {}{}\makeatother}{}%
  \oldState\ignorespaces%
}%

 \usepackage{setspace}
 
\algrenewcommand{\algorithmiccomment}[1]{$//$ #1}
 \usepackage{stmaryrd}
 \usepackage{lastpage}
\fancypagestyle{plain}{%
  \fancyhf{}%
}

\newcommand{\punt}[1]{}
\newcommand{\cmnt}[1]{}

\newtheorem{theorem}{Theorem}

\newtheorem{lemma}[theorem]{Lemma}
\newtheorem{corollary}[theorem]{Corollary}

\newtheorem{observation}[theorem]{\textbf{Observation}}
\newtheorem{assumption}[theorem]{\textbf{Assumption}}

\newtheorem{definition}[theorem]{Definition}

\newcounter{history}

\newcounter{Step}[section]

\newenvironment{proof}[1][Proof]{\noindent\textbf{#1.} }{\hfill $\Box$\\[0.4mm]} 

\newcommand{\ignore}[1]{}

\newcommand{\secref}[1]{Section~\ref{sec:#1}}
\newcommand{\figref}[1]{Figure~\ref{fig:#1}}

\newcommand{\stref}[1]{Step~\ref{step:#1}}

\newcommand{\thmref}[1]{Theorem~\ref{thm:#1}}
\newcommand{\lemref}[1]{Lemma~\ref{lem:#1}}

\newcommand{\defref}[1]{Definition~\ref{def:#1}}

\newcommand{\obsref}[1]{Observation~\ref{obs:#1}}
\newcommand{\asmref}[1]{Assumption~\ref{asm:#1}}

\newcommand{\lineref}[1]{Line~\ref{lin:#1}}

\newcommand{\subsecref}[1]{Section~\ref{subsec:#1}}

\ignore{

\newcommand{\secref}[1]{Section~\ref{sec:#1}}
\newcommand{\figref}[1]{Figure~\ref{fig:#1}}

\newcommand{\stref}[1]{Step~\ref{step:#1}}

\newcommand{\thmref}[1]{Theorem~\ref{thm:#1}}

\newcommand{\lemref}[1]{Lemma~\ref{lem:#1}}

\newcommand{\defref}[1]{Definition~\ref{def:#1}}

\newcommand{\obsref}[1]{Observation~\ref{obs:#1}}
\newcommand{\asmref}[1]{Assumption~\ref{asm:#1}}

\newcommand{\lineref}[1]{Line~\ref{lin:#1}}

\newcommand{\subsecref}[1]{Subsection{\ref{subsec:#1}}}

}

\newtheorem{innercustomgeneric}{\customgenericname}
\providecommand{\customgenericname}{}
\newcommand{\newcustomtheorem}[2]{%
  \newenvironment{#1}[1]
  {%
   \renewcommand\customgenericname{#2}%
   \renewcommand\theinnercustomgeneric{##1}%
   \innercustomgeneric
  }
  {\endinnercustomgeneric}
}

\newcustomtheorem{customtheorem}{\textbf{Theorem}}
\newcustomtheorem{customlemma}{\textbf{Lemma}}
\newcustomtheorem{customcorollary}{\textbf{Corollary}}
\newcustomtheorem{customproposition}{\textbf{Proposition}}
\newcustomtheorem{customobservation}{\textbf{Observation}}
\newcustomtheorem{customgenerictheorem}{\textbf{Generic Theorem}}
\newcustomtheorem{customgenericlemma}{\textbf{Generic Lemma}}
\newcustomtheorem{customdsspecific}{\textbf{DS-Specific Lemma}}
\newcustomtheorem{customdefinition}{\textbf{Definition}}

\newcommand{\qed}{$\hfill \Box$}

\newcommand{\state}[1] {#1.state\xspace}
\newcommand{\alls}[1] {#1.allStates\xspace}

\newcommand{\pres}[1] {PreS[#1]\xspace}
\newcommand{\posts}[1] {PostS[#1]\xspace}
\newcommand{\pree}[1] {PreE[#1]\xspace}
\newcommand{\poste}[1] {PostE[#1]\xspace}
\newcommand{\prem}[1] {PreM[#1]\xspace}

\newcommand{\subh}  {sub\text{-}history\xspace}

\newcommand{\mth} {method\xspace}
\newcommand{\upm} {updtMethod\xspace}

\newcommand{\cc} {correctness-criterion\xspace}

\newcommand{\eevts}[1] {#1.evts}
\newcommand{\mths}[1] {#1.mths}
\newcommand{\upms}[1] {#1.updtMethods}

\newcommand{\inv} {\emph{inv}\xspace}
\newcommand{\rsp} {\emph{rsp}\xspace}

\newcommand{\sspec} {sequential\text{-}specification\xspace}
\newcommand{\lble} {linearizable\xspace}
\newcommand{\lbty} {linearizability\xspace}
\newcommand{\legal} {legal\xspace}

\newcommand{\lp} {LP\xspace}

\newcommand{\cds} {CDS\xspace}
\newcommand{\pset} {partial-set\xspace}

\ignore{
\newcommand{\pres}[1] {#1.pre\text{-}state\xspace}
\newcommand{\posts}[1] {#1.post\text{-}state\xspace}

\newcommand{\subh}  {sub\text{-}history\xspace}

\newcommand{\mth} {method\xspace}
\newcommand{\cc} {correctness-criterion\xspace}

\newcommand{\eevts}[1] {#1.evts}
\newcommand{\mths}[1] {#1.mths}
\newcommand{\inv} {inv\xspace}
\newcommand{\rsp} {rsp\xspace}

\newcommand{\sspec} {sequential\text{-}specification\xspace}
\newcommand{\legal} {legal\xspace}
\newcommand{\lp} {LP\xspace}

}

\newcommand{\lazy}{lazy-list\xspace}
\newcommand{\hoh}{hoh-locking-list\xspace}

\newcommand{\abds} {AbDS\xspace}

\newcommand{\abs} {AbDS\xspace}

\newcommand{\node}{node}
\newcommand{\nodes}[1] {#1.nodes\xspace}
\newcommand{\head}{Head\xspace}
\newcommand{\tail}{Tail\xspace}
\newcommand{\add}{Add\xspace}
\newcommand{\rem}{Remove\xspace}
\newcommand{\con}{Contains\xspace}
\newcommand{\loct}{Locate\xspace}
\newcommand{\valid}{Validate\xspace}

\newcommand{\preel}[1]{PreE[E^H.m_{#1}.LP]\xspace}
\newcommand{\postel}[1]{PostE[E^H.m_{#1}.LP]\xspace}

\newcommand{\preeds}[1]{PreE[E^H.m_{#1}.LP].AbDS\xspace}
\newcommand{\premds}[1]{PreM[E^{\spl{S}}.m_{#1}].AbDS\xspace}
\newcommand{\posteds}[1]{PostE[E^H.m_{#1}.LP].AbDS\xspace}
\newcommand{\postmds}[1]{PostM[E^{\spl{S}}.m_{#1}].AbDS\xspace}
\newcommand{\prespmds}[1]{PreM[E^{\sh{H}}.m_{#1}].AbDS\xspace}
\newcommand{\postspmds}[1]{PostM[E^{\sh{H}}.m_{#1}].AbDS\xspace}

\newcommand{\retmds}[1]{E^{\spl{S}}.m_{#1}.rsp\xspace}

\newcommand{\reteds}[1]{E^H.m_{#1}.rsp\xspace}
\newcommand{\invmds}[1]{E^{\spl{S}}.m_{#1}.inv\xspace}
\newcommand{\inveds}[1]{E^H.m_{#1}.inv\xspace}
\newcommand{\lpmds}[1]{E^{\spl{S}}.m_{#1}.LP\xspace}
\newcommand{\lpeds}[1]{E^H.m_{#1}.LP\xspace}

\newcommand{\retspmds}[1]{E^{\sh{H}}.m_{#1}.rsp\xspace}

\newcommand{\sh}[1]{CS(#1)\xspace}
\newcommand{\spl}{\mathbb}

\newcommand{\prees}[1]{PreE[E^H.{#1}.LP].AbS\xspace}
\newcommand{\prems}[1]{PreM[E^{\spl{S}}.{#1}].AbS\xspace}
\newcommand{\postes}[1]{PostE[E^H.{#1}.LP].AbS\xspace}
\newcommand{\postms}[1]{PostM[E^{\spl{S}}.{#1}].AbS\xspace}

\newcommand{\retms}[1]{E^{\spl{S}}.{#1}.rsp\xspace}
\newcommand{\retes}[1]{E^H.{#1}.rsp\xspace}
\newcommand{\invms}[1]{E^{\spl{S}}.{#1}.inv\xspace}
\newcommand{\inves}[1]{E^H.{#1}.inv\xspace}

\newcommand{\hnode}{node\xspace}
\newcommand{\hnodes}[1] {#1.nodes\xspace}
\newcommand{\hhead}{Head}
\newcommand{\htail}{Tail}
\newcommand{\habs} {AbS\xspace}
\newcommand{\hadd}{HoHAdd\xspace}
\newcommand{\hrem}{HoHRemove\xspace}
\newcommand{\hcon}{HoHContains\xspace}
\newcommand{\hloct}{HoHLocate\xspace}

\newcommand{\cdse}{\emph{\cds Specific Equivalence}\xspace}
\newcommand{\llse}{\emph{\lazy Specific Equivalence}\xspace}
\newcommand{\hohse}{\emph{\hoh Specific Equivalence}\xspace}

\DeclareCaptionFont{white}{\color{white}}
\DeclareCaptionFormat{listing}{\colorbox{blue}{\parbox{\textwidth}{\hspace{15pt}#1#2#3}}} 
\makeatletter
\def\@copyrightspace{\relax}
\makeatother

\begin{document}
	
\title{\bf Proving Correctness of Concurrent Objects by Validating Linearization Points\thanks{A preliminary version of this work was accepted in AADDA 2018 as work in progress.} \footnote{Author sequence follows lexical order of last names.}}

\author{
      Sathya Peri, Muktikanta Sa, Ajay Singh, Nandini Singhal, Archit Somani \\
     Department of Computer Science \& Engineering \\
      Indian Institute of Technology Hyderabad, India \\
    \{sathya\_p, cs15resch11012, cs15mtech01001, cs15mtech01004, \\ cs15resch01001\}@iith.ac.in
}

\date{}

\maketitle

\begin{abstract}
\emph{Concurrent data structures} or \emph{\cds} such as concurrent stacks, queues, sets etc. have become very popular in the past few years partly due to the rise of multi-core systems. Such concurrent \cds{s} offer great performance benefits over their sequential counterparts. But one of the greatest challenges with \cds{s} has been developing correct structures and then proving correctness of these structures. We believe that techniques that help prove correctness of these \cds{s} can also guide in developing new \cds{s}. 

An intuitive \& popular techniques to prove correctness of \cds{s} is using \emph{Linearization Points} or \emph{\lp{s}}. A \lp is an (atomic) event in the execution interval of each \mth such that the execution of the entire \mth seems to have taken place in the instant of that event.
One of the main challenges with the \lp based approach is to identify the correct \lp{s} of a \cds. Identifying the correct \lp{s} can be deceptively wrong in many cases. In fact in many cases, the \lp identified or even worse the \cds itself could be wrong. To address these issues, several automatic tools for verifying linearizability have been developed. But we believe that these tools don't provide insight to a programmer to develop the correct concurrent programs or identify the LPs.

Considering the complexity of developing a \cds and verifying its correctness, we address the most basic problem of this domain in this paper: given the set of \lp{s} of a \cds, how to show its correctness? We assume that we are given a \cds and its \lp{s}. We have developed a hand-crafted technique of proving correctness of the \cds{s} by validating its \lp{s}. As observed earlier, identifying the correct \lp{s} is very tricky and erroneous. But since our technique is hand-crafted, we believe that the process of proving correctness might provide insight to identify the correct \lp{s}, if the currently chosen \lp is incorrect. We also believe that this technique might also offer the programmer some insight to develop more efficient variants of the \cds. 

The proposed proof technique can be applied to prove the correctness of several commonly used \cds{s} developed in literature such as Lock-free Linked based Sets, Skiplists etc. Our technique will also show correctness of \cds{s} in which the \lp{s} of \mth might lie outside the \mth{s} (may seem to take effect in code of other \mth{}) such as \lazy based set. To show the efficacy of this technique, we show the correctness of \lazy and \hoh based set. 
\end{abstract}

\textbf{Keywords: }linearizability; concurrent data structure; linearization points; correctness;

\section{Introduction}
\emph{Concurrent data structures} or \emph{\cds} such as concurrent stacks, queues, lists etc. have become very popular in the past few years due to the rise of multi-core systems and due to their performance benefits over their sequential counterparts. This makes the concurrent data structures highly desirable in big data applications such data structures in combination with multi-core machines can be exploited to accelerate the big data applications. But one of the greatest challenges with \cds{s} is developing correct structures and then proving their correctness either through automatic verification or through hand-written proofs \cite{Derrick-fm2011}. Also, the techniques which help to prove correctness of \cds{s} can also guide in developing new \cds{s}.

A \cds exports \mth{s} which can be invoked concurrently by different threads. A \emph{history} generated by a \cds is a collection of \mth invocation and response events. Each invocation or \inv event of a method call has a subsequent response or \rsp event which can be interleaved with invocation, responses from other concurrent \mth{s}.

To prove a concurrent data structure to be correct, \textit{\lbty} proposed by Herlihy \& Wing \cite{HerlWing:1990:TPLS} is the standard correctness criterion used. They consider a history generated by the \cds which is collection of \mth invocation and response events. Each invocation of a method call has a subsequent response which can be interleaved with invocation, responses from other concurrent \mth{s}. A history is \lble if (1) The invocation and response events can be reordered to get a valid sequential history. (2) The generated sequential history satisfies the object's sequential specification. (3) If a response event precedes an invocation event in the original history, then this should be preserved in the sequential reordering.  

\cmnt {
\begin{enumerate}
\item The invocation and response events can be reordered to get a valid sequential history.
\item The generated sequential history satisfies the object's sequential specification.
\item If a response event precedes an invocation event in the original history, then this should be preserved in the sequential reordering.  
\end{enumerate}
}

A concurrent object is linearizable if each of their histories is linearizable. Linearizability ensures that every concurrent execution simulates the behavior of some sequential execution while not actually executing sequentially and hence leveraging on the performance. 

One of the intuitive techniques to prove correctness of \cds{s} is using \emph{Linearization Points} or \emph{\lp{s}}. A \lp is an (atomic) event in the execution interval of each \mth such that the execution of the entire \mth seems to have taken place in the instant of that event. 

Several techniques have been proposed for proving \lbty: both hand-written based and through automatic verification. Many of these techniques consider lazy linked-list based concurrent set implementation,  denoted as \emph{\lazy}, proposed by Heller at al \cite{Heller-PPL2007}. This is one of the popular \cds{s} used for proving correctness due to the intricacies of \lp{s} of its \mth{s} in their execution. The \lp of an unsuccessful \emph{contains} \mth can sometimes be outside the code of its \mth{s} and depend on an concurrently executing \emph{add} \mth (refer \figref{lpcase}). This is illustrated in \figref{lpcase} of \subsecref{con-lazy-list}.  Such scenarios can also occur with other \cds{s} as well, to name a few Herlihy and Wing queue\cite{HerlWing:1990:TPLS}, the optimistic queue\cite{ladan2004optimistic}, the elimination
queue\cite{moir2005using}, the baskets queue\cite{hoffman2007baskets}, the flat-combining queue\cite{hendler2010flat}

Vafeiadis et al. \cite{Vafeiadis-ppopp2006} hand-crafted one of the earliest proofs of linearizability for \lazy using the rely-guarantee approach \cite{Jones:RG:IFIP:1983} which can be generalized to other \cds{s} as well. O'Hearn et al. \cite{O'Hearn-PODC2010} have developed a generic methodology for linearizability by identifying new property known as \textit{Hindsight} lemma. Their technique is non-constructive in nature. Both these techniques don't depend on the notion of \lp{s}. 

Recently Lev-Ari et al. \cite{Lev-Aridisc2014, Lev-Aridisc2015} proposed a constructive methodology for proving correctness of \cds{s}. They have developed a very interesting notion of base-points and base-conditions to prove \lbty. Their methodology manually identifies the base conditions, commuting steps, and base point preserving steps and gives a roadmap for proving correctness by writing semi-formal proofs. Their seminal technique, does not depend on the notion of \lp{s}, can help practitioners and researchers from other fields to develop correct \cds{s}.

In spite of several such techniques having been proposed for proving \lbty, \lp{s} continue to remain most popular guiding tool for developing efficient \cds{s} and illustrating correctness of these \cds{s} among practitioners. \lp{s} are popular since they seem intuitive and more importantly are constructive in nature. In fact, we believe using the notion of \lp{s}, new \cds can be designed as well.

But one of the main challenges with the \lp based approach is to identify the correct \lp{s} of a \cds. Identifying the correct \lp{s} can be deceptively wrong in many cases. For instance, it is not obvious to a novice developer that the \lp of an unsuccessful \emph{contains} \mth of \lazy could be outside the \emph{contains} \mth. In fact in many cases, the \lp identified or even worse the \cds could be wrong. 

The problem of proving correctness of \cds using \lp{s} has been quite well explored in the verification community in the past few years. Several efficient automatic proving tools and techniques have been developed \cite{Liufm09, Amitcav07, Vafeiadiscav10, Zhangicse11, Zhu+:Poling:CAV:2015, Bouajjanipopl15} to address this issue. In fact, many of these tools can also show correctness even without the information of \lp{s}. But very little can be gleaned from these techniques to identify the correct \lp{s} of a \cds by a programmer. Nor do they provide any insight to a programmer to develop new \cds{s} which are correct. The objective of the most of these techniques has been to efficiently automate proving correctness of already developed \cds{s}. 

Considering the complexity of developing a \cds and verifying its correctness, we address the most basic problem of this domain in this paper: given the set of \lp{s} of a \cds, how to show its correctness? We assume that we are given a \cds and its \lp{s}. We have developed a hand-crafted technique of proving correctness of the \cds{s} by validating its \lp{s}. We believe that our technique can be applied to prove the correctness of several commonly used \cds{s} developed in literature such as Lock-free Linked based Sets \cite{Valoispodc1995}, \hoh \cite{Bayerai1977, MauriceNir} , \lazy \cite{Heller-PPL2007, MauriceNir}, Skiplists \cite{Levopodis2006} etc. Our technique will also work for \cds{s} in which the \lp{s} of a \mth might lie outside the \mth such as \lazy. To show the efficacy of this technique, we show the correctness of \lazy and hand-over-hand locking list (\emph{\hoh}) \cite{Bayerai1977, MauriceNir}. 

As observed earlier, identifying the correct \lp{s} is very tricky and erroneous. But since our technique is hand-crafted, we believe that the process of proving correctness might provide insight to identify the correct \lp{s}, if the currently chosen \lp is incorrect. We also believe that this technique might also offer the programmer some insight to develop more efficient variants of the \cds. 


Our technique is inspired from the notion of rely-guarantee approach \cite{Jones:RG:IFIP:1983} and Vafeiadis et al. \cite{Vafeiadis-ppopp2006}. For the technique to work, we make some assumptions about the \cds and its \lp{s}. We describe the main idea here and the details in the later sections. 

\vspace{1mm}
\noindent
\textbf{Main Idea: Proving Correctness of \lp{s}.}  In this technique, we consider executions corresponding to the histories. For a history $H$, an \emph{execution} $E^H$ is a totally ordered sequence of atomic events which are executed by the threads invoking the \mth{s} of the history. Thus an execution starts from an initial \emph{global state} and then goes from one global state to the other as it executes atomic events.

With each global state, we associate the notion of \emph{abstract data-structure} or \emph{\abds}. This represents the state of the \cds if it had executed sequentially. Vafeiadis et al. \cite{Vafeiadis-ppopp2006} denote it as \emph{abstract set} or \emph{\abs} in the context of the \lazy. 

We assume that each \mth of the \cds has a unique atomic event as the \lp within its execution. Further, we assume that only a (subset) of \lp events can change the \abs. We have formalized these assumptions in \subsecref{lps}.

With these assumptions in place, to show the correctness of a history $H$, we first construct a sequential history $\sh{H}$: we order all the \mth{s} of $H$ by their \lp{s} (which all are atomic and hence totally ordered). Then based on this \mth ordering, we invoke the \mth{s} (using a single thread) with the same parameters on the \cds sequentially. The resulting history generated is sequential. The details of this construction is described in \subsecref{csh}.

Since $\sh{H}$ is generated sequentially, it can be seen that it satisfies the \sspec of the \cds. All the \mth invocations of \sh{H} respect the \mth ordering of $H$. If we can show that all the response events in $H$ and $\sh{H}$ are the same then $H$ is \lble. 

The proof of this equivalence naturally depends on the properties of the \cds being considered. We have identified a \cdse (\defref{pre-resp} of \subsecref{generic}) as a part of our proof technique, which if shown to be true for all the \mth{s} of the \cds, implies \lbty of the \cds. In this definition, we consider the pre-state of the \lp of a \mth $m_i$ in a history $H$. As the name suggests, pre-state is the global state of the \cds just before the \lp event. This definition requires that the \abds in the pre-state to be the result of some sequential execution of the \mth{s} of the \cds. Similarly, the \abds in the post-state of the \lp must be as a result of some sequential execution the \mth{s} with $m_i$ being the final \mth in the sequence. We show that if the \cds ensures these conditions then it is \lble. 



The definition that we have identified is generic. We show that any \cds for which this definition is true and satisfies our assumptions on the \lp{s}, is \lble. Thus, we would like to view this definition as an abstract class in a language like C++. It is specific to each \cds and has to be proved (like instantiation of the abstract class in C++). In \secref{ds-proofs}, we demonstrate this technique by giving a high-level overview of the correctness of this definition for \lazy and of \hoh. 

\vspace{1mm}
\noindent
\textbf{Roadmap.} In \secref{model}, we describe the system model. In \secref{gen-proof}, we describe the proof technique. In \secref{ds-proofs}, we illustrate this technique by giving outline of the proof for \lazy and \hoh. Finally, we conclude in \secref{conc}. 

\ignore{


}

\section{System Model \& Preliminaries}
\label{sec:model}
In this paper, we assume that our system consists of finite set of $p$ processors, accessed by a finite set of $n$ threads that run in a completely asynchronous manner and communicate using shared objects. The threads communicate with each other by invoking higher-level \emph{methods} on the shared objects and obtaining the corresponding responses. Consequently, we make no assumption about the relative speeds of the threads. We also assume that none of these processors and threads fail. We refer to a shared objects as a \emph{concurrent data-structure} or \emph{\cds}.  

\vspace{1mm}
\noindent 
\textbf{Events \& Methods.} We assume that the threads execute atomic \emph{events}. Similar to Lev{-}Ari et. al.'s work, \cite{Lev-Aridisc2015, Lev-Aridisc2014} we assume that these events by different threads are (1) atomic \emph{read, write} on shared/local memory objects; (2) atomic \emph{read-modify-write} or \emph{rmw} operations such compare \& swap etc. on shared memory objects (3) method invocation or \emph{\inv} event \& response or \emph{\rsp} event on \cds{s}. 

A thread executing a \mth $m_i$, starts with the \inv event, say $inv_i$, executes the events in the $m_i$ until the final \rsp event $rsp_i$. The \rsp event $rsp_i$ of $m_i$ is said to \emph{match} the \inv event $inv_i$. On the other hand, if the \inv event $inv_i$ does not have a \rsp event $rsp_i$ in the execution, then we say that both the \inv event $inv_i$ and the \mth $m_i$ are \emph{pending}.

The \mth \inv \& \rsp events are typically associated with invocation and response parameters. The invocation parameters are passed as input while response parameters are obtained as output to and from the \cds respectively. For instance, the invocation event of the enqueue \mth on a queue object $Q$ is denoted as $\inv(Q.enq(v))$ while the \rsp event of a dequeue \mth can be denoted as $\rsp(Q.deq(v))$. We combine the \inv and \rsp events to represent a \mth as follows: $m_i(inv\text{-}params, rsp\text{-}params)$ where $\inv(m_i(inv\text{-}params))$ and $\rsp(m_i(rsp\text{-}params))$ represent the \inv, \rsp events respectively. For instance, we represent enqueue as $enq(v, ok)$, or a successful add to a set as $add(k, T)$. If there are multiple invocation or response parameters, we use delimiters to differentiate them. In most cases, we ignore these invocation and response parameters unless they are required for the context and denote the \mth as $m_i$. In such a case, we simply denote $m_i.inv, m_i.rsp$ as the \inv and \rsp events. 

\vspace{1mm}
\noindent
\textbf{Global States, Execution and Histories.} We define the \emph{global state} or \emph{state} of the system as the collection of local and shared variables across all the threads in the system. The system starts with an initial global state. Each event changes possibly the global state of the system leading to a new global state. The events read, write, rmw on shared/local memory objects change the global state. The $\inv$ \& $\rsp$ events on higher level shared-memory objects do not change the contents of the global state. Although we denote the resulting state with a new label in this case. 

We denote an \emph{execution} of a concurrent threads as a finite sequence of totally ordered atomic events. We formally denote an execution $E$ as the tuple $\langle evts, <_E \rangle$, where $\eevts{E}$ denotes the set of all events of $E$ and $<_E$ is the total order among these events. A \emph{history} corresponding to an execution consists only of \mth $\inv$ and $\rsp$ events (in other words, a history views the \mth{s} as black boxes without going inside the internals). Similar to an execution, a history $H$ can be formally denoted as $\langle evts, <_H \rangle$ where $evts$ are of type $\inv$ \& $\rsp$ and $<_H$ defines a total order among these events. 
With this definition, it can be seen that an execution uniquely characterizes a history. For a history $H$, we denote the corresponding execution as $E^H$. 

We denote the set of \mth{s} invoked by threads in a history $H$ (and the corresponding execution $E^H$) by $\mths{H}$ (or $\mths{E^H}$). Similarly, if a \mth $m_x$ is invoked by a thread in a history $H$ ($E^H$), we refer to it as $H.m_x$ ($E^H.m_x$). Although all the events of an execution are totally ordered in $E^H$, the \mth{s} are only partially ordered. We say that a \mth $m_x$ is ordered before \mth $m_y$ in \emph{real-time} if the \rsp event of $m_x$ precedes the invocation event of $m_y$, i.e. $(m_x.rsp <_{H} m_y.inv)$. We denote the set of all real-time orders between the \mth{s} of $H$ by $\prec^{rt}_H$. 


Next, we relate executions (histories) with global states. An execution takes the system through a series of global states with each event of the execution starting from the initial state takes the global state from one to the next. We associate the state of an execution (or history) to be the global state after the last event of the execution. We denote this final global state $S$ of an execution E as $S = \state{E}$ (or $\state{H}$). We refer to the set of all the global states that a system goes through in the course of an execution as $\alls{E}$ (or $\alls{H}$). It can be seen that for $E$, $\state{E} \in \alls{E}$. \figref{exec/hist} shows a concurrent execution $E^H$ and its corresponding history $H$. In the figure, the curved line represents an $event$ and the vertical line is a $state$. The open([) \& close(]) square brackets simply demarcate the methods of a thread and have no specific meaning in the figure.

Given an event $e$ of an execution $E$, we denote global state just before the $e$ as the pre-state of $e$ and denote it as $\pree{e}$. Similarly, we denote the state immediately after $e$ as the post-state of $e$ or $\poste{e}$. Thus if an event $e$ is in $\eevts{E}$ then both $\pree{e}$ and $\poste{e}$ are in $\alls{E}$. 

\begin{figure}[!htbp]
\captionsetup{font=scriptsize}
	\centerline{\scalebox{0.65}{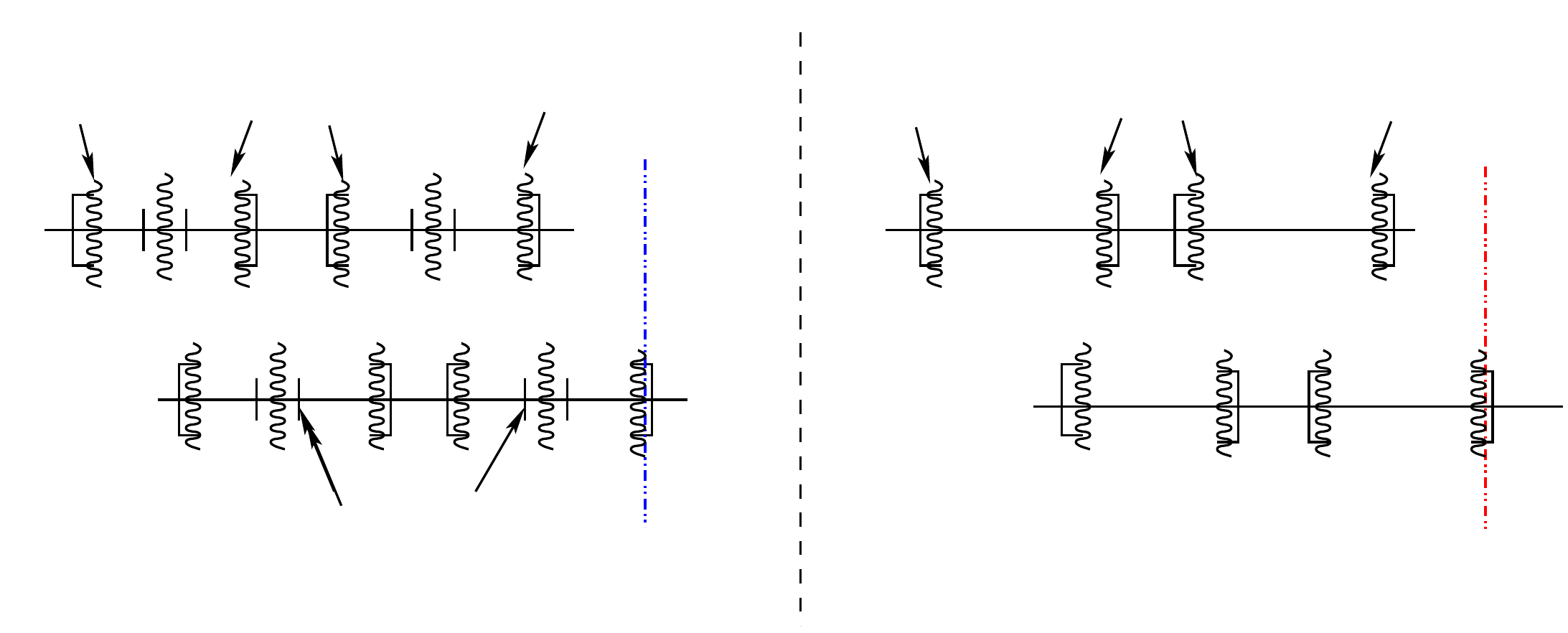}}
	\caption{Figure (a) depicts a Concurrent Execution $E^H$ comprising of multiple events $\eevts{E}$. $\state{E}$ denotes the global state after the last event of the execution. Consider a $read/write$ event $e$, then pre-state of event $e$ is $\pree{e}$ and the post-state is $\poste{e}$ and both belong to $\alls{E}$.  Figure (b) depicts the corresponding concurrent history $H$ consisting only of $inv$ and $resp$ events. $H.state$ denotes the global state after the last event in the history. }
	\label{fig:exec/hist}
\end{figure}

The notion of pre \& post states can be extended to \mth{s} as well. We denote the pre-state of a \mth $m$ or $\prem{m}$ as the global state just before the invocation event of $m$ whereas the post-state of $m$ or $\prem{m}$ as the global state just after the return event of $m$. 
\figref{conc-seq} illustrates the global states immediately before and after $m_i.LP$ which are denoted as $\preel{i}$  and $\postel{i}$ respectively in the execution $E^H$. 

\begin{figure}[!htbp]
\captionsetup{font=scriptsize}
	\centerline{\scalebox{0.65}{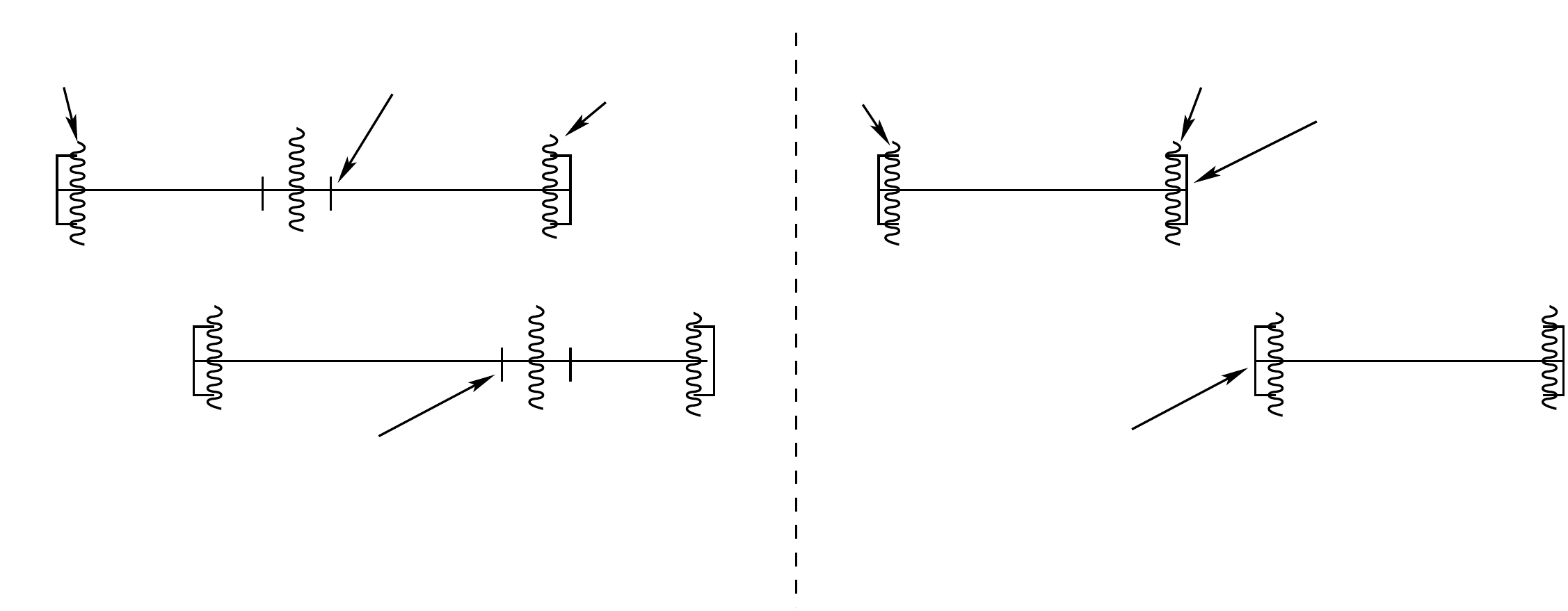}}
	\caption{Figure (a) illustrates an example of a concurrent execution $E^H$. Then, $m_i.LP$ is the $\lp$ event of the method $m_i$. The global state immediately after this event is represented as Post-state of ($E^H.m_i.LP$). Figure (b) represents sequential execution $E^{\spl{S}}$ corresponding to (a) with post-state of method $m_i$ as the state after its $resp$ event.}
	\label{fig:conc-seq}
\end{figure}

\noindent
\textbf{Notations on Histories.} We now define a few notations on histories which can be extended to the corresponding executions. We say two histories $H1$ and $H2$ are \emph{equivalent} if the set of events in $H1$ are the same as $H2$, i.e., $\eevts{H1} = \eevts{H2}$ and denote it as $H1 \approx H2$. We say history $H1$ is a \emph{\subh} of $H2$ if all the events of $H1$ are also in $H2$ in the same order, i.e., $\langle (\eevts{H1} \subseteq \eevts{H2}) \land (<_{H1} \subseteq <_{H2}) \rangle$. 
Let a thread $T_i$ invoke some \mth{s} on a few \cds{s} (shared memory objects) in a history $H$ and $d$ be a \cds whose \mth{s} have been invoked by threads in $H$. Using the notation of \cite{HerlWing:1990:TPLS}, we denote $H|T_i$ to be the \subh of all the events of $T_i$ in $H$. Similarly, we denote $H|d$ to be the \subh of all the events involving $d$.

We define that a history $H$ is \emph{well-formed} if a thread $T_i$ does not invoke the next \mth on a \cds until it obtains the matching response for the previous invocation. We assume that all the executions \& histories considered in this paper are well-formed. Note that since an execution is well-formed, there can be at most only one pending invocation for each thread. 

We say the history $H$ is \emph{complete} if for every \mth $\inv$ event there is a matching $\rsp$ event, i.e., there are no pending \mth{s} in $H$. The history $H$ is said to be \emph{sequential} if every $\inv$ event, except possibly the last, is immediately followed by the matching $\rsp$ event. In other words, all the \mth{s} of $H$ are totally ordered by real-time and hence $\prec^{rt}_H$ is a total order. Note that a complete history is not always sequential and the vice-versa. It can be seen that in a well-formed history $H$, for every thread $T_i$, we have that $H|T_i$ is sequential. \figref{seqhist} shows the execution of a sequential history $\spl{S}$. 
\begin{figure}[!htbp]
\captionsetup{font=scriptsize}
	\centerline{\scalebox{0.65}{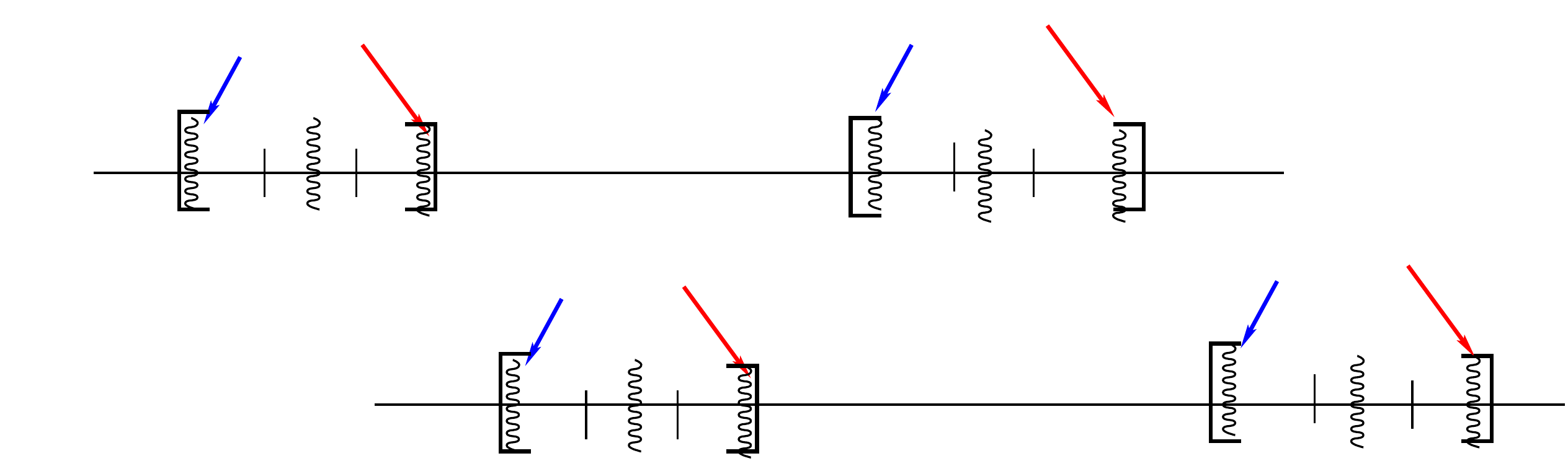}}
	\caption{An illustration of a sequential execution $E^{\spl{S}}$.}
	\label{fig:seqhist}
\end{figure}
\noindent
\\
\textbf{Sequential Specification.} 
We next discuss about \emph{\sspec} \cite{HerlWing:1990:TPLS} of \cds{s}. The \sspec of a \cds $d$ is defined as the set of (all possible) sequential histories involving the \mth{s} of $d$. Since all the histories in the \sspec of $d$ are sequential, this set captures the behavior of $d$ under sequential execution which is believed to be correct. A sequential history $\spl{S}$  is said to be \emph{\legal} if for every \cds $d$ whose \mth is invoked in $\spl{S}$, $\spl{S}|d$ is in the \sspec of $d$. 
\vspace{1mm}

\noindent
\textbf{Safety:} A \emph{safety} property is defined over histories (and the corresponding executions) of shared objects and generally states which executions of the shared objects are acceptable to any application. The safety property that we consider is \emph{linearizability} \cite{HerlWing:1990:TPLS}. A history $H$ is said to be linearizable if (1) there exists a completion $\overline{H}$ of $H$ in which some pending $\inv$ events are completed with a matching response and some other pending $\inv$ events are discarded; (2) there exists a sequential history $\spl{S}$ such that $\spl{S}$ is equivalent to $\overline{H}$, i.e., $\overline{H} \approx \spl{S}$; (3) $\spl{S}$ respects the real-time order of $H$, i.e., $\prec^{rt}_H \subseteq \prec^{rt}_{\spl{S}}$; (4) $\spl{S}$ is legal.
Another way to say that history $H$ is linearizable if it is possible to assign an atomic event as a \emph{linearization point} or \emph{\lp} inside the execution interval of each \mth such that the result of each of these \mth{s} is the same as it would be in a sequential history $\spl{S}$ in which the \mth{s} are ordered by their $\lp${s} \cite{MauriceNir}. In this document, we show how to prove the correctness of \lp{s} of the various \mth{s} of a data-structure.

\ignore{
A concurrent object is linearizable if all the histories generated by it are linearizable.  We prove the linearizability of a concurrent history by defining a $\lp${s} for each method. The $\lp$ of a \mth implies that the method appears to take effect instantly at its $\lp$. 

\textbf{Progress:} The \emph{progress} properties specifies when a thread invoking \mth{s} on shared objects completes in presence of other concurrent threads. Some progress conditions used in this paper are mentioned here which are based on the definitions in Herlihy \& Shavit \cite{opodis_Herlihy}. The progress condition of a method in concurrent object is defined as: (1) Blocking: In this, an unexpected delay by any thread (say, one holding a lock) can prevent other threads from making progress. (2) Deadlock-Free: This is a \textbf{blocking} condition which ensures that \textbf{some} thread (among other threads in the system) waiting to get a response to a \mth invocation will eventually receive it. (3) Wait-Free: This is a \textbf{non-blocking} condition which ensures that \textbf{every} thread trying to get a response to a \mth, eventually receives it.
It can be seen that wait-free \mth{s} are desirable since they can complete regardless of the execution of other concurrent \mth{s}. On the other hand, deadlock-free \mth{s} are system (or underlying scheduler) dependent progress condition since they involve blocking. It ensures that among multiple threads in the system, at least one of them will make progress. 
}

\ignore {
\section{System Model \& Preliminaries}
\label{sec:System-Model-Preliminaries}
In this paper, we assume that our system consists of finite set of $p$ processors, accessed by a finite set of $n$ threads that run in a completely asynchronous manner and communicate using shared objects. The threads communicate with each other by invoking higher-level \emph{methods} on the shared objects and getting corresponding responses. Consequently, we make no assumption about the relative speeds of the threads. We also assume that none of these processors and threads fail.
\vspace{1mm}
\noindent
\textbf{Events \& Methods.} We assume that the threads execute atomic \emph{events}. Similar to Lev{-}Ari et. al.'s work, \cite{Lev-Aridisc2015, Lev-Aridisc2014} we assume that these events by different threads are (1) atomic \emph{read, write} on shared/local memory objects; (2) atomic \emph{read-modify-write} or \emph{rmw} operations such compare \& swap etc. on shared memory objects (3) method invocations or \emph{\inv} event \& responses or \emph{\rsp} event on higher level shared-memory objects. 
The \mth \inv \& \rsp events are typically associated with invocation and response parameters. For instance, the invocation event of the enqueue \mth on a queue object $Q$ is denoted as $\inv(Q.enq(v))$ while the \rsp event of a dequeue \mth can be denoted as $\rsp(Q.deq(v))$. In most cases, we ignore these parameters unless they are required for the context. We assume that all the events executed by different threads are totally ordered. 
A method consists of all the events from its invocation until response. A representation of a \mth $m_i$ is: $m_i(inv\text{-}params\uparrow, rsp\text{-}params\downarrow)$ where $inv\text{-}params\uparrow$ are all the parameters passed during invocation and $rsp\text{-}params\downarrow$ are all the parameters returned by the data-structure. But again in most case, we ignore these parameters unless they are required for the context.
\vspace{1mm}
\noindent
\textbf{Global States, Execution and Histories.} We define the \emph{global state} or \emph{state} of the system as the collection of local and shared variables across all the threads in the system. The system starts with an initial global state. Each event changes possibly the global state of the system leading to a new global state. The events read, write, rmw on shared/local memory objects change the global state. The $\inv$ \& $\rsp$ events on higher level shared-memory objects do not change the contents of the global state. Although we denote the resulting state with a new label in this case. 
We denote an \emph{execution} of a concurrent threads as a finite sequence of totally ordered atomic events. We formally denote an execution $E$ as the tuple $\langle evts, <_E \rangle$, where $\eevts{E}$ denotes the set of all events of $E$ and $<_E$ is the total order among these events. A \emph{history} corresponding to an execution consists only of \mth $\inv$ and $\rsp$ events (in other words, a history views the \mth{s} as black boxes without going inside the internals). Similar to an execution, a history $H$ can be formally denoted as $\langle evts, <_H \rangle$ where $evts$ are of type $\inv$ \& $\rsp$ and $<_H$ defines a total order among these events. We denote the set of \mth{s} invoked by threads in a history $H$ (and the corresponding execution $E$) by $\mths{H}$ (or $\mths{E}$). 
With this definition, it can be seen that a history uniquely characterizes an execution and vice-versa. Thus we use these terms interchangeably in our discussion. For a history $H$, we denote the corresponding execution as $E^H$. 
Next, we relate executions (histories) with global states. An execution takes the system through a series of global states with each event of the execution stating from the initial state takes the global state from one to the next. We associate the state of an execution (or history) to be global state after the last event of the execution. We denote this final global state $S$ of an execution E as $S = \state{E}$ (or $\state{H}$). 
We refer to the set of all the global states that a system goes through in the course of an execution as $\alls{E}$ (or $\alls{H}$). It can be seen that for $E$, $\state{E} \in \alls{E}$. \figref{exec/hist} shows a concurrent execution $E^H$ and its corresponding history $H$. In the figure, the curved line represents an $event$ and the vertical line is a $state$. The open([) \& close(]) square brackets demarcate the methods of a thread and have no specific meaning in the figure.
\vspace{-1mm}
\begin{figure}[!htbp]
\captionsetup{font=scriptsize}
	\centerline{\scalebox{0.65}{\input{figs/ExeHist.pdf_t}}}
	\caption{Figure (a) depicts a Concurrent Execution $E^H$ comprising of multiple events. $E.state$ denotes the global state after the last event of the execution. Figure (b) depicts the corresponding concurrent history $H$ consisting only of $inv$ and $resp$ events. $H.state$ denotes the global state after the last event in the history. \Comment{We need to add $\pres{e}$, $\posts{e}$, $\eevts{E}$, $\alls{E}$, etc. }}
	\label{fig:exec/hist}
\end{figure}
Given an event $e$ of an execution $E$, we denote global state just before the $e$ as the pre-state of $e$ and denote it as $\pres{e}$. Similarly, we denote the state immediately after $e$ as the the post-state of $e$ or $\posts{e}$. Thus if an event $e$ is in $\eevts{E}$ then both $\pres{e}$ and $\posts{e}$ are in $\alls{E}$. 
The notion of pre \& post states can be extended to \mth{s} as well. We denote the pre-state of a \mth $m$ or $\pres{m}$ as the global state just before the invocation event of $m$, whereas the post-state of $m$ or $\pres{m}$ as the global state just after the return event of $m$. 

\vspace{1mm}
\noindent
\textbf{Notations on Histories.} We now define a few notations on histories which can be extended to the corresponding executions. We say two histories $H1$ and $H2$ are \emph{equivalent} if the set of events in $H1$ are the same as $H2$, i.e., $\eevts{H1} = \eevts{H2}$ and denote it as $H1 \approx H2$. We say history $H1$ is a \emph{\subh} of $H2$ if all the events of $H1$ are also in $H2$ in the same order, i.e., $\langle (\eevts{H1} \subseteq \eevts{H2}) \land (<_{H1} \subseteq <_{H2}) \rangle$. 
Let a thread $T_i$ invoke some \mth{s} on a few shared memory objects in a history $H$ and $o$ be a shared memory object whose \mth{s} have been invoked by threads in $H$. Using the notation of \cite{HerlWing:1990:TPLS}, we denote $H|T_i$ to be the \subh all the events of $T_i$ in $H$. Similarly, we denote $H|o$ to be the \subh all the events involving $o$.
We say that a $\rsp$ event \emph{matches} an $\inv$ event in an execution (as the name suggests) if the both $\rsp$ and $\inv$ events are on the same \mth of the object, have been invoked by the same thread and the $\rsp$ event follows $\inv$ event in the execution. 
We assume that a history $H$ is \emph{well-formed} if a thread $T_i$ does not invoke a \mth on a shared object until it obtains the matching response for the previous invocation. We assume that all the executions \& histories considered in this paper are well-formed. We denote an $\inv$ event as \emph{pending} if it does not have a matching $\rsp$ event in the execution. Note that since an execution is well-formed, there can be at most only one pending invocation for each thread. 
We say a history $H$ is \emph{complete} if for every \mth $\inv$ event there is a matching $\rsp$ event. The history $H$ is said to be \emph{sequential} if every $\inv$ event, except possibly the last, is immediately followed by the matching $\rsp$ event. Note that a complete history is not sequential and the vice-versa. It can be seen that in a well-formed history $H$, for every thread $T_i$, we have that $H|T_i$ is sequential. \figref{seqhist} shows a the execution of a sequential history $\spl{S}$. 
\begin{figure}[!htbp]
	\centerline{\scalebox{0.6}{\input{figs/seqspec.pdf_t}}}
	\caption{An illustration of a sequential execution $E^{\spl{S}}$.}
	\label{fig:seqhist}
\end{figure}

\vspace{1mm}
\noindent
\textbf{Sequential Specification.} 
We next discuss about \emph{\sspec} \cite{HerlWing:1990:TPLS} of shared memory objects. The \sspec of a shared memory object $o$ is defined as the set of (all possible) sequential histories involving the \mth{s} of $o$. Since all the histories in the \sspec of $o$ are sequential, this set captures the behavior of $o$ under sequential execution which is believed to be correct. A sequential history $\spl{S}$  is said to be \emph{\legal} if for every shared memory object $o$ whose \mth is invoked in $\spl{S}$, $\spl{S}|o$ is in the \sspec of $o$. 

\vspace{1mm}
\noindent
\textbf{Safety:} A \emph{safety} property is defined over histories (and the corresponding executions) of shared objects and generally states which executions of the shared objects are acceptable to any application. The safety property that we consider is \emph{linearizability} \cite{HerlWing:1990:TPLS}. A history $H$ is said to be linearizable if (1) there exists a completion $\overline{H}$ of $H$ in which some pending $\inv$ events are completed with a matching response and some other pending $\inv$ events are discarded; (2) there exists a sequential history $\mathbb{S}$ such that $\mathbb{S}$ is equivalent to $\overline{H}$; (3) $\mathbb{S}$ is legal.
Another way to say that history $H$ is linearizable if it is possible to assign an atomic event as a \emph{linearization point} (\emph{\lp}) inside the execution interval of each \mth such that the result of each of these \mth{s} is the same as it would be in a sequential history $\mathbb{S}$ in which the \mth{s} are ordered by their $\lp${s} \cite{MauriceNir}. In this document, we show how to prove the correctness of \lp{s} of the various \mth{s} of a data-structure. 
}

\ignore{
A concurrent object is linearizable if all the histories generated by it are linearizable.  We prove the linearizability of a concurrent history by defining a $\lp${s} for each method. The $\lp$ of a \mth implies that the method appears to take effect instantly at its $\lp$. 
\textbf{Progress:} The \emph{progress} properties specifies when a thread invoking \mth{s} on shared objects completes in presence of other concurrent threads. Some progress conditions used in this paper are mentioned here which are based on the definitions in Herlihy \& Shavit \cite{opodis_Herlihy}. The progress condition of a method in concurrent object is defined as: (1) Blocking: In this, an unexpected delay by any thread (say, one holding a lock) can prevent other threads from making progress. (2) Deadlock-Free: This is a \textbf{blocking} condition which ensures that \textbf{some} thread (among other threads in the system) waiting to get a response to a \mth invocation will eventually receive it. (3) Wait-Free: This is a \textbf{non-blocking} condition which ensures that \textbf{every} thread trying to get a response to a \mth, eventually receives it.
It can be seen that wait-free \mth{s} are desirable since they can complete regardless of the execution of other concurrent \mth{s}. On the other hand, deadlock-free \mth{s} are system (or underlying scheduler) dependent progress condition since they involve blocking. It ensures that among multiple threads in the system, at least one of them will make progress. 
}

\section{Generic Proof Technique}
\label{sec:gen-proof}

In this section, we develop a generic framework for proving the correctness of a \cds based on \lp events of the methods. Our technique of proving is based on hand-crafting and is not automated. We assume that the developer of the \cds has also identified the \lp{s} of the \mth{s}. We assume that the \lp{s} satisfy a few properties that we outline in the course of this section. 

In \secref{ds-proofs}, we illustrate this technique by showing the correctness at a high level of two structures (1) lazy-list based concurrent set implementation \cite{Heller-PPL2007} denoted as \lazy in \subsecref{con-lazy-list} and hand-over-hand locking based concurrent set implementation denoted as \hoh in \subsecref{con-hoh}.

\subsection{Linearization Points Details}
\label{subsec:lps}

Intuitively, \lp is an (atomic) event in the execution interval of each \mth such that the execution of the entire \mth seems to have taken place in the instant of that event. As discussed in \secref{model}, the \lp of each \mth is such that the result of execution of each of these \mth{s} is the same as it would be in a sequential history $\mathbb{S}$ in which the \mth{s} are ordered by their $\lp${s} \cite{MauriceNir}. 

Given, the set of \lp{s} of all the \mth{s} of a concurrent data-structure, we show how the correctness of these \lp{s} can be verified. We show this by proving the correctness of the \cds assuming that it is \lble and the \lp{s} are chosen correctly in the first place. 

Consider a \mth $m_i(inv\text{-}params\downarrow, rsp\text{-}params\uparrow)$ of a \cds $d$. Then the precise \lp of $m_i$ depends on $rsp\text{-}params\uparrow$. For instance in the \lazy \cite{Heller-PPL2007}, the \lp of $contains(k \downarrow, true\uparrow)$ \mth is different from $contains(k\downarrow, false\uparrow)$. Furthermore, the \lp of a \mth also depends on the execution. For instance, considering the contains \mth of the \lazy again, the \lp of $contains(k\downarrow, false\uparrow)$ depends on whether there is an $add(k\downarrow, true\uparrow)$ \mth concurrently executing with it or not. The details of the \lp{s} of the lazy-list are described in the original paper by Heller et. al \cite{Heller-PPL2007} and also in \subsecref{con-lazy-list}. 


We denote the LP event of $m_i$ in a history $H$ as $E^H.m_i(inv\text{-}params\downarrow, rsp\text{-}params\uparrow).LP$ or $E^H.m_i.LP$ (depending on the context). The global state in the execution $E^H$ immediately before and after $m_i.LP$ is denoted as $\preel{i}$  and $\postel{i}$ respectively. 

\ignore {

\begin{figure}[H]
\centerline{\scalebox{0.6}{\input{figs/ConExe.pdf_t}}}
\caption{Figure (a) illustrates an example of a concurrent execution $E^H$. Then, $m_i.LP$ is the $\lp$ event of the method $m_i$. The global state immediately after this event is represented as Post-state of ($E^H.m_i.LP$). Figure (b) represents sequential execution $E^S$ corresponding to (a) with post-state of method $m_i$ as the state after its $resp$ event.}
\label{fig:conc-seq1}
\end{figure}

}

\subsection{Abstract Data-Structure \& LP Assumptions}
\label{subsec:abds}

To prove correctness of a \cds $d$, we associate with it an \emph{abstract data-structure} or \emph{\abds}. The \abds captures the behavior of \cds if it had executed sequentially. Since sequential executions are assumed to be correct, it is assumed that \abds is correct. In fact, the \sspec of $d$ can be defined using \abds since in any global state the internal state of \abds is the result of sequential execution. Thus, we can say that \cds $d$ refines \abds \cite{He+:DRR:ESOP:1986}.

The exact definition of \abds depends on the actual \cds being implemented. In the case of \lazy, \abds is the set of unmarked nodes reachable from the head while the \cds is the set of all the nodes in the system. Vafeiadis et. al \cite{Vafeiadis-ppopp2006} while proving the correctness of the \lazy refer to \abds as \emph{abstract set} or \abs. In the case of \hoh, \abds is the set nodes reachable from the head while the \cds is the set of all nodes similar to \lazy. Normally the \cds maintains more information (such as sentinel nodes) than \abds to implement the desired behavior. For a given global state $S$, we use the notation $S.\abds$ and $S.\cds$ to refer to the contents of these structures in $S$. 


Now we state a few assumptions about the \cds and its \lp{s} that we require for our proof technique to work. 

\begin{assumption}
\label{asm:determ-cds}
The design of \cds and its \abds is deterministic and its \lp{s} are known.
\end{assumption}

\begin{assumption}
\label{asm:mth-total}
In any sequential execution, any method of the \cds can be invoked in any global state and yet get a response.
\end{assumption}

Intuitively, \asmref{mth-total} states that if threads execute the \mth{s} of the \cds sequentially then every \mth invocation will have a matching response. Such \mth{s} are called as \emph{total} \cite[Chap 10]{MauriceNir}. 

\begin{assumption}
\label{asm:seq-legal}
Every sequential history $\spl{S}$ generated by the \cds is legal.
\end{assumption}

\asmref{seq-legal} says that sequential execution of the \cds is correct and does not result in any errors. We next make the following assumptions based on the \lp{s}. This fundamentally comes from the definition of sequential execution.

\begin{assumption}
\label{asm:lp-evt}

Consider a \mth $m_i(inv\text{-}params\downarrow, rsp\text{-}params\uparrow)$ of the \cds in a concurrent execution $E^H$. Then $m_i$ has a unique \lp which is an atomic event within the \inv and \rsp events of $m_i$ in $E^H$. The \lp event can be identified based on the $inv\text{-}params\downarrow$, $rsp\text{-}params\uparrow$ and the execution $E^H$. 

\end{assumption}

\ignore{
\begin{assumption}
\label{asm:lp-rsp}
Consider a \mth $m_i(inv\text{-}params\uparrow, rsp\text{-}params\downarrow)$ of the \cds in a concurrent execution $E^H$. The exact \lp of $m_i$ is decided by the $rsp\text{-}params\downarrow$ of the \rsp event. 
\end{assumption}
}

\begin{assumption}
\label{asm:change-abds}
Consider an execution $E^H$ of a \cds $d$. Then only the $\lp$ events of the \mth{s} can change the contents \abds of the given \cds $d$. 
\end{assumption}

Assumptions \ref{asm:lp-evt} \& \ref{asm:change-abds} when combined imply that there is only one event in each \mth that can change the \abds. As per \asmref{change-abds}, only the \lp{s} can change the contents of \abds. But this does not imply that all the \lp{s} change the \abds. It implies that if an event changes \abds then it must be a \lp event. For instance in the case of \lazy, the \lp{s} of $add(k, false), remove(k, false)$ and the \lp{s} of the contains \mth{s} do not change the \abds. 

\ignore{
We denote the \mth{s} that change the \abds as \upm. In case of \lazy, they are $add(k, true)$ and $remove(k, true)$. We next make an assumption about \upm{s} in which we compare the execution of an \upm in a concurrent and a sequential history. 

\begin{assumption}
\label{asm:resp-lp}
Consider a concurrent history $H$ and a sequential history $\spl{S}$. Let $m_x, m_y$ be \upm{s} in $H$ and $\spl{S}$ respectively. If the \inv and the \rsp events of both $m_x, m_y$ are the same then the \lp{s} of both these \mth{s} are the same in $H$ \& $\spl{S}$. Formally, 
$\langle \forall m_x \in \upms{E^{H}}, \forall m_y \in \upms{E^{\spl{S}}}: (\inveds{x} = \invmds{y}) \wedge (\reteds{x} = \retmds{y}) \Longrightarrow (\lpeds{x} = \lpmds{y}) \rangle$.
\end{assumption}
This assumption implies that an \upm in a concurrent history has a fixed \lp and does not depend on other \mth{s} concurrently executing. This is unlike the $contains(k, false)$ \mth of \lazy which has its \lp dependent on concurrently executing successful $add(k, true)$ which we explained in \subsecref{lps}. 
}

We believe that the assumptions made by us are generic and are satisfied by many of the commonly used \cds{s} such as Lock-free Linked based Sets \cite{Valoispodc1995}, \hoh \cite{Bayerai1977, MauriceNir} , \lazy \cite{Heller-PPL2007, MauriceNir}, Skiplists \cite{Levopodis2006} etc. In fact, these assumptions are similar in spirit to the definition of Valid LP by Zhu et al \cite{Zhu+:Poling:CAV:2015}. 

It can be seen that the Assumptions \ref{asm:lp-evt} \& \ref{asm:change-abds} characterize the \lp events. Any event that does not satisfy these assumptions is most likely not a \lp (please refer to the discussion section \secref{conc} more on this). 

\subsection{Constructing Sequential History}
\label{subsec:csh}

To prove \lbty of a \cds $d$ which satisfies the Assumptions \ref{asm:mth-total}, \ref{asm:seq-legal}, \ref{asm:lp-evt}, \ref{asm:change-abds} we have to show that every history generated by $d$ is \lble. To show this, we consider an arbitrary history $H$ generated by $d$. First we complete $H$, to form $\overline{H}$ if $H$ is incomplete. We then construct a sequential history denoted as $\sh{H}$ (constructed sequential history). $H$ is \lble if (1) $\sh{H}$ is equivalent to a completion of $H$; (2) $\sh{H}$ respects the real-time order of $H$ and (3) $\sh{H}$ is \legal. 

\noindent We now show how to construct $\overline{H}$ \& $\sh{H}$. We then analyze some properties of $\sh{H}$. 


\vspace{1mm}
\noindent
\textbf{Completion of ${H}$.} Suppose $H$ is not complete. This implies $H$ contains some incomplete \mth{s}. Note that since these \mth{s} are incomplete, they could have executed multiple possible \lp events. Based on these \lp events, we must complete them by adding appropriate \rsp event or ignore them. We construct the completion $\overline{H}$ and $E^{\overline{H}}$ as follows: 

\begin{enumerate}

\item Among all the incomplete \mth{s} of $E^H$ we ignore those \mth{s}, say $m_i$, such that: (a) $m_i$ did not execute a single \lp event in $E^H$; (b) the \lp event executed by $m_i$ did not change the \abds. 


\item The remaining incomplete \mth{s} must have executed an \lp event in $E^H$ which changed the \abds. Note from Assumptions \ref{asm:lp-evt} \& \ref{asm:change-abds}, we get that each \mth has only one event which can change the \abds and that event is the \lp event. We build an ordered set consisting of all these incomplete \mth{s} which is denoted as \emph{\pset}. The \mth{s} in \pset are ordered by their \lp{s}. 

\item To build $\overline{H}$, for each incomplete \mth $m_i$ in \pset considered in order, we append the appropriate \rsp event to $H$ based on the \lp event of $m_i$ executed. Since the \mth{s} in \pset are ordered by their \lp events, the appended \rsp events are also ordered by their \lp events. Here, we assumed that once a \mth executes a \lp event that changes the \abds, its \rsp event can be determined. 

\item To construct $E^{\overline{H}}$, for each incomplete \mth $m_i$ in \pset considered in order, we sequentially append all the remaining events of $m_i$ (after its \lp) to $E^H$. All the appended events are ordered by the \lp{s} of their respective \mth{s}. 
\end{enumerate}

\noindent From this construction, one can see that if $\overline{H}$ is \lble then $H$ is also \lble. Formally, $\langle (\overline{H} \text{ is \lble}) \implies (H \text{ is \lble}) \rangle$. 

For simplicity of presentation, we assume that all the concurrent histories \& executions that we consider in the rest of this document are complete unless stated otherwise. Given any history that is incomplete, we can complete it by the transformation mentioned here. Next, we show how to construct a $\sh{H}$ for a complete history $H$. 

\vspace{1mm}
\noindent
\textbf{Construction of $\sh{H}$.} Given a complete history $H$ consisting of \mth{} \inv \& rsp events of a \cds $d$, we construct $\sh{H}$ as follows: We have a single (hypothetical) thread invoking each \mth of $H$ (with the same parameters) on $d$ in the order of their \lp events. Only after getting the response for the currently invoked \mth, the thread invokes the next \mth. From \asmref{mth-total}, which says that the \mth{s} are total, we get that for every \mth invocation $d$ will issue a response. 

Thus we can see that the output of these \mth invocations is the sequential history $\sh{H}$. From \asmref{seq-legal}, we get that $\sh{H}$ is \legal. The histories $H$ and $\sh{H}$ have the same \inv events for all the \mth{s}. But, the \rsp events could possibly be different. Hence, the two histories may not be equivalent to each other unless we prove otherwise. 

In the sequential history $\sh{H}$ all the \mth{s} are totally ordered. So we can enumerate all its \mth{s} as: $m_{1}(inv\text{-}params, 
rsp\text{-}params)~ m_{2}(inv\text{-}params, rsp\text{-}params)~\dots~ 
m_{n}(inv\text{-}params, \\
rsp\text{-}params)$. On the other hand, the \mth{s} in a concurrent history $H$ are not ordered. From our model, we have that all the events of the execution $E^H$ are ordered. In \asmref{lp-evt}, we have assumed that each complete \mth has a unique \lp event which is atomic. All the \mth{s} of $H$ and $E^H$ are complete. Hence, we can order the \lp{s} of all the \mth{s} in $E^H$. Based on \lp ordering, we can enumerate the corresponding \mth{s} of the concurrent history $H$ as $m_{1} (inv\text{-}params, rsp\text{-}params),~ m_{2} (inv\text{-}params, rsp\text{-}params),\\~\dots~
m_{n} (inv\text{-}params, rsp\text{-}params)$. Note that this enumeration has nothing to do with the ordering of the \inv and \rsp events of the \mth{s} in $H$. 

Thus from the construction of $\sh{H}$, we get that for any \mth $m_i$, $H.\inv(m_{i}(inv\text{-}params)) = \sh{H}.\inv(m_{i}(inv\text{-}params))$ but the same need not be true for the \rsp events. 

\ignore{
In the sequential history $\sh{H}$ all the \mth{s} are totally ordered. So we represent it as: $m_{s1}(inv\text{-}params_{s1}, rsp\text{-}params_{s1})~ m_{s2}(inv\text{-}params_{s2}, rsp\text{-}params_{s2})~ ...~ m_{sn}(inv\text{-}params_{sn}, 
rsp\text{-}params_{sn})$ where $s1...sn$ are indices. On the other hand, the \mth{s} in a concurrent history $H$ are not ordered. From our model, we have that all the events of the execution $E^H$ are ordered. In \asmref{lp-evt}, we have assumed that each complete \mth has a unique \lp event which is atomic. All the \mth{s} of $H$ and $E^H$ are complete. Hence, we can order the \lp{s} of all the \mth{s} in $E^H$. Based on this ordering, we can enumerate all the \mth{s} of the concurrent history $H$ as $m_{h1}(inv\text{-}params_{h1}, rsp\text{-}params_{s1}),~ m_{h2}(inv\text{-}params_{h2}, rsp\text{-}params_{s2}),~ ...~ m_{hn}(inv\text{-}params_{hn}, 
rsp\text{-}params_{hn})$. Note that this enumeration has nothing to do with the ordering of the \inv and \rsp events of the \mth{s}. 
}

For showing $H$ to be \lble, we further need to show $\sh{H}$ is equivalent to $H$ and respects the real-time order $H$. Now, suppose $\sh{H}$ is equivalent to $H$. Then from the construction of $\sh{H}$, it can be seen that $\sh{H}$ satisfies the real-time order of $H$. The following lemma proves it. 

\begin{lemma}
\label{lem:sh-rt}
Consider a history $H$ be a history generated by a \cds $d$. Let $\sh{H}$ be the constructed sequential history. If $H$ is equivalent to $\sh{H}$ then $\sh{H}$ respects the real-time order of $H$. Formally, $\langle \forall H: (H \approx \sh{H}) \implies (\prec^{rt}_H \subseteq \prec^{rt}_{\sh{H}}) \rangle$. 
\end{lemma}

\begin{proof}
This lemma follows from the construction of $\sh{H}$. Here we are given that for every \mth $m_i$, $H.m_i.inv = \sh{H}.m_i.inv$ and $H.m_i.rsp = \sh{H}.m_i.rsp$. 

Now suppose two \mth{s}, $m_i, m_j$ are ordered by real-time. This implies that $m_i.rsp <_H m_j.inv$. Hence, we get that $m_i.inv <_H m_i.rsp <_H m_j.inv$ which means that $m_i$ is invoked before $m_j$ in $H$. Thus, from the construction of $\sh{H}$, we get that $m_i$ is invoked before $m_j$ in $\sh{H}$ as well. Since $\sh{H}$ is sequential, we get that $m_i.rsp <_{\sh{H}} m_j.inv$. Thus $\sh{H}$ respects the real-time order of $H$. 
\end{proof}
Now it remains to prove that $H$ is equivalent to $\sh{H}$ for showing \lbty of $H$. But this proof depends on the properties of the \cds $d$ being implemented and is specific to $d$. Now we give a generic outline for proving the equivalence between $H$ and $\sh{H}$ for any \cds. As mentioned earlier, later in \secref{ds-proofs}, we illustrate this technique by showing at a high level the correctness of  \lazy \& \hoh. 

\subsection{Details of the Generic Proof Technique} 
\label{subsec:generic}
As discussed above, to prove the correctness of a concurrent (\& complete) history $H$ representing an execution of a \cds $d$, it is sufficient to show that $H$ is equivalent to $\sh{H}$. To show this, we have developed a generic proof technique. 

It can be obviously seen that to prove the correctness, this proof depends on the properties of the \cds $d$ being considered. To this end, we have identified a \cdse which captures the properties required of the \cds $d$. Proving this definition for each \cds would imply equivalence of $H$ between $\sh{H}$ and hence \lbty of the \cds. 

In the following lemmas, we assume that all the histories and execution considered here are generated from the \cds $d$. The \cds $d$ satisfies the Assumptions \ref{asm:mth-total}, \ref{asm:seq-legal}, \ref{asm:lp-evt}, \ref{asm:change-abds}. Since we are only considering \cds $d$, we refer to its abstract data-structure as $\abds$ and refer to its state in a global state $S$ as $S.\abds$. 


In the following lemmas, as described in \subsecref{csh}, we enumerate all the \mth{s} of a sequential history $\spl{S}$ as: $m_1, m_2 ... m_n$. We enumerate all the \mth{s} of the concurrent history $H$ as $m_1, m_2 ... m_n$ based on the order of their \lp{s}.


\begin{lemma}
\label{lem:seqe-post-pre}
The \abds of $d$ in the global state after the \rsp event of a method $m_x$ is the same as the \abds before the \inv event of the consecutive method $m_{x+1}$ in an execution $E^{\spl{S}}$ of a sequential history $\spl{S}$. Formally, $\langle \forall m_x \in \mths{E^{\spl{S}}}: \postmds{x} = \premds{x+1}  \rangle$.
\end{lemma}

\begin{proof}
From the definition of Sequential Execution. 
\end{proof}

\begin{lemma}
\label{lem:conce-post-pre}
Consider a concurrent execution $E^H$ of the \mth{s} of $d$. Then, the contents of $\abds$ in the post-state of \lp of $m_x$ is the same as the $\abds$ in pre-state of the next \lp belonging to $m_{x+1}$. Formally, $ \langle \forall m_x\in\mths{E^{H}}: \posteds{x}  = \preeds{x+1} \rangle$.
\end{lemma}

\begin{proof}
\cmnt{Let us prove by contradiction.\\
Assume that there exists a method $m_k$ which modifies the $\abds$ between method $m_i$ and $m_{i+1}$.\\}
From the assumption \ref{asm:change-abds}, we know that any event between the post-state of $m_i.LP$ and the pre-state of $m_{i+1}.LP$ will not change the $\abds$. Hence we get this lemma. 
\end{proof}

\cmnt {
\begin{assumption}
\label{asm:resp-lp1}
If for any method, the invocation and method response are the same in both concurrent and sequential history then the method LP is the same in both concurrent \& sequential histories as well. Formally, 
$\langle \forall m_x \in E^{H}, m_y \in E^{\spl{S}}: (\inveds{x} = \invmds{y}) \wedge (\reteds{x} = \retmds{y}) \Longrightarrow (\lpeds{x} = \lpmds{y}) \rangle$.
\end{assumption}

\asmref{resp-lp} holds true for the update methods of all variants of the concurrent set based list implementation. However, it does not hold for the wait-free contains (read-only) method.
}

Now, we describe a \cdse. This definition can be considered to be generic template. Based on the \cds involved, this has to be appropriately proved. 
\begin{tcolorbox}
\begin{definition}{\textbf{\cdse:}}
\label{def:pre-resp}

Consider a concurrent history $H$ and a sequential history $\spl{S}$. Let $m_x, m_y$ be \mth{s} in $H$ and $\spl{S}$ respectively. Suppose the following are true (1) The \abds in the pre-state of $m_x$'s \lp in $H$ is the same as the \abds in the pre-state of $m_y$ in $\spl{S}$;  (2) The \inv events of $m_x$ and $m_y$ are the same. Then (1) the $\rsp$ event of $m_x$ in $H$ must be same as $\rsp$ event of $m_y$ in $\spl{S}$; (2) The \abds in the post-state of $m_x$'s \lp in $H$ must be the same as the \abds in the post-state of $m_y$ in $\spl{S}$. Formally, $\langle \forall m_x \in \mths{E^{H}}, \forall m_y \in \mths{E^{\spl{S}}}: (\preeds{x} = \premds{y}) \wedge (\inveds{x} = \invmds{y}) \Longrightarrow (\posteds{x} = \postmds{y}) \wedge (\reteds{x} = \retmds{y}) \rangle$.
\end{definition}
\end{tcolorbox}
\begin{figure}[H]
\captionsetup{font=scriptsize}
	\centerline{\scalebox{0.5}{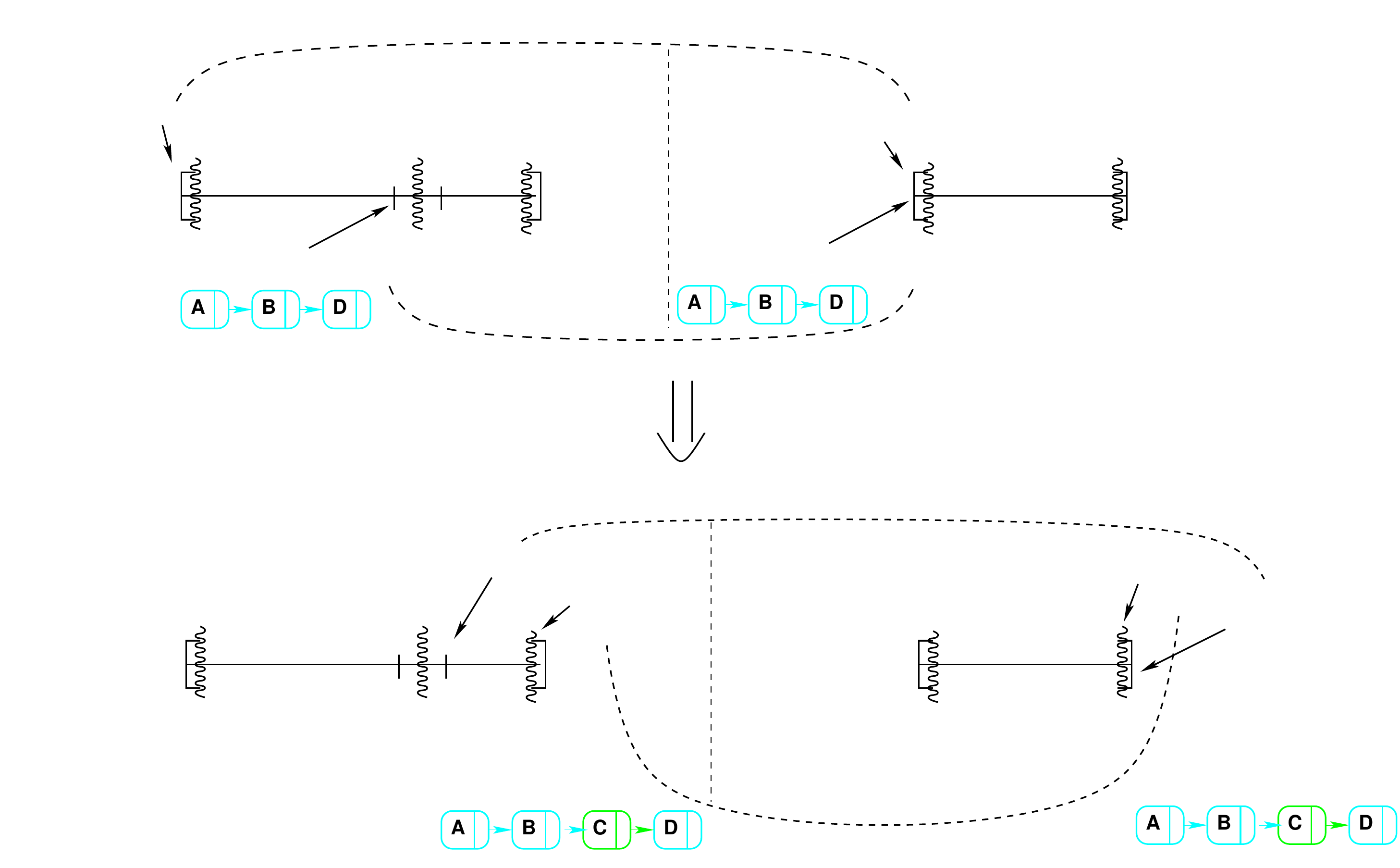}}
	\caption{The pictorial representation of the \cdse over the lazy list. Assume \mth $add(C)$ executes over the initial list $A, B, D$. Figure (a) \& (b) represent the same  \inv and pre-state for $add(C)$ in concurrent and sequential execution respectively. Then for $add(C)$ execution to be correct its respective post-state and \rsp should be same in concurrent and sequential executions as depicted in Figure (c) \& (d). Note, wlog $add(C)$ is $m_x$ and $m_y$ in $E^H$ \& $E^S$ respectively. }
	\label{fig:aadda3}
\end{figure}

Readers familiar with the work of Zhu et. al \cite{Zhu+:Poling:CAV:2015} can see that \cdse is similar to Theorem 1 on showing \lbty of \cds $d$. In \subsecref{con-lazy-list} and in \subsecref{con-hoh} we prove \cdse specifically for \lazy and \hoh.

Next, in the following lemmas we consider the \mth{s} of $H$ and $\sh{H}$. As observed in \subsecref{csh}, for any \mth $m_x$ in $\sh{H}$ there is a corresponding \mth $m_x$ in $H$ having the same \inv event, i.e., $H.m_x.inv = \sh{H}.m_x.inv$. We use this observation in the following lemma.


\begin{lemma}
\label{lem:conce-seqe-pre}
For any \mth $m_x$ in $H, \sh{H}$ the \abds in the pre-state of the \lp of $m_x$ in $H$ is the same as the \abds in the pre-state of $m_x$ in $\sh{H}$. Formally, $\langle \forall m_x \in \mths{E^H}, \mths{E^{\sh{H}}}: \preeds{x} = \prespmds{x} \rangle$.
\end{lemma}
\begin{proof}
	We prove by Induction on events which are the linearization points of the methods,\\[0.1cm]
	\textbf{Base Step:} Before the $1^{st}$ $\lp$ event, the initial $\abds$ remains same because all the events in the concurrent execution before the $1^{st}$ $\lp$ do not change $\abds$.\\[0.3cm]
	\textbf{Induction Hypothesis}: Let us assume that for $k$ $\lp$ events, we know that,\\[0.2cm]  
	$\langle \preeds{k} = \prespmds{k} \rangle$. \\[0.3cm]
	\textbf{Induction Step:} We have to prove that: $\preeds{k+1} = \prespmds{k+1}$ holds true.\\[0.3cm]
	We know from Induction Hypothesis that for $k^{th}$ method,\\[0.3cm]
	$\preeds{k} = \prespmds{k}$\\[0.3cm]
	From the construction of $\sh{H}$, we get that $H.m_x.inv = \sh{H}.m_x.inv$. Combining this with \defref{pre-resp} we have,
	\begin{equation} 
	\begin{split}
	\label{lab:eq:conce-seqe-mi}
	(H.m_x.inv = \sh{H}.m_x.inv) \land  (\preeds{k} = \prespmds{k}) \\ \xRightarrow[]{\text{\defref{pre-resp}}}  (\posteds{k} = \postspmds{k})
	\end{split}
	\end{equation} 

	From the \lemref{seqe-post-pre}, we have,
	
	\begin{equation} 
	\label{lab:eq:seqe-post-pre}
	\postspmds{k} \xRightarrow[]{\text{Lemma \ref{lem:seqe-post-pre}}}  \prespmds{k+1} 
	\end{equation} 
	
	From the equation \ref{lab:eq:conce-seqe-mi} we have,
	\begin{equation} 
	\label{lab:eq:seqe-mi-post}
	\posteds{k} =  \postspmds{k}
	\end{equation} 
	
	By combining the equation \ref{lab:eq:seqe-mi-post} and \ref{lab:eq:seqe-post-pre} we have,
	
	\begin{equation} 
	\label{lab:eq:seqe-mi-pre}
	\posteds{k} =  \prespmds{k+1}
	\end{equation} 
	
	And from the Lemma \ref{lem:conce-post-pre} we have, 
	\begin{equation} 
	\label{lab:eq:conce-pre}
	\posteds{k} \xRightarrow[]{\text{Lemma \ref{lem:conce-post-pre}}}   \preeds{k+1}
	\end{equation}
	
	So, by combining equations \ref{lab:eq:conce-pre} and \ref{lab:eq:seqe-mi-pre} we get,
	
	\begin{equation} 
	\label{lab:eq:conce-seuq-pre}
	\preeds{k+1} =  \prespmds{k+1}
	\end{equation}
	
	This holds for all $m_i$ in $E^H$. Hence the lemma. 	
\end{proof}

\begin{figure}[H]
	\centerline{\scalebox{0.5}{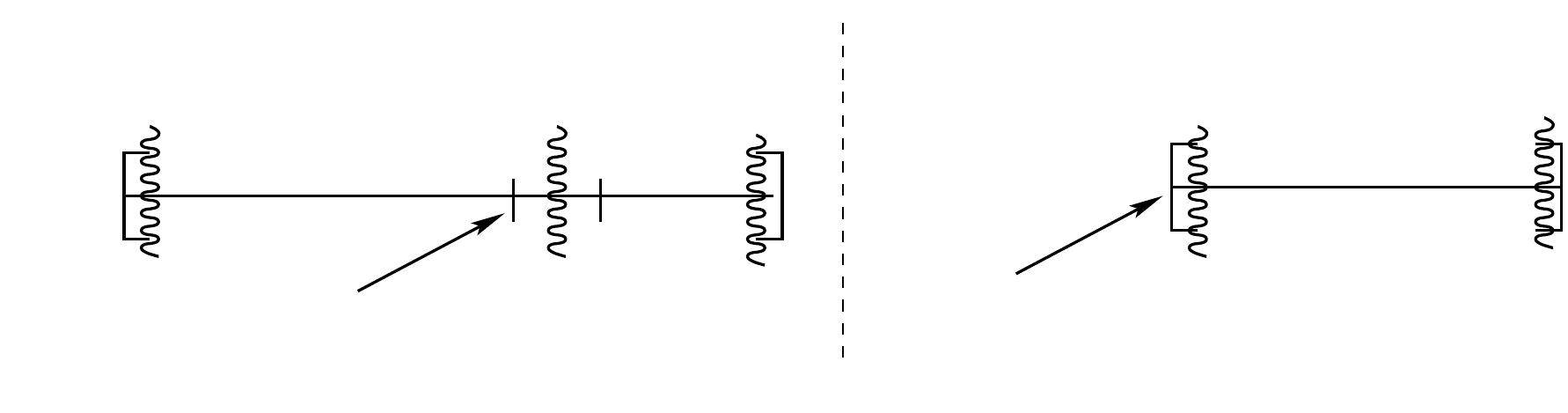}}
	\caption{Pictorial representation of pre-state in $E^H$ and $E^{CS(H)}$}
	\label{fig:aadda1}
\end{figure}

\begin{lemma}
\label{lem:ret-lin}
The return values for all the \mth{s} in $H$ \& $\sh{H}$ are the same. Formally, $\langle \forall m_x \in \mths{E^H}, \mths{E^{\sh{H}}}: \reteds{x} = \retspmds{x} \rangle$.
\end{lemma}

\begin{proof}
From the construction of $\sh{H}$, we get that for any \mth $m_x$ in $H$, $\sh{H}$ the invocation parameters are the same. From \lemref{conce-seqe-pre}, we get that the pre-states of all these \mth{s} are the same. Combining this result with \defref{pre-resp}, we get that the responses parameters for all these \mth{s} are also the same. 
\end{proof}

\begin{theorem}
\label{thm:hist-lin}
All histories ${H}$ generated by the \cds $d$ are \lble.
\end{theorem}
\vspace{-3mm}
\begin{proof}
From \lemref{ret-lin}, we get that for all the \mth{s} $m_x$, the responses in $H$ and $\sh{H}$ are the same. This implies that $H$ and $\sh{H}$ are equivalent to each other. Combining this with \lemref{sh-rt}, we get that $\sh{H}$ respects the real-time order of $H$. We had already observed from \asmref{seq-legal} that $\sh{H}$ is \legal. Hence $H$ is $\lble$. 
\end{proof}

\vspace{1mm}
\noindent
\textbf{Analysis of the Proof Technique:} \thmref{hist-lin} shows that proving \cdse (\defref{pre-resp}) implies that the \cds $d$ under consideration is linearizable. \cdse states that the contents of the \abds in the pre-state of the \lp event of a \mth $m_x$ should be the same as the result of sequential execution of the \mth{s} of $d$. Thus if the contents of \abds in the pre-state of the \lp event (satisfying the assumptions \ref{asm:lp-evt} \& \ref{asm:change-abds}) cannot be produced by some sequential execution of the \mth{s} of $d$ then it is most likely the case that either the \lp or the algorithm of the $d$ is incorrect. 

Further \cdse requires that after the execution of the \lp, the \abds in the post-state must again be same as the sequential execution of some \mth{s} of $d$ with the final \mth being $m_x$. If this is not the case, then it implies that some other events of the \mth are also modifying the \abds and hence indicating some error in the analysis. 

Extending this thought, we also believe that the intuition gained in proving \cdse for $d$ might give the programmers new insights in the working of the \cds which can result in designing new variants of it having some desirable properties. 

\section{Data-Structure Specific Proofs}
\label{sec:ds-proofs}

In this section, we prove the proposed \cdse for \lazy \& \hoh as described in \defref{pre-resp} and therefore we show that they are linearizable. In the Section \ref{subsec:con-lazy-list} and \ref{subsec:con-hoh} we show this for the state of art \lazy (Algorithm \ref{alg:validate}-\ref{alg:contains}) \cite{Heller-PPL2007, MauriceNir} \& \hoh(Algorithm \ref{alg:hlocate}-\ref{alg:hremove}) \cite{Bayerai1977, MauriceNir}, respectively. This we achieve by showing that both the \cds satisfy the requirements of the \cdse. It is renamed as \llse{} or \hohse{} according to the \cds under consideration.

\subsection{Lazy List}
\label{subsec:con-lazy-list}
In this section, we define the lazy list data structure. It is implemented as a set of nodes - concurrent set which is dynamically being modified by a fixed set of concurrent threads. In this setting, threads may perform insertion or deletion of nodes to the set. We describe lazy list based set algorithm based on Heller et al.\cite{Heller-PPL2007}. 
This is a linked list of nodes of type \emph{Node} and it has four fields. The \emph{val} field is a unique identifier of the node. The nodes are sorted in order of the $val$ field. The $marked$ field is of type boolean which indicates whether that \textit{\node} is logically present in the list or not. The $next$ field is a reference to the next \textit{\node} in the list. The $lock$ field is for ensuring access to a shared \textit{\node} which happens in a mutually exclusive manner. We say a thread  acquires a lock and releases the lock  when it executes a \textit{lock.acquire()} and \textit{lock.release()} method call respectively. We assume the $next$ and $marked$ of the $\node$ are atomic. This ensures that operations on these variables happen atomically. In the context of a particular application, the node structure can be easily modified to carry useful data (like weights etc).

\begin{tcolorbox}
\begin{verbatim}
class Node{
    int val;
    Node next;
    boolean marked;
    Lock lock;
    Node(int key){  
        val = key;
        marked = false;
        next = null;
        lock = new Lock();
    }
};
\end{verbatim}
\end{tcolorbox}
\ignore{
\noindent Given a global state $S$, we define a few structures and notations as follows: 
\begin{enumerate}
	\item We denote node, say $n$, as a $\node$ class object. 
    \item $\nodes{S}$ as a set of $\node$ class objects that have been created so far in $S$. Its structure is defined above. Each $\node$ object $n$ in $\nodes{S}$ is initialized with $key$, $next$ to $null$, $marked$ field initially set to $false$.
    
    \item $S.\head$ is a $\node$ class object (called sentinel head node), which is initialized with a val $-\infty$. This sentinel node is never removed from the list. Similarly, $S.\tail$ is a $\node$ class object (called sentinel tail  node), which is initialized with a val $+\infty$. This sentinel node is never removed from the list.
   
    \item The contents \& the status of a $\node$ $n$ keeps changing with different global states. For a global state $S$, we denote $S.n$ as its current state and the individual fields of $n$ in $S$ as $S.n.val, S.n.next, ...$ etc.  
    
\end{enumerate}
  }
 \noindent 
 \ignore{Having defined a few notions on $S$, we now define the notion of an abstract set, $\abs$ for a global state $S$ which we will use for guiding us in correctness of the \mth{s}.} 
 \ignore{
 \begin{definition}
\label{def:abs}
	$S.\abs \equiv \{n | (n \in \nodes{S}) \land (S.\head \rightarrow^* S.n) \land (\neg S.n.marked)\}$.
\end{definition}
\noindent This definition of $\abs$ captures the set of all nodes of $\abs$ for the global state $S$. It consists of all the $\node{s}$ that are reachable from $\head$ of the list and are not marked for deletion. 
}
\subsubsection{\textbf{Methods Exported \& Sequential Specification}}
\label{sec:con-lazylist:methods}

\noindent \\ In this section, we describe the \mth{s} exported by the lazy list data structure.
\noindent 
\begin{enumerate}
\item The $\add(n)$ method adds a node $n$ to the list, returns $true$ if the node is not already present in the list else returns $false$. 

\item The $\rem(n)$ method removes a node $n$ from the list, if it is present and returns $true$. If the node is not present, it returns $false$. 
\item The $\con(n)$ returns $true$, if the list contains the node $n$; otherwise returns $false$.
\end{enumerate}

\begin{table}[!htbp]
 \centering
 \scriptsize
 \captionsetup{font=scriptsize}
  \captionof{table}{Sequential Specification of the Lazy list } \label{tab:seq-spe} 
\begin{tabular}{ |m{1.5cm}|m{1cm}|m{2.5cm}|m{2.5cm}| } 

\hline
  \textbf{Method} & \textbf{Return Value} & \textbf{Pre-state}($S$: Pre-State of the method)  & \textbf{Post-state}( $S'$: Post-State of the method)  \\ 
  \hline
	$\add(n)$  & $true$ &$ S: \langle n \notin S.\abs \rangle $ &$ S':\langle n \in S'.\abs \rangle$	 \\
\hline 
	$\add (n)$  & $false$ &$ S: \langle n \in S.\abs \rangle $ & $ S':\langle n \in S'.\abs \rangle$	
	\\
\hline 
	$\rem(n)$  & $true$ &$ S: \langle n \in S.\abs \rangle $ & $ S':\langle n \notin S'.\abs \rangle$	
	\\
	\hline 
$\rem(n)$  & $false$ &$ S: \langle n \notin S.\abs \rangle $ & $ S':\langle n \notin S'.\abs \rangle$
\\
\hline 
$\con(n)$  & $true$ &$ S: \langle n \in S.\abs \rangle$ &$ S':\langle n \in S'.\abs \rangle$	\\
	\hline 
	$\con(n)$  & $false$ &$ S: \langle n \notin S.\abs \rangle$ & $ S':\langle n \notin S'.\abs \rangle$	\\
	\hline 
\end{tabular}
\end{table}

  \noindent Table \ref{tab:seq-spe} shows the sequential specification of the lazy-list. As the name suggests, it shows the behaviour of the list when all the \mth{s} are invoked sequentially. The \textit{Pre-state} of each method is the shared state before $\inv$ event and the \textit{Post-state} is also the shared state just after the $\rsp$ event of a method (after executing it sequentially), as depicted in the Figure \ref{fig:exec/hist}.
 
 
\subsubsection{\textbf{Working of Lazy List Methods}}
\label{sec:working-con-lazy-methods}
\vspace{1mm}
\noindent \\
In this section, we describe the implementation of the lazy list based set algorithm based on Heller et al.\cite{Heller-PPL2007} and the working of the various \mth{s}.\\[0.2cm]

\noindent \textbf{Notations used in PseudoCode:}\\[0.2cm]
$\downarrow$, $\uparrow$ denote input and output arguments to each method respectively. The shared memory is accessed only by invoking explicit \emph{read()} and \emph{write()} methods. The $flag$ is a local variable which returns the status of each operation. We use nodes $n_1$, $n_2$, $n$ to represent $\node$ references.

 \begin{algorithm}[!htb]
 \captionsetup{font=scriptsize}
	\caption{\valid Method: Takes two nodes, $n_1, n_2$, each of type $\node$ as input and validates for presence of nodes in the list and returns $true$ or $false$}
	\label{alg:validate}
	\begin{algorithmic}[1]
	\scriptsize
		\Procedure{\valid($n_1 \downarrow, n_2\downarrow, flag \uparrow$)}{}
		\If {($read(n_1.marked) = false) \wedge (read(n_2.marked) = false) \wedge (read(n_1.next) = n_2$) )}
		\State {$flag$ $\gets$ $true$;}  
		\Else 
		\State {$flag$ $\gets$ $false$;} 
		\EndIf
		\State {$return$; }
		\EndProcedure	
	    \algstore{valid}
	\end{algorithmic}
  \end{algorithm}
 \begin{algorithm}[!htb]
   \captionsetup{font=scriptsize}
	\caption{\loct Method: Takes $key$ as input and returns the corresponding pair of neighboring $\node$ $\langle n_1, n_2 \rangle$. Initially $n_1$ and $n_2$ are set to $null$.}
	\label{alg:locate}
	\begin{algorithmic}[1]
	\scriptsize
	\algrestore{valid}
	\Procedure{\loct ($key\downarrow, n_1\uparrow, n_2\uparrow$)}{}
	\While{(true)}
	\State {$n_1 \gets read(\head)$;} \label{lin:loc3} 
	\State {$n_2 \gets read(n_1.next)$;} \label{lin:loc4}
		\While{$(read(n_2.val) < key)$} \label{lin:loc5}
		\State {$n_1 \gets n_2$;} \label{lin:loc6}
		\State {$n_2 \gets read(n_2.next)$;} \label{lin:loc7}
		\EndWhile
		\State {$lock.acquire(n_1)$;} \label{lin:loc9}
		\State {$lock.acquire(n_2)$;} \label{lin:loc10}
		\If{($\valid(n_1\downarrow, n_2\downarrow, flag\uparrow$))} \label{lin:loc11}
		\State {return;} \label{lin:loc-ret}
		\Else
		\State {$lock.release(n_1)$;}
		\State {$lock.release(n_2)$;}
		\EndIf
		\EndWhile
		\EndProcedure
		\algstore{locate}
	\end{algorithmic}
  \end{algorithm}

 \begin{algorithm}[!htb]
    \captionsetup{font=scriptsize}
	\caption{\add Method: $key$ gets added to the set if it is not already part of the set. Returns $true$ on successful add and returns $false$ otherwise.	}
	\label{alg:add}
	\begin{algorithmic}[1]
	\scriptsize
	\algrestore{locate}
		\Procedure{\add ($key\downarrow, flag \uparrow$)}{}
		\State {$\loct(key\downarrow, n_1 \uparrow, n_2\uparrow)$;}   \label{lin:add2}
		\If {$(read(n_2.val) \neq key$)} \label{lin:add3}
		\State {$write(n_3, \text{new \node}(key))$;} \label{lin:add4}
		\State {$write(n_3.next, n_2)$;} \label{lin:add5}
		\State {$write(n_1.next, n_3)$;} \label{lin:add6}
		\State {$flag$ $\gets$ $true$;}
		\Else
		\State {$flag$ $\gets$ $false$;}    
		\EndIf
	    \State {$lock.release(n_1)$;} \label{lin:add7}
		\State {$lock.release(n_2)$;} \label{lin:add8}
		\State{ $return$;}
		\EndProcedure
		\algstore{add}
	\end{algorithmic}
  \end{algorithm}
 \begin{algorithm}[!htb]
  \captionsetup{font=scriptsize}
	\caption{Remove Method: $key$ gets removed from the set if it is already part of the set. Returns $true$ on successful remove otherwise returns $false$.}
	\label{alg:remove}
	\begin{algorithmic}[1]
	\scriptsize
	\algrestore{add}
		\Procedure{\rem ($key\downarrow, flag\uparrow$)}{}
		\State {$\loct(key\downarrow, n_1 \uparrow, n_2\uparrow);$}   \label{lin:rem2}   
		\If {$(read(n_2.val) = key)$} \label{lin:rem3} 
		\State {$write(n_2.marked, true)$;} \label{lin:rem4}
		\State {$write(n_1.next, n_2.next)$;} \label{lin:rem5}
		\State {$flag$ $\gets$ $true$;}\label{lin:rem8}
		\Else
	    \State {$flag$ $\gets$ $false$;}
		\EndIf
		\State {$lock.release(n_1)$;} \label{lin:rem6}
		\State {$lock.release(n_2)$;} \label{lin:rem7}
		\State {$return$;}
		\EndProcedure
		\algstore{rem}
	\end{algorithmic}
  \end{algorithm}

 \begin{algorithm}[!htb]
    \captionsetup{font=scriptsize}
		\caption{\con Method: Returns $true$ if $key$ is part of the set and returns $false$ otherwise.}
	\label{alg:contains}
	\begin{algorithmic}[1]
	\scriptsize
	\algrestore{rem}
		\Procedure{\con ($key\downarrow, flag\uparrow$)}{}
		\State {$n \gets read(\head);$}  \label{lin:con2} 
		\While {$(read(n.val) < key)$} \label{lin:con3} 
		\State {$n \gets read(n.next);$} \label{lin:con4} 
		\EndWhile \label{lin:con5}
		\If {$(read(n.val) \neq key) \vee (read(n.marked))$} \label{lin:cont-tf} \label{lin:con6}
    	\State {$flag$ $\gets$ $false$;}
		\Else
		\State {$flag$ $\gets$ $true$;}
		\EndIf
		\State {$return$;}
		\EndProcedure
		\algstore{con}
	\end{algorithmic}
  \end{algorithm}

\noindent
\subsubsection{\textbf{\emph{Working of the methods}}}

\noindent \textbf{\\Working of the \add() method:}
When a thread wants to add a node $n$ to the list, it traverses the list from $\head$ without acquiring any locks until it finds a node with its key greater than or equal to $n$, say $ncurr$ and it's predecessor $\node$, say $npred$. It acquires locks on the nodes $npred$ and $ncurr$ itself. It validates to check if $ncurr$ is reachable from $npred$, and if both the nodes have not been deleted (marked). The algorithm maintains an invariant that all the unmarked nodes are reachable from $\head$. If the validation succeeds, the thread adds the $\node(key)$ between $npred$ and $ncurr$ in the list and returns true after unlocking the nodes. If it fails, the thread starts the traversal again after unlocking the locked nodes. This is described in Algorithm \ref{alg:add}.\\[0.2cm]
\noindent \textbf{Working of the \rem() method:}
\noindent Each $\node$ of list has a boolean $marked$ field. The removal of a $\node$ $n$ happens in two steps: (1) The $\node$ $n$'s marked field is first set to $true$. This is referred to as logical removal. This ensures that if any node is being added or removed concurrently corresponding to that node, then $\add$ method will fail in the validation process after checking the marked field. (2) Then, the pointers are changed so that $n$ is removed from the list. This is referred to as physical deletion which involves changing the pointer of the predecessor of the marked node to its successor so that the deleted node is no longer reachable from the $\head$ in the list. To achieve this, $\rem(n)$ method proceeds similar to the $\add(n)$. The thread iterates through the list until it identifies the node $n$ to be deleted. Then after $n$ and its predecessor have been locked, logical removal occurs by setting the marked field to true. This is described in Algorithm \ref{alg:remove}. \\[0.2cm]
\noindent \textbf{Working of the \con() method:}
\noindent Method $\con(n)$ traverses the list without acquiring any locks. This method returns $true$ if the node it was searching for is present and unmarked in the list, otherwise returns $false$. This is described in Algorithm \ref{alg:contains}.

\subsubsection{\textbf{The Linearization Points of the Lazy list methods}}
\noindent \\ Here, we list the linearization points (\lp{s}) of each method. Note that each method of the list can return either $true$ or $false$. So, we define the $\lp$ for six methods:

\begin{enumerate}
\item $\add(key, true)$: $write(n_1.next, n_3)$ in Line \ref{lin:add6} of $\add$ method.
\item $\add(key, false)$: $read(n_2.val)$ in Line \ref{lin:add3} of $\add$ method. 
\item $\rem(key, true)$: $write(n_2.marked, true)$ in Line \ref{lin:rem4} of $\rem$ method.
\item $\rem(key, false)$: $read(n_2.val)$ in Line \ref{lin:rem3} of $\rem$ method.
\item $\con(key, true)$: $read(n.marked)$ in \lineref{cont-tf} of $\con$ method \label{step:contT}.
\item $\con(key, false)$: $\lp$ is the last among the following lines executed. There are three cases here: 
\begin{itemize}
	\item [(a)] $read(n.val) \neq key$ in \lineref{cont-tf} of $\con$ method is the $\lp$, in case of no concurrent $\add(key, true)$.
	\item [(b)] $read(n.marked)$ in \lineref{cont-tf} of $\con$ method is the $\lp$, in case of no concurrent $\add(key, true)$ (like the case of \stref{contT}).
    \item [(c)] in case of concurrent $\add(key, true)$ by another thread, we add a dummy event just before Line \ref{lin:add6} of $add(key, true)$. This dummy event is the $\lp$ of $\con$ method if: (i) if in the post-state of $read(n.val)$ event in \lineref{cont-tf} of \con method, $n.val \neq key$ and $write(n_1.next, n_3)$ (with $n_3.val = key$) in Line \ref{lin:add6} of \add method executes before this $read(n.val)$. (ii) if in the post-state of $read(n.marked)$ event in \lineref{cont-tf} of \con method, $n.marked = true$ and $write(n_1.next, n_3)$ (with $n_3.val = key$) in Line \ref{lin:add6} of \add method executes before this $read(n.marked)$. An example is illustrated in Figure \ref{fig:lpcase}.
	
\end{itemize}
\end{enumerate}

Another important point to consider is that the \mth $m_i$ in an execution can go through several possible \lp events before returning a value. We then assume that the final \lp event executed decides the return parameters of the \mth. Let us illustrate this again with the case of contains \mth of the \lazy \cds. Consider an execution $E^H$ having the contains \mth $m_i$ concurrently executing with $add(k, true)$ \mth. In this case, the \lp of $m_i$ depends on the \lp of $add(k, true)$ if $m_i$ returns false. Suppose $m_i$ executes the event, say $e_x$, that corresponds to the \lp of $contains(k, false)$. Then later, the contains \mth also executes the event, say $e_y$ corresponding to the \lp of $contains(k, true)$ which is reading of a shared memory variable $n.marked$ of node $n$. If $n.marked$ is false then the contains \mth $m_i$ returns true and $e_y$ is the \lp. Otherwise, $m_i$ returns false and $e_x$ is \lp. Thus $m_i$ executes both $e_x$ and $e_y$. Either of them can be the \lp depending on the system state. 
\cmnt{
\begin{figure}[H]
	\centerline{\scalebox{0.5}{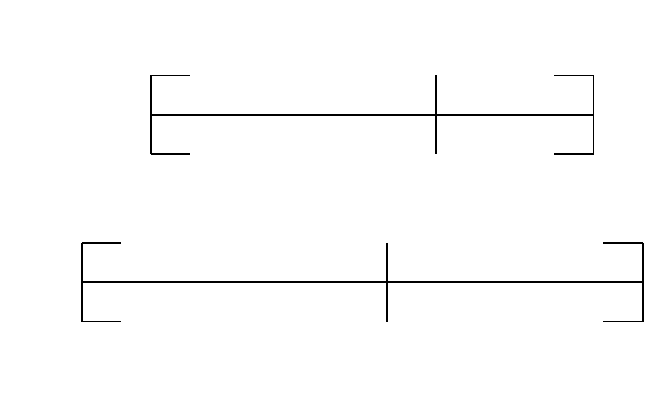}}
	\caption{$Contains(7, false)$ LP is depending on $Add(7, true)$ method LP}
	\label{fig:aadda2}
\end{figure}
    }

\begin{figure*}[th]
\captionsetup{font=scriptsize}
	\centerline{\scalebox{0.55}{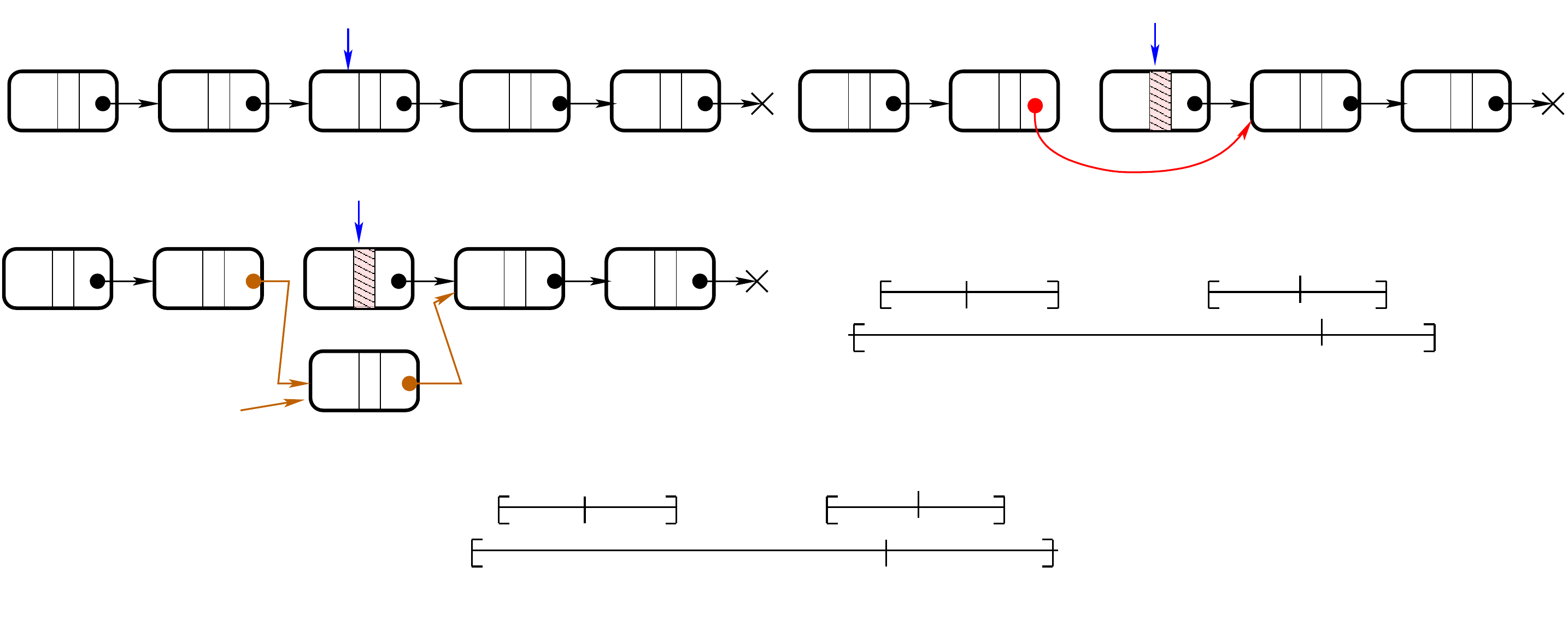}}
	\caption{An illustration of a concurrent set based linked list where the LP of the \con method does not lie in the code of the method. (a) Thread $T_3$ begins executing $\con(7)$ by traversing the list until it finds a node with key greater than or equal to 7 (Line \ref{lin:con6}). At the same time, thread $T_2$ starts the process of deletion of node $7$. (b) depicts that $T_2$ successfully performs deletion of $7$. (c) After this, Thread $T_1$ tries to add a new node with key $7$ and upon not encountering it in the list already; adds it successfully. Here thread $T_3$ has become slow and is still pointing to the deleted node $7$. It now executes Line \ref{lin:con6} and returns false; even though the node with key $7$ is present in the list, thus resulting in a illegal sequentialisation. The correct LP order is obtained by linearising \con just before the LP of the \add method. (d) shows the correct sequential history: $T_2.\rem(7, true) <_H T_3.\con(7, false) <_H T_1.\add(7, true)$.}
	\label{fig:lpcase}
\end{figure*}

\subsubsection{\textbf{Proof of Concurrent Lazy Linked List}}
\noindent \\ In this subsection, we describe the lemmas to prove the correctness of the concurrent lazy list structure. We say a node $n$ is a \emph{public} node if it has a incoming link, which makes it reachable from the head of the linked list. We assume that \emph{\head} and \emph{\tail} node are \emph{public} nodes.

\begin{observation}
\label{obs:node-forever}
	Consider a global state $S$ which has a node $n$. Then in any future state $S'$ of $S$, $n$ is node in $S'$ as well. Formally, $\langle \forall S, S': (n \in \nodes{S}) \land (S \sqsubset S') \Rightarrow (n \in \nodes{S'}) \rangle$. 
\end{observation}

\noindent With this observation, we assume that nodes once created do not get deleted (ignoring garbage collection). 

\begin{observation}	\label{obs:node-val}
Consider a global state $S$ which has a node $n$ and it is initialized to $n.val$.
	\begin{enumerate}[label=\ref{obs:node-val}.\arabic*]
	    \item \label{obs:node-val-future}
	     Then in any future state $S'$, where node $n$ exists, the value of $n$ does not change. Formally, $\langle \forall S, S': (n \in \nodes{S}) \land (S \sqsubset S') \land (n \in \nodes{S'})  \Rightarrow (S.n.val = S'.n.val) \rangle$. 
	    \item \label{obs:node-val-past}
	    Then in any past state $S''$, where node $n$ existed, the value of $n$ was the same. Formally, $\langle \forall S, S'': (n \in \nodes{S}) \land (S'' \sqsubset S) \land (n \in \nodes{S''}) \Rightarrow (S.n.val = S''.n.val) \rangle$. 
	\end{enumerate}
\end{observation}

\begin{observation}
	\label{lem:node-mark}
	Consider a global state $S$ which has a node $n$ and it is marked. Then in any future state $S'$ the node $n$ stays marked. Formally, $\langle \forall S, S': (n \in \nodes{S}) \land (S.n.marked) \land (S \sqsubset S') \Rightarrow (S'.n.marked) \rangle$. 
\end{observation}

\begin{observation}
	\label{lem:node-marknext}
Consider a global state $S$ which has a node $n$ which is marked. Then in any future state $S'$, $n.next$ remains unchanged. Formally, $\langle \forall S, S': (n \in \nodes{S}) \land (S.n.marked) \land (S \sqsubset S') \Longrightarrow (S'.n.next = S.n.next) \rangle$. 
\end{observation}

 \begin{definition}
\label{def:abs}
	$S.\abs \equiv \{n | (n \in \nodes{S}) \land (S.\head \rightarrow^* S.n) \land (\neg S.n.marked)\}$.
\end{definition}

\noindent This definition of $\abs$ captures the set of all nodes of $\abs$ for the global state $S$. It consists of all the $\node{s}$ that are reachable from $\head$ of the list (\emph{public}) and are not marked for deletion. 

\begin{observation}	\label{obs:locate}
 Consider a global state $S$ which is the post-state of return event of the method $\loct(key)$ invoked in the $\add$ or $\rem$ methods. Suppose the $\loct$ method returns $\langle n_1, n_2\rangle$. Then in the state $S$, we have, 
 \begin{enumerate}[label=\ref{obs:locate}.\arabic*]
    \item \label{obs:locate1}
    $\langle (n_1, n_2) \in \nodes{S} \rangle$. 
    \item \label{obs:locate2}
    $\langle (S.lock.acquire(n_1) = true) \wedge (S.lock.acquire(n_2) = true) \rangle$
    \item \label{obs:locate3}
    $\langle S.n_1.next = S.n_2 \rangle$ 
      \item \label{obs:locate4}
    $\langle \neg(S.n_1.marked) \land \neg(S.n_2.marked) \rangle$ 
\end{enumerate} 
\end{observation}

\begin{lemma}
\label{lem:loc-ret}
Consider the global state $S$ which is the post-state of return event of the method $\loct(key)$ invoked in the $\add$ or $\rem$ methods. Say, the $\loct$ method returns $(n_1, n_2)$. Then in the state $S$, we have that $(S.n_1.val < key \leq S.n_2.val)$.
\end{lemma}

\begin{proof}
Line \ref{lin:loc3} of $\loct$ method initialises $S.n_1$ to $\head$ and $S.n_2$ $=$ $S.n_1.next$ by Line \ref{lin:loc4}. The last time Line \ref{lin:loc6} in the while loop was executed, we know that $S.n_1.val$ $<$ $S.n_2.val$. The value of node does not change, from \obsref{node-val}. So, before execution of Line \ref{lin:loc9}, we know that $S.n_2.val$ $\geq$ $key$ and $S.n_1.val$ $<$ $S.n_2.val$. These nodes $\in$ $\nodes{S}$ and $S.n_1.val < key \leq S.n_2.val$. Also, putting together Observation \ref{obs:locate2}, \ref{obs:locate3} and \ref{obs:node-val} that node $n_1$ and $n_2$ are locked (do not change), hence, the lemma holds when $\loct$ returns. 
\end{proof}

\begin{observation}
\label{obs:n-loct-marked}
Consider a global state $S$ which has a node $n$ that is marked. Then there will surely be some previous state $S'$ ($S'\sqsubset S$) such that $S'$ is the state after return of \loct($n.val$) method.
\end{observation}

\begin{observation}
\label{obs:unmarked-marked-lock}
Consider the global state $S$ which has a node n. If S.n is unmarked and S.n.next is marked, then n and n.next are surely locked in the state S.
\end{observation}

\begin{lemma}
\label{lem:n2-next-marked}
 Consider a global state $S$ which is the post-state of return event of the $\loct(key)$ method (invoked by the $\add$ or $\rem$ methods). Say, the $\loct$ method returns $\langle n_1, n_2\rangle$. Then in the state $S$, we have that the successor node of $n_2$ (if it exists) is unmarked i.e. $\neg(S.n_2.next.marked)$.
\end{lemma}

\begin{proof}
 We prove the lemma by using induction on the return events of the $\loct$ method in $E^H$.\\[0.1cm]
\textbf{Base condition:} Initially, before the first return of the \loct, we know that ($\head.key$ $<$ $\tail.key$) and $\head.next$ is $\tail$ and $\tail.marked$ is set to $false$ and $(\head, \tail)$ $\in$ $\nodes{S}$. In this case, locate will return $\langle \head, \tail \rangle$ such that the successor of \tail does not exist.\\[0.1cm] 
\noindent \textbf{Induction Hypothesis:} Say, upto the first $k$ return events of \loct, the successor of $n_2$ (if it exists) is unmarked. \\[0.1cm]
\noindent \textbf{Induction Step:} So, by the observing the code, the $(k+1)^{st}$ event which can be the return of the \loct method can only be at \lineref{loc-ret}.\\
We prove by contradiction. Suppose when thread $T_1$ returns $\langle n_1, n_2 \rangle$ after invoking \emph{\loct} method in state $S$, $n_2.next$ is $marked$. By Observation \ref{obs:locate}, it is known that, $(n_1, n_2) \in \nodes{S}$, $n_1,n_2$ are locked, $n_1.next = n_2$ and $(n_1, n_2)$ are unmarked. Suppose another thread say $T_2$ is trying to remove the node $n_2.next$. From the \obsref{n-loct-marked}, it needs to invoke the \emph{\loct} method. Again, we know from the \obsref{locate} that when \emph{\loct} method returns, it must have acquired lock on $n_2$ and $n_2.next$. However, since $n_2$ is already locked, it cannot proceed until $T_1$ has released its lock on $n_2$. Hence the node $n_2.next$ cannot be marked. This contradicts our initial assumption. 
\end{proof}

\begin{observation}
\label{obs:consec-next}
Consider a global state $S$ which has two non-consecutive nodes $n_p$, $n_q$ where $n_p$ is unmarked and $n_q$ is marked. Then we have that in any future state $S'$, $n_p$ cannot point to $n_q$. Formally, $\langle \neg (S.n_p.marked) \land (S.n_q.marked) \land (S.n_p.next \neq n_q) \land (S \sqsubset S') \Longrightarrow (S'.n_p.next \neq S'.n_q) \rangle$. 
\end{observation}

\begin{lemma}
\label{lem:next-next-unmarked}
In any global state $S$, consider three nodes p, q \& r such that p.next = q and q.next = r and only q is marked. Then in a future state $S'$ $(S \sqsubset S')$ where p.next = q and p is still unmarked, r will surely be unmarked.
\end{lemma}

\begin{proof}
We prove the lemma by contradiction. Suppose in state $S'$, node $r$ is marked and $p.next = q$ and $q.next = r$. From Observation \ref{lem:node-mark}, we know that $q$ will remain marked. From the \obsref{n-loct-marked} we know that any node is marked only after invoking the \emph{\loct} method. Say, the node $q$ was marked by the thread $T_1$ by invoking the \emph{\rem} method. As we know from the \lemref{n2-next-marked} that when \emph{$T_1.$\loct} returns $\langle q, q.next = r \rangle$, the successor of $q$ (i.e. $r$) is unmarked, which contradicts our intial assumption. Hence the lemma holds.
\end{proof}

\begin{lemma}
\label{lem:val-change}
For any node $n$ in a global state $S$, we have that $\langle \forall n \in \nodes{S} \land n.next \neq null: S.n.val < S.n.next.val \rangle$. 
\end{lemma}
\begin{proof}
We prove the lemma by inducting on all events in $E^H$ that change the $next$ field of a node $n$.\\[0.1cm]
\textbf{Base condition:} Initially, before the first event that changes the next field, we know that ($\head.key$ $<$ $\tail.key$) $\land$ $(\head, \tail)$ $\in$ $\nodes{S}$.\\[0.1cm] 
\textbf{Induction Hypothesis:} Say, in any state $S$ upto first $k$ events that change the $next$ field of any node, $\forall n \in \nodes{S} \land S.n.next \neq null:$ $S.n.val$ $<$ $S.n.next.val$.\\[0.1cm]
\textbf{Induction Step:}
So, by observing the code, the $(k+1)^{st}$ event which can change the $next$ field can be only one among the following: 
\begin{enumerate}
\item \textbf{Line \ref{lin:add5} of $\add$ method:} Let $S_1$ be the state after the \lineref{add3}. We know that when $\loct$ (\lineref{add2}) returns by the \obsref{locate}, $S_1.n_1$ \& $S_1.n_2$ are not marked,  $S_1.n_1$ \& $S_1.n_2$ are locked, $S_1.n_1.next =S_1.n_2$. By the \lemref{loc-ret} we have $(S_1.n_1.val \leq S_1.n_2.val)$. Also we know from Observation \ref{obs:node-val} that node value does not change, once initialised. To reach Line \ref{lin:add5}, $n_2.val \neq key$ in the \lineref{add3} must evaluate to true. Therefore, $(S_1.n_1.val < key < S_1.n_2.val)$. So, a new node $n_3$ is created in the \lineref{add4} with the value $key$ and then a link is added between $n_3.next$ and $n_2$ in the \lineref{add5}. So this implies $n_3.val < n_2.val$ even after execution of line \ref{lin:add5} of $\add$ method.


\item \textbf{Line \ref{lin:add6} of $\add$ method:} Let $S_1$ and $S_2$ be the states after the \lineref{add2} and \lineref{add6} respectively. By observing the code, we notice that the Line \ref{lin:add6} (next field changing event) can be executed only after the $\loct$ method returns. From Lemma \ref{lem:loc-ret}, we know that when $\loct$ returns then $S_1.n_1.val$ $<$ key $\leq$ $S_1.n_2.val$. To reach Line \ref{lin:add6} of $\add$ method, Line \ref{lin:add3} should ensure that $S_1.n_2.val$ $\neq$ key. This implies that $S_1.n_1.val$ $<$ $key$ $<$ $S_1.n_2.val$. From Observation~\ref{obs:locate3}, we know that $S_1.n_1.next$ = $S_1.n_2$. Also, the atomic event at Line \ref{lin:add6} sets $S_2.n_1.next$ = $S_2.n_3$ where $S_2.n_3.val = key$. \\[0.1cm]
Thus from $S_2.n_1.val$ $<$ $(S_2.n_3.val=key)$ $<$ $S_2.n_2.val$ and $S_2.n_1.next$ = $S_2.n_3$, we get $S_2.n_1.val$ $<$ $S_2.n_1.next.val$. Since $(n_1, n_2)$ $\in$ $\nodes{S}$ and hence, $S.n_1.val$ $<$ $S.n_1.next.val$.

\item \textbf{Line \ref{lin:rem5} of $\rem$ method:} Let $S_1$ and $S_2$ be the states after the \lineref{rem2} and \lineref{rem4} respectively. By observing the code, we notice that the Line \ref{lin:rem5} (next field changing event) can be executed only after the $\loct$ method returns. From Lemma \ref{lem:loc-ret}, we know that when $\loct$ returns then $S_1.n_1.val$ $<$ $key$ $\leq$ $S_1.n_2.val$. To reach Line \ref{lin:rem5} of $\rem$ method, Line \ref{lin:rem3} should ensure that $S_1.n_2.val$ $=$ $key$. Also we know from Observation \ref{obs:node-val} that node value does not change, once initialised. This implies that $S_2.n_1.val$ $<$ ($key$ $=$ $S_2.n_2.val$). From Observation~\ref{obs:locate3}, we know that $S_2.n_1.next$ = $n_2$. Also, the atomic event at line \ref{lin:rem6} sets $S_2.n_1.next$ = $S_2.n_2.next$. \\[0.1cm]
We know from Induction hypothesis, $S_2.n_2.val < S_2.n_2.next.val$. Thus from $S_2.n_1.val$ $<$ $S_2.n_2.val$ and $S_2.n_1.next$ = $S_2.n_2.next$, we get $S_2.n_1.val$ $<$ $S_2.n_1.next.val$. Since $(n_1, n_2)$ $\in$ $\nodes{S}$ and hence, $S.n_1.val$ $<$ $S.n_1.next.val$. 
\end{enumerate}
\end{proof}

\begin{corollary}
\label{cor:same-node}
There cannot exist two $\node$s with the same key in $S.\abs$ of a particular global state $S$. 
\end{corollary}

\begin{lemma}
\label{lem:mark-reach}
	In a global state $S$, any non-marked public node $n$ is reachable from $\head$. Formally, \\ $\langle \forall S, n: (n \in \nodes{S}) \land (\neg S.n.marked) \implies (S.\head \rightarrow^* S.n) \rangle$. 
\end{lemma}

\begin{proof}
We prove by Induction on events that change the next field of the node (as these affect reachability), which are Line \ref{lin:add5} \& \ref{lin:add6} of $\add$ method and Line \ref{lin:rem5} of $\rem$ method. It can be seen by observing the code that $\loct$ and $\con$ method do not have any update events.\\[0.1cm]
\textbf{Base step:} Initially, before the first event that changes the next field of any node, we know that $\langle (\head, \tail)$ $\in$ $\nodes{S}$ $\land$ $\neg$($\head.marked$) $\land$ $\neg$($\tail.marked$) $\land$ $(\head$ $\rightarrow$ $\tail) \rangle$.\\[0.1cm]
\textbf{Induction Hypothesis:} Say, the first $k$ events that changed the next field of any node in the system did not make any unmarked node unreachable from the $\head$.\\[0.1cm]
\textbf{Induction Step:} As seen by observing the code, the $(k+1)^{st}$ event can be one of the following events that change the next field of a node:
\begin{enumerate}
\item \textbf{Line \ref{lin:add4} \& \ref{lin:add5} of $\add$ method:} Let $S_1$ be the state after the \lineref{add2}. Line \ref{lin:add4} of the $\add$ method creates a new node $n_3$ with value $key$. Line \ref{lin:add5} then sets $S_1.n_3.next$ $=$ $S_1.n_2$. Since this event does not change the next field of any node reachable from the $\head$ of the list, the lemma is not violated. 
\item \textbf{Line \ref{lin:add6} of $\add$ method:} By observing the code, we notice that the Line \ref{lin:add5} (next field changing event) can be executed only after the $\loct$ method returns. Let $S_1$ and $S_2$ be the states after the \lineref{add3} and \lineref{add6} respectively. From Observation \ref{obs:locate3}, we know that when $\loct$ returns then $S_1.n_1.marked = S_1.n_2.marked = false$. From Line \ref{lin:add4} \& \ref{lin:add5} of $\add$ method, $(S_1.n_1.next =S_1.n_3)$ $\land$ $(S_1.n_3.next = S_1.n_2)$ $\land$ ($\neg S_1.n_3.marked$). It is to be noted that (From Observation \ref{obs:locate2}), $n_1$ \& $n_2$ are locked, hence no other thread can change $S_1.n_1.marked$ and $S_1.n_2.marked$. Also from \obsref{node-val}, a node's key field does not change after initialization. Before executing Line \ref{lin:add6}, $S_1.n_1.marked = false$ and $S_1.n_1$ is reachable from $\head$. After Line \ref{lin:add6}, we know that from $S_2.n_1$, unmarked node $S_2.n_3$ is also reachable. Formally,	
$(S_2.\head \rightarrow^* S_2.n_1) \land \neg (S_2.n_1.marked) \land (S_2.n_1 \rightarrow S_2.n_3) \land \neg (S_2.n_3.marked) \implies (S_2.\head \rightarrow^* S_2.n_3)$.

\item \textbf{Line \ref{lin:rem5} of $\rem$ method:} Let $S_1$ and $S_2$ be the states after the execution of \lineref{rem3} and \lineref{rem5} respectively. By observing the code, we notice that the Line \ref{lin:rem5} (next field changing event) can be executed only after the $\loct$ method returns. From Observation \ref{obs:locate2}, we know that when $\loct$ returns then $S_1.n_1.marked = S_1.n_2.marked = false$. We know that $S_1.n_1$ is reachable from $\head$ and from Line \ref{lin:rem4} and \ref{lin:rem5} of $\rem$ method, $S_2.n_2.marked = true$ and later sets $S_2.n_1.next$ $=$ $S_2.n_2.next$. It is to be noted that (From Observation \ref{obs:locate2}), $S_1.n_1$ \& $S_1.n_2$ are locked, hence no other thread can change $S_1.n_1.marked$ and $S_1.n_2.marked$. This event does not affect reachability of any non-marked node. Also from \obsref{node-val}, a node's key does not change after initialization. And from Observation \ref{lem:node-mark}, a marked node continues to remain marked. If $S_2.n_2.next$ is unmarked (reachable), then it continues to remain unmarked \& reachable. So this event does not violate the lemma.
\end{enumerate}
\end{proof}

\begin{lemma}
\label{lem:val-abs}
Consider the global state $S$ such that for any unmarked node $n$, if there exists a key strictly greater than $n.val$ and strictly smaller than $n.next.val$, then the node corresponding to the key does not belong to $S.\abs$. Formally, $\langle \forall S, n, key$ : $\neg$$(S.n.marked)$ $\land$ $(S.n.val < key < S.n.next.val)$ $\implies$ $\node(key)$ $\notin S.\abs \rangle$. 
\end{lemma}

\begin{proof}
 We prove by contradiction. Suppose there exists a $key$ which is strictly greater than $n.val$ and strictly smaller than $n.next.val$ and then it belongs to $S.\abs$. From the \obsref{node-forever}, we know that node $n$ is unmarked in a global state $S$, so it is belongs to $\nodes{S}$. But we know from Lemma \ref{lem:mark-reach} that any unmarked node should be reachable from \head. Also, from Definition \ref{def:abs}, any unmarked node i.e. $n$ in this case, is reachable from \head and belongs to $S.\abs$. From the \obsref{node-val}, we know that the node's key value does not change after initialization. So both the nodes $n$ and $n.next$ belong to $S.\abs$. From the \lemref{val-change} we know that $n.val<n.next.val$. So node $n'$ can not be present in between $n$ and $n.next$. Which contradicts the initial assumption. Hence $\langle \forall S, n, key$ : $\neg(S.n.marked)$ $\land$ $(S.n.val < key < S.n.next.val)$ $\implies$ $\node(key)$ $\notin S.\abs \rangle$. 
 
\end{proof}

\begin{lemma}
\label{lem:change-abs}
	Only the events $write(n_1.next, n_3)$ in \ref{lin:add6} of \add method and $write(n_2.marked, true)$ in \ref{lin:rem4} of \rem method can change the $\abs$.
\end{lemma}
\begin{proof}
It is to be noted that the $\loct$ and $\con$ methods do not have any update events. By observing the code, it appears that the following (write) events of the $\add$ and $\rem$ method can change the $\abs$:
\begin{enumerate}
\item \textbf{Line \ref{lin:add4} \& \ref{lin:add5} of $\add$ method:}  In Algorithm~\ref{alg:add}, let $S_1.\abs$ be the initial state of the $\abs$, such that we know from Line \ref{lin:add3} that $key$ $\notin$ $S_1.\abs$. Line \ref{lin:add4} of the $\add$ method creates a node $n_3$ with value $key$, i.e. $n_3.val = key$. Now, Line \ref{lin:add5} sets $S_1.n_3.next$ $=$ $S_1.n_2$. Since this event does not change the next field of any node reachable from the $\head$ of the list, hence from Definition \ref{def:abs}, $S_1.\abs$ remains unchanged after these events.

\item \textbf{Line \ref{lin:add6} of $\add$ method:}  Let $S_1$ and $S_2$ be the states after the \lineref{add3} and \lineref{add6} respectively. At line \ref{lin:add3}, $true$ evaluation of the condition leads to the execution of $S_1.n_1.next = S_1.n_3$ at Line \ref{lin:add6}. Also, $S_1.n_1$ and $S_1.n_2$ are locked, therefore from Observation \ref{obs:locate}, $\head$ $\rightarrow^*$ $S_1.n_1$. From line \ref{lin:add5} \& \ref{lin:add6} we get: $S_1.n_1$ $\rightarrow$ $S_1.n_3$ $\rightarrow$ $S_1.n_2$. Hence, $\head$ $\rightarrow$ $S_1.n_1$ $\rightarrow$ $S_1.n_3$ $\rightarrow$ $S_1.n_2$ follows. We have $\neg$ $(S_2.n_3.marked)$ $\land$ $(\head$ $\rightarrow$ $S_2.n_3)$. Thus from Definition \ref{def:abs}, $S_1.\abs$ changes to $S_2.\abs$ $=$ $S_1.\abs$ $\cup$ $n_3$.

\item \textbf{Line \ref{lin:rem4} of $\rem$ method:} Let $S_1$ be the state after the \lineref{rem4}. By observing the code, we notice that the state before execution of Line \ref{lin:rem4} satisfies that $key$ $\in$ $S.\abs$. After execution of line \ref{lin:rem4}, $\abs$ changes such that $key$ $\notin$ $S.\abs$. Note that this follows from Definition \ref{def:abs}.

\item \textbf{Line \ref{lin:rem5} of $\rem$ method:} Let $S_1$ be the state after the execution of \lineref{rem4}. Till line \ref{lin:rem4} of the $\rem$ method, $S.\abs$ has changed such that $S_1.n_2.val$ $\notin$ $S.\abs$. So even after the execution of Line \ref{lin:rem5} when $S_1.n_1.next$ is set to $S_1.n_2.next$, $S.\abs$ remains unchanged (from Definition \ref{def:abs}).
\end{enumerate}
Hence, only the events in Line \ref{lin:add6} of $\add$ method and in Line \ref{lin:rem4} of $\rem$ method can change the $\abs$. 
\end{proof}

\begin{corollary}
\label{cor:change-abs}
Both these events $write(n_1.next, n_3)$ in \ref{lin:add6} of \add method and $write(n_2.marked, true)$ in \ref{lin:rem4} of \rem method change the $\abs$ are in fact the Linearization Points(LPs) of the respective methods. 
\end{corollary}
\ignore{
\begin{figure}[!htbp]
\centerline{\scalebox{0.6}{\input{figs/lemma17.pdf_t}}}
\caption{LP of $\con(key, false)$ is at $read(n.val$ = $key)$ at Line 6}
\end{figure}
}

\begin{observation}
\label{obs:seq-spec}
Consider a sequential history $\spl{S}$. Let $S$ be a global state in $\spl{S}.allStates$ before the execution of the \mth and $S'$ be a global state just after the return of the \mth $(S \sqsubset S')$. Then we have the sequential specification of all methods as follows, 
	\begin{enumerate}[label=\ref{obs:seq-spec}.\arabic*]
    \item \label{obs:addT-seq}
    For a given key, suppose node(key) $\notin$ S.\abs. In this state, suppose \add(key) \mth is (sequentially) executed. Then the \add \mth will return true and node(key) will be present in $S'.\abs$.
    Formally, $\langle \forall S: (\node(key) \notin S.\abs) \xRightarrow[]{ seq\text{-}add} \spl{S}.\add(key, true) \land (S \sqsubset S') \land (\node(key) \in S'.\abs) \rangle$.
    
	\item \label{obs:addF-seq}
	For a given key, suppose node(key) $\in$ S.\abs. In this state, suppose \add(key) \mth is (sequentially) executed. Then the \add \mth will return false and node(key) will continue to be present in $S'.\abs$.  
	 Formally, $\langle \forall S: (\node(key) \in S.\abs) \xRightarrow[]{seq\text{-}add} \spl{S}.\add(key, false) \land (S \sqsubset S') \land (\node(key) \in S'.\abs)\rangle$.
    
    \item \label{obs:removeT-seq}
    For a given key, suppose node(key) $\in$ S.\abs. In this state, suppose \rem(key) \mth is (sequentially) executed. Then the \rem \mth will return true and node(key) will not be present in $S'.\abs$. 
     Formally, $\langle \forall S: (\node(key) \in S.\abs) \xRightarrow[]{seq\text{-}remove} \spl{S}.\rem(key, true) \land (S \sqsubset S') \land (\node(key) \notin S'.\abs) \rangle$.
    
    \item 	\label{obs:removeF-seq}
   For a given key, suppose node(key) $\notin$ S.\abs. In this state, suppose \rem(key) \mth is (sequentially) executed. Then the \rem \mth will return false and node(key) will continue to be not present in $S'.\abs$.
   Formally, $\langle \forall S: (\node(key) \notin S.\abs) \xRightarrow[]{seq\text{-}remove} \spl{S}.\rem(key, false) \land (S \sqsubset S') \land (\node(key) \notin S'.\abs)\rangle$.
    
    \item 	\label{obs:containT-seq}
   For a given key, suppose node(key) $\in$ S.\abs. In this state, suppose \con(key) \mth is (sequentially) executed. Then the \con \mth will return true and node(key) will continue to be present in $S'.\abs$. 
   Formally, $\langle \forall S: (\node(key) \in S.\abs) \xRightarrow[]{seq\text{-}contains} \spl{S}.\con(key, true) \land (S \sqsubset S') \land (\node(key) \in S'.\abs) \rangle$. 
    
\item \label{obs:containF-seq}
     For a given key, suppose node(key) $\notin$ S.\abs. In this state, suppose \con(key) \mth is (sequentially) executed. Then the \con \mth will return false and node(key) will continue to be not present in $S'.\abs$.  
     Formally, $\langle \forall S: (\node(key) \notin S.\abs) \xRightarrow[]{seq\text{-}contains} \spl{S}.\con(key, false) \land (S \sqsubset S') \land (\node(key) \notin S'.\abs) \rangle$. 
\end{enumerate} 
\end{observation}

\begin{lemma}
\label{lem:addT-conc}
If some \add(key) method returns true in $E^H$ then,
	\begin{enumerate}[label=\ref{lem:addT-conc}.\arabic*]
    \item \label{lem:addT-conc-pre}
 The $\node(key)$ is not present in the pre-state of $LP$ event of the method. Formally, $\langle \add(key, true) \Longrightarrow (node(key) \notin (\prees{\add(key, true)}) \rangle$. 
 \item \label{lem:addT-conc-post}
  The $\node(key)$ is present in the post-state of $LP$ event of the method. Formally, \\ $\langle \add(key, true) \Longrightarrow (node(key) \in  (\postes{\add(key, true)}) \rangle$. 
 \end{enumerate}
\end{lemma}

\begin{proof}
\noindent
\begin{itemize}
\item \textbf{\ref{lem:addT-conc-pre}}: From Line \ref{lin:add2}, when $\loct$ returns in state $S_1$ we know that (from Observation \ref{obs:locate} \& Lemma \ref{lem:mark-reach}), nodes $n_1$ and $n_2$ are locked, ($n_1, n_2$) $\in$ $\nodes{S_1}$ and $n_1.next = n_2$. Also, $S_1.n_1.val$ $<$ $key$ $\leq$ $S_1.n_2.val$ from Lemma \ref{lem:loc-ret}. If this method is to return true, Line \ref{lin:add3}, $n_2.val$ $\neq$ $key$ must evaluate to true. Also from Lemma \ref{lem:val-abs}, we conclude that $\node(key)$ does not belong to $S_1.\abs$. And since from \obsref{node-val}, no node changes its key value after initialization, $\node(key)$ $\notin$ $S_2.\abs$, where $S_2$ is the pre-state of the $LP$ event of the method. Hence $\node(key)$ $\notin$ $(\prees{\add(key, true)})$. 
\item \textbf{\ref{lem:addT-conc-post}}: From the Lemma \ref{lem:addT-conc-pre} we get that \emph{\node(key)} is not present in the pre-state of the \lp event. From Lemma \ref{lem:change-abs}, it is known that only \lp event can change the $S.\abs$. Now after execution of the \lp event i.e. $write(n_1.next, n_3)$ in the Line \ref{lin:add6}, $\node(key)$ $\in$ $S'.\abs$, where $S'$ is the post-state of the \lp event of the method. Hence, $\langle \add(key, true) \Longrightarrow (\node(key) \in (\postes{\add(key, true)}) \rangle$.
\end{itemize}
\end{proof}

\begin{lemma}
\label{lem:addF-conc}
If some \add(key) method returns false in $E^H$, then 
\begin{enumerate}[label=\ref{lem:addF-conc}.\arabic*]
    \item \label{lem:addF-conc-pre}
    The $\node(key)$ is present in the pre-state of $LP$ event of the method. Formally, \\ $\langle \add(key, false) \Longrightarrow (node(key) \in (\prees{\add(key, false)}) \rangle$. 
     \item \label{lem:addF-conc-post}
    The $\node(key)$ is present in the post-state of $LP$ event of the method. Formally, \\ $\langle \add(key, false) \Longrightarrow (node(key) \in (\postes{\add(key, false)}) \rangle$. 
\end{enumerate}    
\end{lemma}

\begin{proof}
\noindent
\begin{itemize}
\item \textbf{\ref{lem:addF-conc-pre}}:
From Line \ref{lin:add2}, when $\loct$ returns in state $S_1$ we know that (from Observation \ref{obs:locate} \& Lemma \ref{lem:mark-reach}), nodes $n_1$ and $n_2$ are locked, ($n_1, n_2$) $\in$ $\nodes{S}$ and $n_1.next = n_2$. Also, $n_1.val$ $<$ $key$ $\leq$ $n_2.val$ from Lemma \ref{lem:loc-ret}.  If this method is to return false, Line \ref{lin:add3}, $n_2.val$ $\neq$ $key$ must evaluate to false. So $\node(key)$ which is $n_2$ belongs to $S_1.\abs$. 
And since from \obsref{node-val}, no node changes its key value after initialization and the fact that it is locked, $\node(key)$ $\in$ $S_2.\abs$, where $S_2$ is the pre-state of the $LP$ event of the method. Hence $\node(key)$ $\in$ $(\prees{\add(key, false)})$. 
\item \textbf{\ref{lem:addF-conc-post}}:
From the Lemma \ref{lem:addF-conc-pre} we get that \emph{\node(key)} is present in the pre-state of the \lp event. This \lp event  $n_2.val$ $\neq$ $key$ in Line \ref{lin:add3} does not change the $S.\abs$, Now after execution of the \lp event the $\node(key)$ also present in the $S'.\abs$, where $S'$ is the post-state of the \lp event of the method. Hence, $\langle \add(key, false) \Longrightarrow (\node(key) \in \\ (\postes{\add(key, false)}) \rangle$.

\end{itemize}
\end{proof}

\begin{lemma}
\label{lem:removeT-conc}
If some \rem(key) method returns true in $E^H$,  then 
\begin{enumerate}[label=\ref{lem:removeT-conc}.\arabic*]
    \item \label{lem:removeT-conc-pre}
    The $\node(key)$ is present in the pre-state of $LP$ event of the method. Formally, \\ $\langle \rem(key, true) \Longrightarrow (\node(key) \in (\prees{\rem(key, true)}) \rangle$. 
    \item \label{lem:removeT-conc-post}
     The $\node(key)$ is not present in the post-state of $LP$ event of the method. Formally, $\langle \rem(key, true) \Longrightarrow (\node(key) \notin (\postes{\rem(key, true)}) \rangle$. 
     \end{enumerate}
\end{lemma}

\begin{proof}
\noindent
\begin{itemize}
\item \textbf{\ref{lem:removeT-conc-pre}}:
From Line \ref{lin:rem2}, when $\loct$ returns in state $S_1$ we know that (from Observation \ref{obs:locate} \& Lemma \ref{lem:mark-reach}), nodes $n_1$ and $n_2$ are locked, ($n_1, n_2$) $\in$ $\nodes{S_1}$ and $n_1.next = n_2$. Also, $S_1.n_1.val$ $<$ $key$ $\leq$ $S_1.n_2.val$ from Lemma \ref{lem:loc-ret}. If this method is to return true, Line \ref{lin:rem3}, $n_2.val$ $=$ $key$ must evaluate to true. So we know that $\node(key)$ which is $n_2$ belongs to $S_1.\abs$. And since from \obsref{node-val}, no node changes its key value after initialization, $\node(key)$ $\in$ $S_2.\abs$, where $S_2$ is the pre-state of the $LP$ event of the method. Hence \\ $\node(key)$ $\notin$ $(\prees{\rem(key, true)})$. 
\item \textbf{\ref{lem:removeT-conc-post}}:
From the Lemma \ref{lem:removeT-conc-pre} we get that \emph{\node(key)} is present in the pre-state of the \lp event. This \lp event $write(n_2.marked, true)$ in the \lineref{rem4} changes the $S.\abs$. Now after execution of the \lp event the $\node(key)$ will not present in the $S'.\abs$, where $S'$ is the post-state of the \lp event of the method. Hence, $\langle \rem(key, true) \Longrightarrow \\ (\node(key) \notin (\postes{\rem(key, true)}) \rangle$.

\end{itemize}
\end{proof}

\begin{lemma}
\label{lem:removeF-conc}
If some \rem(key) method returns false in $E^H$,  then 
\begin{enumerate}[label=\ref{lem:removeF-conc}.\arabic*]
    \item \label{lem:removeF-conc-pre}
    The $\node(key)$ is not present in the pre-state of $LP$ event of the method. Formally, $\langle \rem(key, false) \Longrightarrow (\node(key) \notin  \prees{\rem(key, false)}) \rangle$. 
     \item \label{lem:removeF-conc-post}
    The $\node(key)$ is not present in the post-state of $LP$ event of the method. Formally, $\langle \rem(key, false) \Longrightarrow (\node(key) \notin  \postes{\rem(key, false)}) \rangle$. 
\end{enumerate}    

\end{lemma}

\begin{proof}
\noindent
\begin{itemize}
\item \textbf{\ref{lem:removeF-conc-pre}}:
From Line \ref{lin:rem2}, when $\loct$ returns in state $S_1$ we know that (from Observation \ref{obs:locate} \& Lemma \ref{lem:mark-reach}), nodes $n_1$ and $n_2$ are locked, ($n_1, n_2$) $\in$ $\nodes{S_1}$ and $n_1.next = n_2$. Also, $S_1.n_1.val$ $<$ $key$ $\leq$ $S_1.n_2.val$ from Lemma \ref{lem:loc-ret}. If this method is to return false, Line \ref{lin:rem3}, $n_2.val$ $=$ $key$ must evaluate to false. Also from Lemma \ref{lem:val-abs}, we conclude that $\node(key)$ does not belong to $S_1.\abs$. And since from \obsref{node-val}, no node changes its key value after initialization, $\node(key)$ $\in$ $S_2.\abs$, where $S_2$ is the pre-state of the $LP$ event of the method. Hence $\node(key)$ $\notin$ $(\prees{\rem(key, false)})$.
\item \textbf{\ref{lem:removeF-conc-post}}:
From the Lemma \ref{lem:removeF-conc-pre} we get that \emph{\node(key)} is not present in the pre-state of the \lp event. This \lp event \\ $(read(n_2.val) = key)$ in the \lineref{rem3} does not change the $S.\abs$. Now after execution of the \lp event the $\node(key)$ will not present in the $S'.\abs$, where $S'$ is the post-state of the \lp event of the method. Hence, $\langle \rem(key, falase) \Longrightarrow (\node(key) \notin  (\postes{\rem(key, false)}) \rangle$.

\end{itemize}
\end{proof}

\ignore{
\begin{proof}
We prove the lemma by contradiction. Let us assume that global state $S$ has two non-consecutive nodes $n_p$ and $n_q$ such that $n_p$ is unmarked, $n_q$ is marked and $n_p.next$ = $n_q$. If $n_p.next$ = $n_q$, than $n_p$ and $n_q$ are two consecutive nodes, which contradicts to the initial assumption that $n_p$ and $n_q$ non-consecutive nodes. The node $n_q$ is marked in the state $S$, before marking it, both $n_p$ and $n_q$ are to be locked by the \emph{\loct} method returns (by the \obsref{locate}). And after marking and before releasing lock $n_p.next$ should point to the $n_q.next$ in the state $S'$, So $S'.n_p.next \neq S'.n_q$, which also contradicts to the assumption that $n_p.next$ = $n_q$.  Hence, $\langle \neg (S.n_p.marked) \land (S.n_q.marked) \land (S.n_p.next \neq n_q) \land (S \sqsubset S') \Longrightarrow (S'.n_p.next \neq S'.n_q) \rangle$. 
\end{proof}
}
\begin{lemma}
\label{lem:consec-mark}
Consider a global state $S$ which has two consecutive nodes $n_p$, $n_q$ which are marked. Then we say that marking event of $n_p$ happened before marking event of $n_q$. Formally, $\langle \forall S: (n_p, n_q \in \nodes{S}) \land (S.n_p.marked) \land (S.n_q.marked) \land (S.n_p.next = S.n_q) \Rightarrow (n_p.marked <_E n_q.marked) \rangle$.
\end{lemma}

\begin{proof}
We prove by contradiction. We assume that $n_q$ was marked before $n_p$. Let $S'$ be the post-state of marking of the node $n_q$. It can be seen as in Figure \ref{fig:lemma26} that the state $S$ follows $S'$, i.e., $S'$ $\sqsubset$ $S$. This is because in state $S$ both $n_p$ \& $n_q$ are marked. So we know that in $S'$, $n_p$ is unmarked and $n_q$ is marked.
\begin{figure}[H]
\captionsetup{font=scriptsize}
\centerline{\scalebox{0.65}{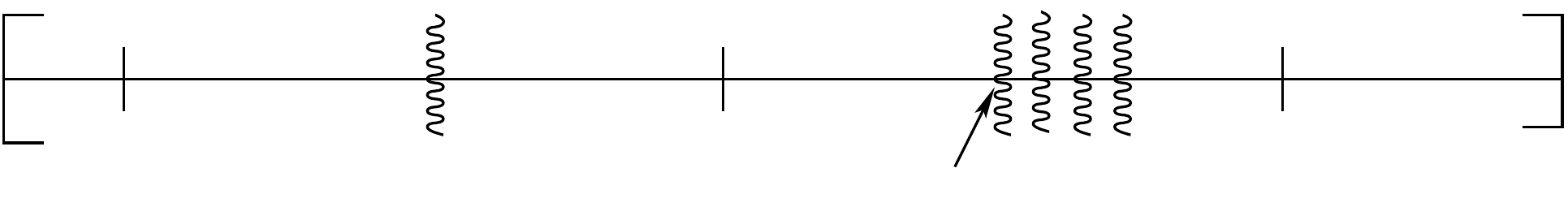}}
\caption{Scenario when event $n_q$.marking happens before $n_p$.marking}
\label{fig:lemma26}
\end{figure}
\noindent Now suppose in $S'$: $(n_p.next$ $\neq$ $n_q$). So, $(S'.n_p.next$ $\neq$ $S'.n_q)$ $\land\\ (\neg S'.n_p.marked)$. Also in the state $S$, we have that $S.n_p.next = S.n_q$ and $n_p$ and $n_q$ are both marked. This contradicts the \obsref{consec-next} that $S'.n_p.next$ $\neq$ $S'.n_q$. Hence in $S'$: $n_p.next$ must point to $n_q$.\\
Consider some state $S''$ immediately before marking event of $n_q$. We know that $S''.n_p.next = S''.n_q$ (similar argument), and $n_p$, $n_q$ are both unmarked (from Observation \ref{obs:locate2}). Then in some state $R$ after $S'$ and before $S$, $n_p.next$ $\neq$ $n_q$. From \obsref{consec-next}, unmarked node cannot point to marked node. Hence in state $S$ also, we will have that $S.n_p.next$ $\neq$ $S.n_q$. This contradicts the given statement that $S.n_p.next = S.n_q$. Hence proved that in $S'$, $n_p$ was marked before $n_q$. 
\end{proof}

\begin{lemma}
\label{lem:containsT-conc}
If some \con(key) method returns true in $E^H$,  then
\begin{enumerate}[label=\ref{lem:containsT-conc}.\arabic*]
    \item \label{lem:containsT-conc-pre}
The $\node(key)$ is present in the pre-state of $LP$ event of the method. Formally,\\ $\langle \con(key, true) \Longrightarrow (\node(key) \in \prees{\con(key, true)}) \rangle$. 
 \item \label{lem:containsT-conc-post}
The $\node(key)$ is present in the post-state of $LP$ event of the method. Formally,\\ $\langle \con(key, true) \Longrightarrow (\node(key) \in \prees{\con(key, true)}) \rangle$.
\end{enumerate}
\end{lemma}

\begin{proof}
\noindent
\begin{itemize}
\item \textbf{\ref{lem:containsT-conc-pre}}:
By observing the code, we realize that at the end of while loop at Line \ref{lin:con5} of $\con$ method, $n.val$ $\geq$ $key$. To return $true$, $n.marked$ should be false in $(\prees{\con}$. But we know from Lemma \ref{lem:mark-reach} that any unmarked node should be reachable from head. Also, from Definition \ref{def:abs}, any unmarked nodes that are reachable belong to $\abs$ in that state. From the \obsref{node-val} we know that the node's key value does not change after initialization. Hence \\ $\node(key)$ $\in$ $(\prees{\con(key, true)}$.
\item \textbf{\ref{lem:containsT-conc-post}}:
From the Lemma \ref{lem:containsT-conc-pre} we get that \emph{\node(key)} is present in the pre-state of the \lp event. This \lp event $(read(n.val) \neq key) \vee (read(n.marked))$ in the \lineref{con6} does not change the $S.\abs$. Now after execution of the \lp event the $\node(key)$ will be present in the $S'.\abs$, where $S'$ is the post-state of the \lp event of the method. Hence, $\langle \con(key, true) \Longrightarrow (\node(key) \in (\postes{\con(key, true)}) \rangle$.

\end{itemize}
\end{proof}

\begin{lemma}
\label{lem:con-rem-mark}
	Consider a global state $S$ which has a node $n$. If \con(key) method is running concurrently with a \rem(key) method and $\node(key) = n$ and $n$ is marked in the state $S$, then marking of $S.n$ happened only after \con(key) started. 
\end{lemma}
\begin{proof}
 
\end{proof}
\noindent \textbf{Notations used in Lemma \ref{lem:containsF-conc}:}\\
$\con(key)$ executes the while loop to find out location of the node $n_x$ where $n_x.val$ $\leq$ $key$ and $n_x$ $\in$ $\abs$. We denote execution of the last step $n_x$ = $read(n_{x-1}.next)$ which satisfies $n_x.val$ $\leq$ $key$. Also note that $n_{x-1}$ represents the execution of penultimate loop iteration in sequential scenario. \figref{lemma36notations} depicts the global state used in the \lemref{containsF-conc}.
\begin{enumerate}
\item {\textbf{$S_{x-1}$}: Global state after the execution of $n_{x-1}$ = $read(n_{x-2}.next)$ at Line \ref{lin:con4}.}
\item {\textbf{$S_{x-1}'$}: Global state after the execution of $read(n_{x-1}.val)$ at Line \ref{lin:con3}.}
\item {\textbf{$S_{x}$}: Global state after the execution of $read(n_{x-1}.next)$ at Line \ref{lin:con4}.}
\item {\textbf{$S_{x}'$}: Global state after the execution of $read(n_{x}.val)$ at Line \ref{lin:con6}.}
\item {\textbf{$S_{x}''$}: Global state after the execution of $read(n_{x}.marked)$ at Line \ref{lin:con6}.}
\end{enumerate}

\begin{figure*}[!htbp]
\captionsetup{font=scriptsize}
\centerline{\scalebox{0.55}{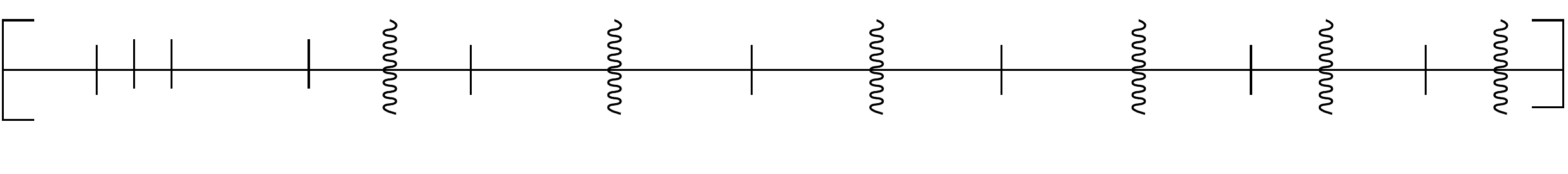}}
\caption{The global state representation for \lemref{containsF-conc} }
\label{fig:lemma36notations}
\end{figure*}

\begin{lemma}
\label{lem:containsF-conc}
If some \con(key) method returns false in $E^H$,  then
\begin{enumerate}[label=\ref{lem:containsF-conc}.\arabic*]
    \item \label{lem:containsF-conc-pre}
The $\node(key)$ is not present in the pre-state of $LP$ event of the method. Formally, $\langle \con(key, false) \Longrightarrow (\node(key) \notin \prees{\con(key, false)}) \rangle$. 
 \item \label{lem:containsF-conc-post}
The $\node(key)$ is also not present in the post-state of $LP$ event of the method. Formally, $\langle \con(key, false) \Longrightarrow (\node(key) \notin \prees{\con(key, false)}) \rangle$.
\end{enumerate}
\end{lemma}

\begin{proof}
\begin{itemize}
\item \textbf{\ref{lem:containsF-conc-pre}}:
There are following cases:
\begin{enumerate}
\item \textbf{Case 1}: \textit{key is not present in the Pre-State of read($n_x.val$ $\neq$ key) event at Line \ref{lin:con6} of Contains method, which is the LP of contains(key, false). We assume that there is no concurrent add from $S_1$ until $S_x'$.}

\begin{figure*}[!htbp]
\captionsetup{font=scriptsize}
\centerline{\scalebox{0.6}{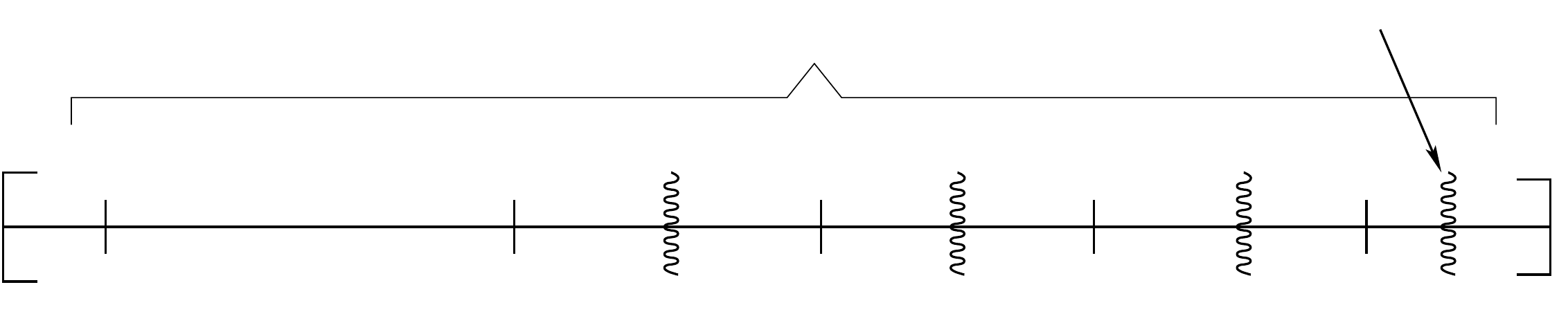}}

\caption{LP of $\con(key, false)$ with no successful concurrent $\add$ is at $read(n_x.val$ $\neq$ $key)$ at Line \ref{lin:con6}.}
\label{fig:con:case1}
\end{figure*}

\begin{enumerate}
\item \textbf{Given}: $(S_{x-1}'.\head \rightarrow^* S_{x-1}'.n_{x-1}) \land$ $(S_{x-1}'.n_{x-1}.marked = false)$ \\
\textbf{To Prove}: $\node(key)$ $\notin$ $S_x'.\abs$
\begin{equation} 
\label{lab:1a.Contains.mFalse.2}
 S_{x-1} . n_{x-1} . val \geq key \text{\hspace{.25cm} (Line \ref{lin:con3} of the \con method)}
\end{equation} 

\begin{equation} 
\label{lab:1a.Contains.mFalse.1}
S_{x-1}' . n_{x-1} . val < key \text{\hspace{.25cm} (Line \ref{lin:con4} of the Contains method)}
\end{equation}

\begin{equation} 
\label{lab:1a.Contains.mFalse.3}
 S_{x}' . n_{x} . val > key \text{\hspace{.25cm} (Line \ref{lin:con6} of the Contains method)}
\end{equation} 
\begin{equation} 
\label{lab:1a.Contains.mFalse.4}
 S_{x-1}' . n_{x-1} . val <  S_{x-1}' . n_{x-1} . next . val \text{\hspace{.25cm} (from Lemma \ref{lem:val-change})}
\end{equation} 
\begin{equation} 
\label{lab:1a.Contains.mFalse.5}
 S_{x} . n_{x-1} . next = S_{x} . n_{x} \text{\hspace{.25cm} (Line \ref{lin:con4} of the contains method)}
\end{equation} 
\begin{equation} 
\begin{split}
\label{lab:1a.Contains.mFalse.6}
 S_{x-1}' . n_{x-1} . val < S_{x-1}' . n_{x}.val \text{\hspace{.25cm} (from Equation \ref{lab:1a.Contains.mFalse.4} \& \ref{lab:1a.Contains.mFalse.5} \&  \obsref{node-val})}
\end{split}
\end{equation} 
Combining the equations \ref{lab:1a.Contains.mFalse.1},\ref{lab:1a.Contains.mFalse.3} \& \ref{lab:1a.Contains.mFalse.6} we have,
\begin{equation} 
\begin{split}
\label{lab:1a.Contains.mFalse.7}
 (S_{x-1}' . n_{x-1} .val < key < S_{x-1}'.n_{x}.val)  \implies (\node(key) \notin S_{x-1}'.\abs)
\end{split}
\end{equation} 
Now since no concurrent add on $key$ happens between $S_1$ until $S_x'$ we have that,
\begin{equation} 
\label{lab:1a.Contains.mFalse.8}
\langle \node(key) \notin S_{x}'.\abs \rangle
\end{equation} 

\item \textbf{Given}: $(S_{x-1}'.\head \rightarrow^* S_{x-1}'.n_{x-1})$ $\land$ $(S_{x-1}'.n_{x-1}.marked = true)$
\\
\textbf{To Prove}: $\node(key)$ $\notin$ $S_x'.\abs$ \\[0.2cm]
From given, we have that, 
\begin{equation}
\label{lab:1b.Contains.mFalse.9}
(n_{x-1} \notin S_{x-1}'.\abs ) 
\end{equation}
Let $n_i$ be the first unmarked node belonging to $S_{i}'.\abs$ while traversing the linked list of $n_1$, \dots, $n_i$, $n_{i+1}$, $n_{i+2}$, $\dots$, $n_{x-1}$, $n_x$, \dots nodes. Therefore,
\begin{equation}
\label{lab:1b.Contains.mFalse.10}
n_i \in S_{i}'.\abs
\end{equation}
In the worst case, $n_i$ could be the $\head$ node $n_1$.
\begin{equation}
\label{lab:1b.Contains.mFalse.11}
\text{We know that, } ( n_{i+1} \text{ to } n_{x-1}) \notin  (S_{i+1}'.\abs \text{ to } S_{x-1}'.\abs)
\end{equation}

In the linked list of $n_1$, \dots, $n_i$, $n_{i+1}$, $n_{i+2}$, $\dots$, $n_{x-1}$, $n_x$, \dots nodes, where $n_{i+1}$, $n_{i+2}$, $\dots$, $n_{x-1}$ are marked and consecutive, we can conclude (from Lemma \ref{lem:consec-mark}) that,

\begin{equation} 
\begin{split}
\label{lab:1b.Contains.mFalse.12}
 (S_{i+2}'.n_{i+1}.next = S_{i+2}'.n_{i+2}) \land (S_{i+2}'.n_{i+1}.marked) \land \\ (S_{i+2}'.n_{i+2}.marked)  \implies (n_{i+1}.marking <_E  n_{i+2}.marking) 
 \end{split}
\end{equation} 

In state $S_i'$, we know that $n_i.next = n_{i+1}$. Depending upon the status of node $n_{i+1}$ in $S_i'$, we have two possible situations:

\begin{enumerate}
\item $S_i'.n_{i+1}.unmarked$ \\[0.2cm]
Since we know that in $S_{i+1}': n_{i+1}.marked$. Thus we have that,
\begin{equation}
\begin{split}
Contains.read(n_i) <_E Remove.marking(n_{i+1}) <_E \\ Remove.marking(n_{i+2})
\end{split}
\end{equation}

\item $S_i'.n_{i+1}.marked$ \\[0.2cm]
We know that in $S_{i+1}': n_{i+1}.next = n_{i+2}$.
From Equation \ref{lab:1b.Contains.mFalse.12}, we can conclude that in $S_i':$ $n_{i+2}$ is $unmarked$.
From \lemref{next-next-unmarked},
\begin{equation}
\begin{split}
Remove1.unlock(n_{i+1}) <_E \\ Remove2.lock(n_{i+2}) <_E Remove2.marking(n_{i+2})
\end{split}
\end{equation}

Hence we can conclude that,

\begin{equation}
\label{lab:1b.Contains.mFalse.13}
\con.read(n_i) <_E n_{i+1}.marking <_E n_{i+2}.marking
\end{equation}
\end{enumerate}

Now consider a state $S_k$ in which $n_{x-1}$ is unmarked. From the \lemref{consec-mark} we have 
\begin{equation} 
\label{lab:nx-marked}
 n_{x-1}.marked <_E n_x.marked
 \end{equation} 

From the \obsref{consec-next} and from the Equation \ref{lab:nx-marked} we have,

\begin{equation} 
\label{lab:1b.Contains.mFalse.15}
 \exists S_k: (S_k.n_{x-1}.marked = false) \xRightarrow[]{\obsref{consec-next}}  S_k.n_{x}.marked = false 
 \end{equation} 
 Let us call the state immediately after the marking of $n_{x-1}$ as $S_k'$ as below:
 \begin{figure*}[!htbp]
 \centering
 \captionsetup{font=scriptsize}
\centerline{\scalebox{0.65}{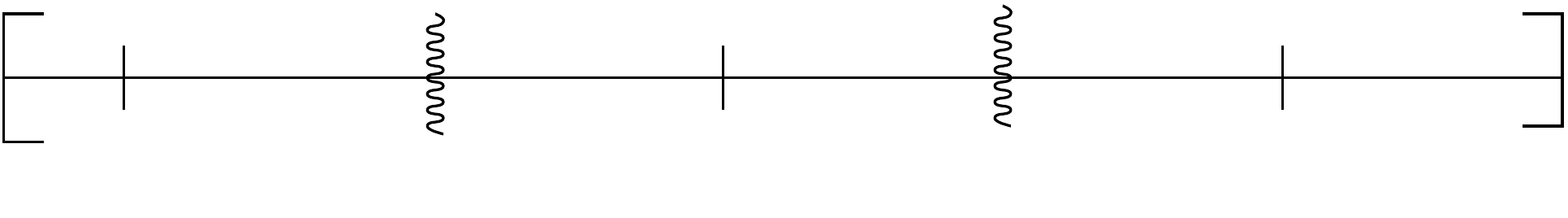}}
\caption{$\con(key, false)$ with no successful concurrent $\add$ on key. $S_{x-1}'.n_{x-1}.marked = true$ \\ and $S_{x-1}'.n_{x}.marked = true$ and $\node(key)$ $\notin$ $S_x'.\abs$ at Line \ref{lin:con6}}
\end{figure*}
 
Combining Observation \ref{lem:node-marknext} and \ref{lem:node-mark}, we know that,
\begin{equation} 
\label{lab:1b.Contains.mFalse.16}
S_k'.n_{x-1}.next = S_k'.n_x
\end{equation} 
Also since $n_{x-1}$.marking is the only event between $S_k$ and $S_k'$, we can say that,
\begin{equation} 
\label{lab:1b.Contains.mFalse.16.1}
S_k.n_{x-1}.next = S_k.n_x
\end{equation} 

Also by observing the code of \con method, we have the following:
\begin{equation} 
\label{lab:1b.Contains.mFalse.17}
S_{x-1}' . n_{x-1} . val < key \text{\hspace{.25cm} (Line \ref{lin:con4} of the $\con$ method)}
\end{equation}
\begin{equation} 
\label{lab:1b.Contains.mFalse.18}
 S_{x} . n_{x} . val \geq key \text{\hspace{.25cm} (Line \ref{lin:con3} of the $\con$ method)}
\end{equation} 
\begin{equation} 
\label{lab:1b.Contains.mFalse.19}
 S_{x}' . n_{x} . val > key \text{\hspace{.25cm} (Line \ref{lin:con6} of the $\con$ method)}
\end{equation} 

\begin{equation} 
\label{lab:nx-nx-1-marked}
 (\neg S_{k}.n_{x-1}.marked) \land (\neg S_{k}.n_{x}.marked)
 \text{\hspace{.25cm} (by the \lemref{consec-mark})}
 \end{equation}

Combining the equations \ref{lab:1b.Contains.mFalse.15},\ref{lab:1b.Contains.mFalse.16.1}, \ref{lab:1b.Contains.mFalse.17} \& \ref{lab:1b.Contains.mFalse.19}, \ref{lab:nx-nx-1-marked} and \obsref{node-val} and \ref{lem:node-mark},

\begin{equation} 
\begin{split}
\label{lab:1b.Contains.mFalse.20}
 (S_{k} . n_{x-1} .val < key < S_{k}.n_{x}.val) \land (\neg S_{k}.n_{x-1}.marked) \land \\ (\neg S_{k}.n_{x}.marked) \land (S_{k}.n_{x-1}.next = S_k.n_x)  \xRightarrow[]{\lemref{val-abs}} \\ (\node(key) \notin S_{k}.\abs)
\end{split}
\end{equation} 
Now since no concurrent $\add$ happens between $S_1$ and $S_x'$ we have that,
\begin{equation} 
\label{lab:1b.Contains.mFalse.21}
 \node(key) \notin S_{x}'.\abs
\end{equation}

 \end{enumerate}

\item \textbf{Case 2}: \textit{key is present, but marked in the Pre-State of \emph{read(n.marked)} event at Line \ref{lin:con6} of \con method, which is the \lp of \con(key, false). We assume that there is no concurrent \add from $S_1$ until $S_x'$.}

\begin{figure*}[!htbp]
\captionsetup{font=scriptsize}
\centerline{\scalebox{0.55}{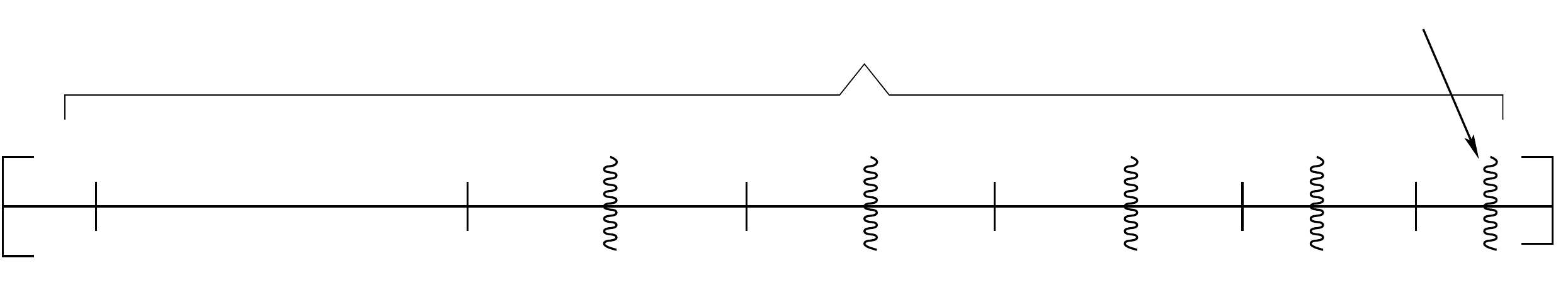}}
\caption{LP of $\con(key, false)$ with no successful concurrent $\add$ is at $read(n.val$ = $key)$ at Line \ref{lin:con6}}
\end{figure*}

\begin{enumerate}
\item \textbf{Given:} $S_{x-1}'.n_{x-1}.marked = false \land S_{x-1}'.n_{x}.marked = true$\\
\textbf{To Prove}: $\node(key)$ $\notin$ $S_x''.\abs$
\begin{equation} 
\label{lab:2a.Contains.mFalse.1}
S_{x-1}' . n_{x-1} . val < key \text{\hspace{.25cm} (Line \ref{lin:con4} of the $\con$ method)}
\end{equation}
\begin{equation} 
\label{lab:2a.Contains.mFalse.2}
 S_{x} . n_{x} . val \geq key \text{\hspace{.25cm} (Line \ref{lin:con3} of the $\con$ method)}
\end{equation} 
\begin{equation} 
\label{lab:2a.Contains.mFalse.3}
 S_{x}' . n_{x} . val = key \text{\hspace{.25cm} (Line \ref{lin:con6} of the $\con$ method)}
\end{equation} 
\begin{equation} 
\label{lab:2a.Contains.mFalse.3.1}
 S_{x-1}' . n_{x} . marked = true \text{\hspace{.25cm} (Given)}
\end{equation} 
\begin{equation} 
\label{lab:2a.Contains.mFalse.3.2}
 S_{x}'' . n_{x} . marked = true \text{\hspace{.25cm} (From Observation \ref{lem:node-mark})}
\end{equation} 
\begin{equation} 
\label{lab:2a.Contains.mFalse.4}
 S_{x-1}' . n_{x-1} . val <  S_{x-1}' . n_{x-1} . next . val \text{\hspace{.25cm} (from Lemma \ref{lem:val-change})}
\end{equation} 
\begin{equation} 
\label{lab:2a.Contains.mFalse.5}
 S_{x-1}' . n_{x-1} . next = S_{x-1}' . n_{x} \text{\hspace{.15cm} (Line \ref{lin:con4} of the $\con$ method)}
\end{equation} 
\begin{equation} 
\label{lab:2a.Contains.mFalse.6}
 S_{x-1}' . n_{x-1} . val < S_{x-1}' . n_{x}.val \text{\hspace{.25cm} (from Equation \ref{lab:2a.Contains.mFalse.4} \& \ref{lab:2a.Contains.mFalse.5})}
\end{equation} 
Combining the equations \ref{lab:2a.Contains.mFalse.1},\ref{lab:2a.Contains.mFalse.3} \& \ref{lab:2a.Contains.mFalse.6} and \obsref{node-val},
\begin{equation} 
\begin{split}
\label{lab:2a.Contains.mFalse.7}
 (S_{x-1}' . n_{x-1} .val < (key = S_{x-1}'.n_{x}.val)) \land (key \neq S_{x-1}' . n_{x-1} .val) \\ \land (S_{x-1}' . n_{x} . marked)  \implies (\node(key) \notin S_{x-1}'.\abs)
\end{split}
\end{equation} 
Now since no concurrent $\add$ happens between $S_1$ and $S_x''$ we have that,
\begin{equation} 
\label{lab:2a.Contains.mFalse.8}
 \node(key) \notin S_{x}''.\abs
\end{equation}

\item \textbf{Given:} $S_{x-1}'.n_{x-1}.marked = true \land S_{x-1}'.n_{x}.marked = true$\\
\textbf{To Prove}: $\node(key)$ $\notin$ $S_x''.\abs$ \\
From given, we have that, 
\begin{equation}
\label{lab:2b.Contains.mFalse.1}
(n_{x-1} \notin S_{x-1}'.\abs ) \land (n_{x} \notin S_{x-1}'.\abs) 
\end{equation}
From $S_{x-1}'.n_{x-1}$ we backtrack the nodes until we find the first node $n_i$ belonging to $S_{x-1}'.\abs$. Therefore,
\begin{equation}
\label{lab:2b.Contains.mFalse.2}
n_i \in S_{x-1}'.\abs
\end{equation}
In the worst case, $S_{x-1}'.n_i$ could be the $Head$ node.

\begin{equation}
\label{lab:2b.Contains.mFalse.3}
\text{We know that, \hspace{1cm}   } ( n_{i+1} \text{ to } n_{x}) \notin  (S_{x-1}'.\abs)
\end{equation}

In the linked list of $n_1$, $n_{i+1}$, $n_{i+2}$, $\dots$, $n_{x-1}$, $n_x$ nodes, where $n_{i+1}$, $n_{i+2}$, $\dots$, $n_{x}$ are marked and consecutive, we can conclude (from Lemma \ref{lem:consec-mark}) that,
\begin{equation}
\begin{split}
\label{lab:2b.Contains.mFalse.5}
\con.read(n_1) <_E \con.read(n_i) <_E \\ \rem.unlock(n_{i+1}) <_E n_{i+2}.marking <_E \\ n_{i+3}.marking  \ldots <_E n_{x-1}.marking  <_E n_{x}.marking
\end{split}
\end{equation}

This implies that marking of $n_{i+1}$to $n_{x}$ completes after $\con(key, false)$ started.
\begin{equation} 
\label{lab:2b.Contains.mFalse.6}
 \con.read(n_1) <_E n_{x-1}.marking
\end{equation} 

Now consider a state $S_{k+1}$ in which $n_{x-1}$ was observed to be unmarked. Let us call the state immediately after the marking of $n_x$ as $S_{k+1}'$ as follows: 
 
 \begin{figure*}[!htbp]
 \centering
 \captionsetup{font=scriptsize}
\centerline{\scalebox{0.6}{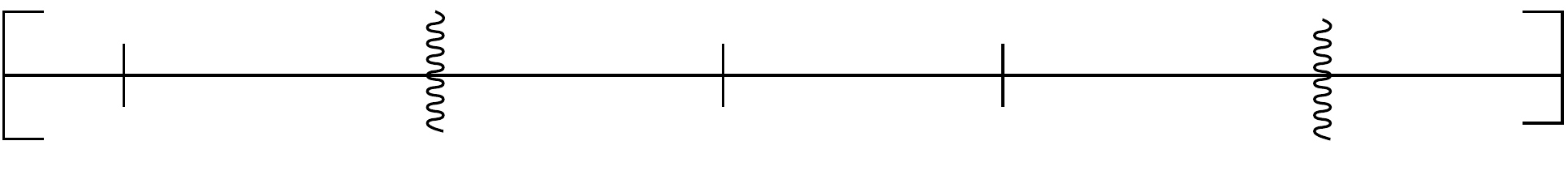}}
\caption{$\con(key, false)$ with no successful concurrent $\add$. $S_{k+1}'.n_{x}.marked = true$,  $\node(key)$ $\notin$ $S_x'.\abs$ at Line \ref{lin:con6} \\LP of $\con(key, false)$ with no successful concurrent $\add$ is at $read(n.val$ = $key)$ at Line \ref{lin:con6}}
\end{figure*}

Since a marked node remains marked (from Observation \ref{lem:node-mark}),
\begin{equation} 
\label{lab:2b.Contains.mFalse.8}
 S_{k+1}'.n_{x}.marked \implies S_{x}''.n_{x}.marked 
 \end{equation} 
 
Also by observing the code of $\con$ method, we have the following:
\begin{equation} 
\label{lab:2b.Contains.mFalse.9}
S_{x-1}' . n_{x-1} . val < key \text{\hspace{.25cm} (Line \ref{lin:con4} of the \con method)}
\end{equation}
\begin{equation} 
\label{lab:2b.Contains.mFalse.10}
 S_{x} . n_{x} . val \geq key \text{\hspace{.25cm} (Line \ref{lin:con3} of the \con method)}
\end{equation} 
\begin{equation} 
\label{lab:2b.Contains.mFalse.10.5}
 S_{x}' . n_{x} . val = key \text{\hspace{.25cm} (Line \ref{lin:con6} of the \con method)}
\end{equation} 
Combining the equations \ref{lab:2b.Contains.mFalse.10.5},\ref{lab:2b.Contains.mFalse.8} \& and from Observation \ref{obs:node-val} \& \ref{lem:node-mark},
\begin{equation} 
\begin{split}
\label{lab:2b.Contains.mFalse.11}
 (S_{k+1}' . n_{x} .val = key) \land (S_{k+1}' . n_{x} . marked) \implies (\node(key) \notin S_{k+1}'.\abs)
\end{split}
\end{equation} 
Now since no concurrent $\add$ happens between $S_1$ and $S_x''$ we have that,
\begin{equation} 
\label{lab:2b.Contains.mFalse.12}
 \node(key) \notin S_{x}''.\abs
\end{equation} 
 
\end{enumerate}

\item \textbf{Case 3}: \textit{key is not present in the Pre-State of the \lp of \con(key, false) method. \lp is a dummy event inserted just before the \lp of the \add. We assume that there exists a concurrent \add from $S_1$ until $S_x'$.}

\begin{figure*}[!htbp]
\captionsetup{font=scriptsize}
 \centerline{\scalebox{0.6}{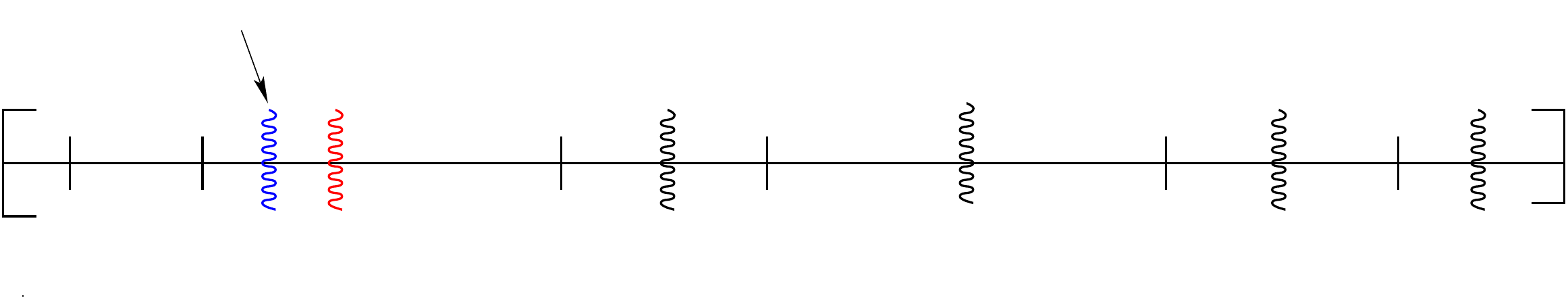}}
\caption{\lp of $\con(key, false)$ with successful concurrent $\add$ is at $read(n.val$ = $key)$}
\label{case3}
\end{figure*}
\textbf{To prove:} $\node(key)$ $\notin$ $S_{dummy}.\abs$\\[0.1cm]
From Lemma \ref{lem:addT-conc}, we know that if add returns true, then $\node(key)$ does not belong to the $\abs$ in the pre-state of the $\lp$ of add method. We add a dummy event just before this $\lp$ event of add method as in Figure \ref{case3}. 

\begin{equation} 
\label{lab:3.Contains.mFalse}
 \node(key) \notin S_{dummy}.\abs 
\end{equation}

\item \textbf{Case 4}: \textit{key is present, but marked in the Pre-State of the \lp of \con(key, false) method. \lp is a dummy event inserted just before the \lp of the \add. We assume that there exists a concurrent \add from $S_1$ until $S_x'$.}

\begin{figure*}[!htbp]
\captionsetup{font=scriptsize}
\centerline{\scalebox{0.55}{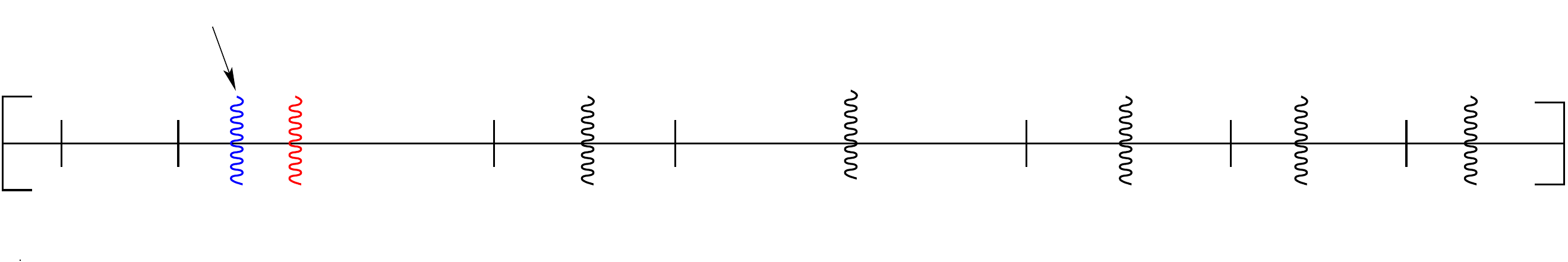}}
\caption{\lp of $\con(key, false)$ with successful concurrent $\add$ is at $read(n.val$ = $key)$ at Line \ref{lin:con6}}
\label{case4}
\end{figure*}

\textbf{To prove:} $\node(key)$ $\notin$ $S_{dummy}.\abs$\\[0.1cm]
From Lemma \ref{lem:addT-conc}, we know that if \add returns true, then $\node(key)$ does not belong to the $\abs$ in the pre-state of the $\lp$ of \add method. We add a dummy event just before this $\lp$ event of \add method as in Figure \ref{case4}. 

\begin{equation} 
\label{lab:4.Contains.mFalse}
 \node(key) \notin S_{dummy}.\abs 
 \end{equation}
 
\end{enumerate} 
\item \textbf{\ref{lem:containsF-conc-post}}:
From the Lemma \ref{lem:containsF-conc-pre} we get that \emph{\node(key)} is not present in the pre-state of the \lp event. This \lp event $(read(n.val) \neq key) \vee (read(n.marked))$ in the \lineref{con6} does not change the $S.\abs$. Now after execution of the \lp event the $\node(key)$ will also not present in the $S'.\abs$, where $S'$ is the post-state of the \lp event of the method. Hence, $\langle \con(key, false) \Longrightarrow (\node(key) \notin (\postes{\con(key, flase)}) \rangle$.
\end{itemize}
\end{proof}

\begin{lemma}
\label{lem:pre-ret}
\textbf{lazy-list Specific Equivalence:}
Consider a concurrent history $H$ and a sequential history $\spl{S}$. Let $m_x, m_y$ be \mth{s} in $H$ and $\spl{S}$ respectively. Suppose the following are true (1) The \abs in the pre-state of $m_x$'s \lp in $H$ is the same as the \abs in the pre-state of $m_y$ in $\spl{S}$;  (2) The \inv events of $m_x$ and $m_y$ are the same. Then (1) the $\rsp$ event of $m_x$ in $H$ must be same as $\rsp$ event of $m_y$ in $\spl{S}$; (2) The \abs in the post-state of $m_x$'s \lp in $H$ must be the same as the \abs in the post-state of $m_y$ in $\spl{S}$. Formally, $\langle \forall m_x \in \mths{E^{H}}, \forall m_y \in \mths{E^{\spl{S}}}: (\prees{m_x} = \prems{m_y}) \wedge (\inves{m_x} = \invms{m_y}) \Longrightarrow \\ (\postes{m_x} = \postms{m_y}) \wedge \\ (\retes{m_x} = \retms{m_y}) \rangle$.


\end{lemma}

\begin{proof}

 Let us prove by contradiction.
So we assume that, 
\begin{equation}
 \label{lab:eq:pre-inv-resp}
     \begin{split}
         \langle (\prees{m_x} = \prems{m_y}) \wedge \\ (\inves{m_x} =  \invms{m_y}) \Longrightarrow  (\retes{m_x} \neq \retms{m_y}) \rangle
     \end{split}
 \end{equation}

We have the following cases that $\inves{m_x}$ is invocation of either of these methods:

\begin{enumerate}
\item \textbf{$m_x.inv$ is \add(key) Method:}
\begin{itemize}
\item \textbf{$m_x.resp$ = true:} Given that the method $m_x.resp$ which is \emph{\add(key)} returns $true$, we know that from the Lemma \ref{lem:addT-conc}, \textit{\node(key)} $\notin$ $\prees{\add(key, true)}$. But since from assumption in equation \ref{lab:eq:pre-inv-resp}, $(\retes{m_x} \neq \retms{m_y})$, $\retms{m_y}$ is false. However, from the Observation \ref{obs:addT-seq}, if $\node(key)$ $\notin$ pre-state of $\lp$ of $\add$ method, then the $\add(key, true)$ method must return $true$ in $E^{\spl{S}}$. This is a contradiction.

\item \textbf{$m_x.resp$ = false:} Given that the method $m_x.resp$ which is \emph{\add(key)} returns $false$, we know that from the Lemma \ref{lem:addF-conc}, \textit{\node(key)} $\in$ $\prees{\add(key, false)}$. But since from assumption in equation \ref{lab:eq:pre-inv-resp}, $(\retes{m_x} \neq \retms{m_y})$, $\retms{m_y}$ is false. However, from the Observation \ref{obs:addF-seq}, if $\node(key)$ $\in$ pre-state of $\lp$ of $\add$ method, then the $\add(key, false)$ method must return $false$ in $E^{\spl{S}}$. This is a contradiction.


\end{itemize}

\item \textbf{$m_x.inv$ is \rem(key) Method:}
\begin{itemize}
\item \textbf{$m_x.resp$ = true:} Given that the method $m_x.resp$ which is \emph{\rem(key)} returns $true$, we know that from the Lemma \ref{lem:removeT-conc}, \textit{\node(key)} $\in$ $\prees{\rem(key, true)}$. But since from assumption in equation \ref{lab:eq:pre-inv-resp}, $(\retes{m_x} \neq \retms{m_y})$, $\retms{m_y}$ is false. However, from the Observation \ref{obs:removeT-seq}, if $\node(key)$ $\in$ pre-state of $\lp$ of $\rem$ method, then the $\rem(key, true)$ method must return $true$ in $E^{\spl{S}}$. This is a contradiction.

\item \textbf{$m_x.resp$ = false:} Given that the method $m_x.resp$ which is \emph{\rem(key)} returns $false$, we know that from the Lemma \ref{lem:removeF-conc}, \textit{\node(key)} $\notin$ $\prees{\rem(key, false)}$. But since from assumption in equation \ref{lab:eq:pre-inv-resp}, $(\retes{m_x} \neq \retms{m_y})$, $\retms{m_y}$ is false. However, from the Observation \ref{obs:removeF-seq}, if $\node(key)$ $\notin$ pre-state of $\lp$ of $\rem$ method, then the $\rem(key, false)$ method must return $false$ in $E^{\spl{S}}$. This is a contradiction.
\end{itemize}

\item \textbf{$m_x.inv$ is \con(key) Method:}
\begin{itemize}
\item \textbf{$m_x.resp$ = true:} Given that the method $m_x.resp$ which is \emph{\con(key)} returns $true$, we know that from the Lemma \ref{lem:containsT-conc}, \textit{\node(key)} $\in$ $\prees{\con(key, true)}$. But since from assumption in equation \ref{lab:eq:pre-inv-resp}, $(\retes{m_x} \neq \retms{m_y})$, $\retms{m_y}$ is false. However, from the Observation \ref{obs:containT-seq}, if $\node(key)$ $\in$ pre-state of $\lp$ of $\con$ method, then the $\con(key, true)$ method must return $true$ in $E^{\spl{S}}$. This is a contradiction.

 
\item \textbf{$m_x.resp$ = false:} Given that the method $m_x.resp$ which is \emph{\con(key)} returns $false$, we know that from the Lemma \ref{lem:containsF-conc}, \textit{\node(key)} $\notin$ \\ $\prees{\con(key, false)}$. But since from assumption in equation \ref{lab:eq:pre-inv-resp}, $(\retes{m_x} \neq \retms{m_y})$, $\retms{m_y}$ is false. However, from the Observation \ref{obs:containF-seq}, if $\node(key)$ $\notin$ pre-state of $\lp$ of $\con$ method, then the $\con(key, false)$ method must return $false$ in $E^{\spl{S}}$. This is a contradiction.



\end{itemize}
\end{enumerate}
Thus we conclude that the $\rsp$ event of $m_x$ in $H$ must be same as $\rsp$ event of $m_y$ in $\spl{S}$. Formally, $\langle \retes{m_x} = \retms{m_y} \rangle$. 
\end{proof}

\begin{lemma}
\label{lem:hist-lin-lazy}
All histories ${H}$ generated by the Lazy List  are \lble.
\end{lemma}

\begin{proof}
 Lemma follows based on the \lemref{pre-ret}, \lemref{conce-seqe-pre} and \lemref{ret-lin}.
\end{proof}

\subsection{Hand-over-Hand Locking List}
\label{subsec:con-hoh}
In this section we define the fine list data structure. It is implemented as a collection of a set of nodes. This is a linked list of node shown in the $Node$ is a class and it has three fields, the $val$ field is the key value of the node. The nodes are sorted in order of the $val$ field. This helps efficiently detect when a \textit{\hnode} is absent in the list. The $next$ field is a reference to the next \textit{\hnode} in the list. The $lock$ field is for ensuring access to a shared \textit{\hnode} happens in a mutually exclusion manner. We say a thread  acquires a lock and releases the lock  when it executes a \textit{lock.acquire()} and \textit{lock.release()} method call respectively. Each thread acquires lock in a \emph{\hoh} order. We assume the $next$ field of the $\hnode$ is atomic.

\begin{tcolorbox}
\begin{verbatim}
class Node{
    int val;     // actual key of node
    Node next;   // next Node in list
    Lock lock;   //synchronizes individual Node
    / *
    For the sentinel Node the @param key should be min or max 
    int value and for the usual Node @param key val be the
    actual element in list
    */
    Node(int key){  
        val = key;
        next = null;
        lock = new Lock();
    }
};
\end{verbatim}
\end{tcolorbox}
\ignore{
Given a global state $S$, we define a few structures and notations for the \hoh as follows: 
\begin{enumerate}
	\item We denote node, say $n$, as a $\hnode$ class object. 
    \item $\hnodes{S}$ as a set of $\hnode$ class objects that have been created so far in $S$. Its structure is defined above. Each $\hnode$ object $n$ in $\hnodes{S}$ is initialized with $key$, $next$ to $null$.
    
    \item $S.\hhead$ is a $\hnode$ class object (called sentinel head node), which is initialized with a val $\text{-}\infty$. This sentinel node is never removed from the list. Similarly, $S.\htail$ is a $\hnode$ class object (called sentinel tail  node), which is initialized with a val $\text{+}\infty$. This sentinel node is never removed from the list.
   
    \item The contents \& the status of a $\hnode$ $n$ keeps changing with different global states. For a global state $S$, we denote $S.n$ as its current state and the individual fields of $n$ in $S$ as $S.n.val, S.n.next, ...$ etc.  
    
\end{enumerate}
 }

\noindent We define concurrent set $S$, which is dynamically being modified by a fixed set of concurrent threads. In this setting, threads may perform insertion or deletion of nodes to the set. We used \hoh based set algorithm based on \cite[Chap 9]{MauriceNir}. We assume that all the nodes have unique identification key.

\begin{table}[!htbp]
 \centering
 \scriptsize
 \captionsetup{font=scriptsize}
  \captionof{table}{Sequential Specification of the \hoh } \label{tab:seq-spe-hoh} 
\begin{tabular}{ |m{2.5cm}|m{1cm}|m{2.5cm}|m{2.5cm}| } 

\hline
  \textbf{Method} & \textbf{Return Value} & \textbf{Pre-state}($S$: global state)  & \textbf{Post-state}( $S'$: future state of $S$ such that $S \sqsubset S'$)  \\ 
  
  \hline
	$\hadd(n)$  & $true$ &$ S: \langle n \notin S.\habs \rangle $ &$ S':\langle n \in S'.\habs \rangle$	 \\
\hline 
	$\hadd (n)$  & $false$ &$ S: \langle n \in S.\habs \rangle $ & $ S':\langle n \in S'.\habs \rangle$	
	\\
\hline 
	$\hrem(n)$  & $true$ &$ S: \langle n \in S.\habs \rangle $ & $ S':\langle n \notin S'.\habs \rangle$	
	\\
	\hline 
$\hrem(n)$  & $false$ &$ S: \langle n \notin S.\habs \rangle $ & $ S':\langle n \notin S'.\habs \rangle$
\\
\hline 
$\hcon(n)$  & $true$ &$ S: \langle n \in S.\habs \rangle$ &$ S':\langle n \in S'.\habs \rangle$	\\
	\hline 
	$\hcon(n)$  & $false$ &$ S: \langle n \notin S.\habs \rangle$ & $ S':\langle n \notin S'.\habs \rangle$	\\
	\hline 

\end{tabular}
\end{table}
\subsubsection{\textbf{Methods Exported \& Sequential Specification}}
\label{sec:con-hoh:methods}
\noindent \\ In this section, we describe the \mth{s} exported by the \hoh data structure.
\noindent 
\begin{enumerate}
\item The $\hadd(n)$ method adds a node $n$ to the list, returns $true$ if the node is not present in the list else it returns $false$. This follows directly from our assumption that all the nodes are assigned distinct keys. 

\item The $\hrem(n)$ method deletes a node $n$ from the list, if it is present and returns $true$. If the node is not in the list earlier, it returns $false$. 
\item The $\hcon(n)$ returns $true$, if the list contains the node $n$; otherwise returns $false$.
\end{enumerate}
 
  \noindent Table \ref{tab:seq-spe-hoh} shows the sequential specification, as the name suggests shows the behaviour of the list when all the \mth{s} are invoked sequentially. We defined
 each method formally in any given global state $S$ before the execution of the method and future state $S'$ after executing it sequentially. The \textit{Pre-state} is the shared state before $\inv$ event and the \textit{Post-state} is also the shared state just after the $\rsp$ event of a method, which is depicted in the Figure \ref{fig:exec/hist}. 
 


\noindent
All the fields in the structure are declared atomic. This ensures that operations on these variables happen atomically. In the context of a particular application, the node structure can be easily modified to carry useful data (like weights etc).\\

\noindent \textbf{Notations used in PseudoCode:}\\[0.2cm]
$\downarrow$, $\uparrow$ denote input and output arguments to each method respectively. The shared memory is accessed only by invoking explicit \emph{read()} and \emph{write()} methods. The $flag$ is a local variable which returns the status of each operation. We use nodes $n_1$, $n_2$, $n$ to represent $\hnode$ references.

 \begin{algorithm}[!htb]
 \captionsetup{font=scriptsize}
    \caption{\hloct Method: Takes $key$ as input and returns the corresponding pair of neighboring $\node$ $\langle n_1, n_2 \rangle$. Initially $n_1$ and $n_2$ are set to $null$.}
    \label{alg:hlocate}
	\begin{algorithmic}[1]
	\scriptsize
	\algrestore{con}
		\Procedure{\hloct ($key\downarrow, n_1\uparrow, n_2\uparrow$)}{}
		\State {$lock.acquire(\hhead)$;} \label{lin:hloc2}
    	\State{$\hnode$ $n_1 = \hhead;$} \label{lin:hloc3}
    	\State{$\hnode$ $n_2 = n_1.next;$}\label{lin:hloc4}
    	\State {$lock.acquire(n_2)$;}  \label{lin:hloc5}
    	\While{($read(n_2.val) < key)$} \label{lin:hloc6}
    	    \State {$lock.release(n_1)$;} \label{lin:hloc7}
    	    \State{$ n_1 \gets n_2;$} \label{lin:hloc8}
    	    \State{$n_2 \gets n_2.next$} \label{lin:hloc9}
    	    \State {$lock.acquire(n_2)$;}  \label{lin:hloc10}
    	\EndWhile\label{lin:hloc11}
    	\EndProcedure
		\algstore{hlocate}
	\end{algorithmic}
  \end{algorithm}
 \begin{algorithm}[!htb]
   \captionsetup{font=scriptsize}
	\caption{\hcon Method: Returns $true$ if $key$ is part of the set and returns $false$ otherwise.}
   		\label{alg:hcontains}
		\scriptsize
	\begin{algorithmic}[1]
	\algrestore{hlocate}
		\Procedure{\hcon ($key\downarrow, flag\uparrow$)}{}
		\State{$\hloct(key\downarrow, n_1 \uparrow, n_2\uparrow)$;}\label{lin:hcon2} 
		\If {$(read(n_2.val) = key) $}  \label{lin:hcon3}
	    \State {$flag$ $\gets$ $true$;}\label{lin:hcon4}
		\Else
	    \State {$flag$ $\gets$ $false$;} \label{lin:hcon5}
		\EndIf
		\State {$lock.release(n_1)$;}
    	\State {$lock.release(n_2)$;}
		\State {$return$;}
		\EndProcedure
		\algstore{hcon}
	\end{algorithmic}
  \end{algorithm}

 \begin{algorithm}[!htb]
    \captionsetup{font=scriptsize}
	\caption{\hadd Method: $key$ gets added to the list if it is not already part of the list. Returns $true$ on successful add and returns $false$ otherwise.	}
		\label{alg:hadd}
   	\begin{algorithmic}[1]
    	\scriptsize
    	\algrestore{hcon}
    		\Procedure{\hadd ($key\downarrow, flag \uparrow$)}{}
    		\State{$\hloct(key\downarrow, n_1 \uparrow, n_2\uparrow)$;}   \label{lin:hadd2}  
    		\If {$(read(n_2.val) \neq key$)} \label{lin:hadd3}
    		\State {$write(n_3, \text{new \hnode}(key))$;} \label{lin:hadd4}
    		\State {$write(n_3.next, n_2)$;} \label{lin:hadd5}
    		\State {$write(n_1.next, n_3)$;} \label{lin:hadd6}
    		\State {$flag$ $\gets$ $true$;}
    		\Else
    		\State {$flag$ $\gets$ $false$;}    
    		\EndIf
    		\State {$lock.release(n_1)$;}
    		\State {$lock.release(n_2)$;}
    		\State{$return$;}
    		\EndProcedure
    		\algstore{hadd}
    	\end{algorithmic}
  \end{algorithm}
 \begin{algorithm}[!htb]
  \captionsetup{font=scriptsize}
	\caption{\hrem Method: $key$ gets removed from the list if it is already part of the list. Returns $true$ on successful remove otherwise returns $false$.}
    \label{alg:hremove}
	\begin{algorithmic}[1]
		\scriptsize
	\algrestore{hadd}
		\Procedure{\hrem ($key\downarrow, flag\uparrow$)}{}
		\State{$\hloct(key\downarrow, n_1 \uparrow, n_2\uparrow)$;}\label{lin:hrem2} 
		\If {$(read(n_2.val) = key)$} \label{lin:hrem3} 
		\State {$write(n_1.next, n_2.next)$;} \label{lin:hrem5}
		\State {$flag$ $\gets$ $true$;}\label{lin:hrem8}
		\Else
	    \State {$flag$ $\gets$ $false$;}
		\EndIf
		\State {$lock.release(n_1)$;} \label{lin:hrem6}
    	\State {$lock.release(n_2)$;}
		\State {$return$;}
		\EndProcedure
	\end{algorithmic}
  \end{algorithm}

\subsubsection{\textbf{\emph{Working of the methods of \hoh}}} 
\noindent \\
We define all methods like \hoh  used in the \cite[Chap 9]{MauriceNir} with some modification. We add a new \emph{\hloct} method, which helps to locate the location of the \emph{key} in the list like lazy list defined in the \subsecref{con-lazy-list}. This \emph{\hloct} method takes $key$ as input and returns the corresponding pair of neighboring $\node$ $\langle n_1, n_2 \rangle$ and both these nodes are locked and reachable from the \hhead. Initially $n_1$ and $n_2$ are set to $null$. \\[0.1cm]
\noindent
\textbf{Working of the \hadd(key) method:}
 When a thread wants to add a node to the list, it invokes \hloct in the \lineref{hadd2}. The \hloct traverses the list from \emph{\hhead} by acquiring locks both predecessor and successor nodes until it finds a node with its key greater than or equal to $key$, say $ncurr$ and it's predecessor node, say $npred$. When \hloct method returns, both the nodes are locked. Then it checks if $read(ncurr.val) \neq key$ is $true$(\lineref{hadd3}), then the thread adds the new $\node(key)$ between $npred$ and $ncurr$ in the list from the Line \ref{lin:hadd4}-\ref{lin:hadd6} and returns true after unlocking the nodes. If the \emph{key} is already present in the list, it returns $false$ by unlocking the locked nodes. This is described in Algorithm \ref{alg:hadd}.\\[0.1cm]
\noindent 
\textbf{Working of the \hrem(key) method:}
 When a thread wants to delete a node from the list, it invokes \hloct in the \lineref{hrem2}. The \hloct traverses the list from \emph{\hhead} by acquiring locks both predecessor and successor nodes until it finds a node with its key greater than or equal to $key$, say $ncurr$ and it's predecessor node, say $npred$. When \hloct method returns, both the nodes are locked. Then it checks if $(read(n_2.val) = key)$ is $true$(\lineref{hrem3}), if it is then the thread removes the $ncurr$ by changing the $next$ pointer of $npred$ to $ncurr.next$ in the \lineref{hrem5}. If the \emph{key} is not present in the list, it returns  $false$ by unlocking the locked nodes. This is described in Algorithm \ref{alg:hremove}.\\[0.1cm]
\noindent \textbf{Working of the \hcon() method:}
When a thread wants to search a node in the list, it invokes \hloct in the \lineref{hcon2}. The \hloct traverses the list from \emph{\hhead} by acquiring locks both predecessor and successor nodes until it finds a node with its key greater than or equal to $key$, say $ncurr$ and it's predecessor node, say $npred$. When \hloct method returns, both the nodes are locked. Then it checks if $(read(n_2.val) = key)$ is $true$(\lineref{hcon3}), if it is then the thread returns $true$ in the \lineref{hcon4}. If the $key$ is not present in the list, it returns $false$ in the Line \ref{lin:hcon5}. This is described in Algorithm \ref{alg:hcontains}.\\[0.1cm]
\subsubsection{\textbf{The LPs of the \hoh}}
\label{subsec:lps-hoh}
\noindent \\ Here, we list the linearization points (\lp{s}) of each method of \hoh. Each method of the list can return either $true$ or $false$. So, we define the $\lp$ for six methods:

\begin{enumerate}
\item $\hadd(key, true)$: $write(n_1.next, n_3)$ in Line \ref{lin:hadd6} of $\hadd$ method.
\item $\hadd(key, false)$: $read(n_2.val)$ in Line \ref{lin:hadd3} of $\hadd$ method. 
\item $\hrem(key, true)$: $write(n_1.next, n_2.next)$ in Line \ref{lin:hrem5} of $\hrem$ method.
\item $\hrem(key, false)$: $(read(n_2.val))$ in Line \ref{lin:hrem3} of $\hrem$ method.
\item $\hcon(key, true)$: $read(n.val)$ in Line \ref{lin:hcon3} of $\hcon$ method.
\item $\hcon(key, false)$:$read(n.val)$ in Line \ref{lin:hcon3} of  $\hcon$ method.
\end{enumerate}

\subsubsection{\textbf{HoH-Locking-List Proof}}
\noindent \\ In this subsection, we describe the lemmas to prove the correctness of concurrent \hoh structure. 

 \noindent Having defined a few notions on $S$, we now define the notion of an abstract set, $\habs$ for a global state $S$ which we will use for guiding us in correctness of our \mth{s} and it is defined below:

 \begin{definition}
\label{def:habs}
	$S.\habs \equiv \{n | (n \in \hnodes{S}) \land (S.\hhead \rightarrow^* S.n) \}$.
\end{definition}

\noindent This definition of $\habs$ captures the set of all nodes of $\habs$ for the global state $S$. It consists of all the $\hnode{s}$ that are reachable from $S.\hhead$. 

\begin{observation}
\label{lem:hoh-node-forever}
	Consider a global state $S$ which has a node $n$. Then in any future state $S'$ of $S$, $n$ is node in $S'$ as well. Formally, $\langle \forall S, S': (n \in \hnodes{S}) \land (S \sqsubset S') \Rightarrow (n \in \hnodes{S}) \rangle$. 
\end{observation}

\noindent With this observation, we assume that nodes once created do not get deleted (ignoring garbage collection). 

\begin{observation}
	\label{lem:hoh-node-val}
	Consider a global state $S$ which has a node $n$ and it is initialized with key $val$. Then in any future state $S'$ the value of $n$ does not change. Formally, $\langle \forall S, S': (n \in \hnodes{S}) \land (S \sqsubset S') \Rightarrow (n \in \hnodes{S}) \land (S.n.val = S'.n.val) \rangle$. 
\end{observation}

\begin{corollary}
\label{cor:hoh-same-node:vertex}
There cannot exist two $\node$s with the same key in the $S.\abs$ of a particular global state $S$. 
\end{corollary}

\begin{observation}
	\label{obs:hoh-locate}
 Consider a global state $S$ which is the post-state of return event of the method $\hloct(key)$ invoked in the $\hadd$ or $\hrem$ or $\hcon$ methods. Suppose the $\hloct$ method returns $\langle n_1, n_2\rangle$. Then in the state $S$, we have, 
 \begin{enumerate} [label=\ref{obs:hoh-locate}.\arabic*]
    \item \label{obs:hoh-locate1}
    $\langle (n_1, n_2 \in \hnodes{S} \rangle$. 
    \item \label{obs:hoh-locate2}
    $\langle (S.lock.acquire(n_1) = true) \wedge (S.lock.acquire(n_2) = true) \rangle$
    \item \label{obs:hoh-locate3}
    $\langle S.n_1.next = S.n_2 \rangle$ 
    
\end{enumerate} 
\end{observation}

\begin{lemma}
\label{lem:hoh-loc-ret}
Consider the global state $S$ which is the post-state of return event of the method $\hloct(key)$ invoked in the $\hadd$ or $\hrem$ or \hcon methods. Suppose the $\hloct$ method returns references as $\langle n_1, n_2 \rangle$. Then in the state $S$, we have that $(S.n_1.val < key \leq S.n_2.val)$ for all nodes whose $next \neq null$.
\end{lemma}

\begin{proof}
Line \ref{lin:hloc2} of $\hloct$ method locks the \emph{\hhead}, in Line \ref{lin:hloc3} initialises $S.n_1$ to $\hhead$ and $S.n_2$ $=$ $S.n_1.next$ in Line \ref{lin:hloc4}. In the last iteration of the while loop in the \lineref{hloc6} the $S.n_1.val$ $<$ $S.n_2.val$ and from the Observation \ref{lem:hoh-node-val} we know that the node key does not change. So, before execution of Line \ref{lin:hloc11}, the $S.n_2.val$ $\geq$ $key$ and $S.n_1.val$ $<$ $S.n_2.val$ and $S.n_1$, $S.n_2$ are locked. Both nodes are belongs to $\hnodes{S}$ and $S.n_1.val < key \leq S.n_2.val$. Also,from the Observations \ref{obs:hoh-locate2}, \ref{obs:hoh-locate3} and \ref{lem:hoh-node-val} the nodes $n_1$ and $n_2$ are locked (do not change), and both are reachable from $\hhead$, hence, the lemma holds even when $\hloct$ returns. 

\end{proof}

\begin{lemma}
\label{lem:hoh-val-change}
For a node $n$ in any global state $S$, we have that $\langle \forall n \in \hnodes{S} \land n.next \neq null: S.n.val < S.n.next.val \rangle$. 
\end{lemma}

\begin{proof}
We prove by induction on all events in $E^H$ that change the $next$ field of the node.\\[0.1cm]
\textbf{Base condition:} Initially, before the first event that changes the next field, we know that ($\hhead.key$ $<$ $\htail.key$) $\land$ $(\hhead, \htail$ $\in$ $\hnodes{S})$.\\[0.1cm] 
\textbf{Induction Hypothesis:} Say, upto $k$ events that change the $next$ field of any node, $\forall n \in \hnodes{S} \land n.next \neq null:$ $S.n.val$ $<$ $S.n.next.val$.\\[0.1cm]
\textbf{Induction Step:}
So, by observation of the code, the $(k+1)^{st}$ event which can change the $next$ field can be only one of the following:
\begin{enumerate}
\item \textbf{Line \ref{lin:hadd5} of $\hadd$ method:}

 Let $S_1$ be the state after the \lineref{hadd3}. We know that when $\hloct$ (\lineref{hadd2}) returns by the \obsref{hoh-locate}, $S_1.n_1$ \& $S_1.n_2$ are locked, $S_1.n_1.next =S_1.n_2$. By the \lemref{hoh-loc-ret} we have $(S_1.n_1.val \leq S_1.n_2.val)$. Also we know from Observation \ref{lem:hoh-node-val} that node value does not change, once initialised. To reach Line \ref{lin:hadd5}, $n_2.val \neq key$ in the \lineref{hadd3} must evaluate to true. Therefore, $(S_1.n_1.val < key < S_1.n_2.val)$. So, a new node $n_3$ is created in the \lineref{hadd4} with the value $key$ and then a link is added between $n_3.next$ and $n_2$ in the \lineref{hadd5}. So this implies $n_3.val < n_2.val$ even after execution of line \ref{lin:hadd5} of $\hadd$ method.


\item \textbf{Line \ref{lin:hadd6} of $\hadd$ method:} By observing the code, we notice that the Line \ref{lin:hadd6} (next field changing event) can be executed only after the $\hloct$ method returns. From \lemref{hoh-loc-ret}, we know that when $\hloct$ returns then $n_1.val$ $<$ key $\leq$ $n_2.val$. To reach Line \ref{lin:hadd6} of $\hadd$ method, Line \ref{lin:hadd3} should ensure that $n_2.val$ $\neq$ $key$. This implies that $n_1.val$ $<$ $key$ $<$ $n_2.val$. From \obsref{hoh-locate3}, we know that $n_1.next$ = $n_2$. Also, the atomic event at Line \ref{lin:hadd6} sets $n_1.next$ = $n_3$ where $n_3.val = key$. \\[0.1cm]
Thus from $n_1.val$ $<$ $n_3.val$ $<$ $n_2.val$ and $n_1.next$ = $n_3$, we get $n_1.val$ $<$ $n_1.next.val$. Since $(n_1, n_2)$ $\in$ $\hnodes{S}$ and hence, $S.n_1.val$ $<$ $S.n_1.next.val$.

\item \textbf{Line \ref{lin:hrem5} of $\hrem$ method:}

Let $S_1$ and $S_2$ be the states after the \lineref{hrem3} and \lineref{hrem5} respectively. By observing the code, we notice that the Line \ref{lin:hrem5} (next field changing event) can be executed only after the $\hloct$ method returns. From Lemma \ref{lem:hoh-loc-ret}, we know that when $\hloct$ returns then $S_1.n_1.val$ $<$ $key$ $\leq$ $S_1.n_2.val$. To reach Line \ref{lin:hrem5} of $\hrem$ method, Line \ref{lin:hrem3} should ensure that $S_1.n_2.val$ $=$ $key$. Also we know from Observation \ref{lem:hoh-node-val} that node value does not change, once initialised. This implies that $S_2.n_1.val$ $<$ ($key$ $=$ $S_2.n_2.val$). From Observation~\ref{obs:hoh-locate3}, we know that $S_2.n_1.next$ = $n_2$. Also, the atomic event at line \ref{lin:hrem5} sets $S_2.n_1.next$ = $S_2.n_2.next$. \\[0.1cm]
We know from Induction hypothesis, $S_2.n_2.val < S_2.n_2.next.val$. Thus from $S_2.n_1.val$ $<$ $S_2.n_2.val$ and $S_2.n_1.next$ = $S_2.n_2.next$, we get $S_2.n_1.val$ $<$ $S_2.n_1.next.val$. Since $(n_1, n_2$ $\in$ $\nodes{S})$ and hence, $S.n_1.val$ $<$ $S.n_1.next.val$.

\end{enumerate}
\end{proof}

\begin{corollary}
\label{cor:hoh-same-node}
There cannot exist two nodes with the same key in the $\abs$ of a particular global state $S$. 
\end{corollary}

\begin{corollary}
\label{cor:hoh-val-abs}
Consider the global state $S$ such that for a node $n$, if there exists a key strictly greater than $n.val$ and strictly smaller than $n.next.val$, then the node corresponding to the key does not belong to $S.\abs$. Formally, $\langle \forall S, n, key$ : $\land$ $(S.n.val < key < S.n.next.val)$ $\implies$ $\node(key)$ $\notin S.\abs \rangle$. 
\end{corollary}
\begin{lemma}
\label{lem:hoh-reach}
	In a global state $S$, for any node $n$, if it is in the list, then $n$ is reachable from $\hhead$. Formally, $\langle \forall S, n: (n \in \hnodes{S}) \implies (S.\hhead \rightarrow^* S.n) \rangle$. 
\end{lemma}
\begin{proof}
We prove by Induction on events that change the next field of the node (as these affect reachability), which are Line \ref{lin:hadd5} \& \ref{lin:hadd6} of $\hadd$ method and Line \ref{lin:hrem5} of $\hrem$ method. It can be seen by observing the code that $\hloct$ and $\hcon$ method do not have any update events.\\[0.1cm]
\textbf{Base step:} Initially, before the first event that changes the next field of any node, we know that $\langle (\hhead, \htail$ $\in$ $\hnodes{S})$ $\land$  $(\hhead$ $\rightarrow^*$ $\htail) \rangle$.\\[0.1cm]
\textbf{Induction Hypothesis:} We assume that the $k^{th}$ event that changes the next field of some node reachable from the $\hhead$.\\[0.1cm]
\textbf{Induction Step:} By observing the code, the $(k+1)^{st}$ event can be one of the following events that change the next field of a node:
\begin{enumerate}
\item \textbf{Line \ref{lin:hadd4} \& \ref{lin:hadd5} of $\hadd$ method:} Let $S_1$ be the state after the \lineref{hadd2}. Line \ref{lin:hadd4} of the $\hadd$ method creates a new node $n_3$ with value $key$. Line \ref{lin:hadd5} then sets $S_1.n_3.next$ $=$ $S_1.n_2$. Since this event does not change the next field of any node reachable from the $\hhead$ of the list, the lemma is not violated. 
\item \textbf{Line \ref{lin:hadd6} of $\hadd$ method:} By observing the code, we notice that the Line \ref{lin:hadd5} (next field changing event) can be executed only after the $\hloct$ method returns.  Let $S_1$ and $S_2$ be the states after the \lineref{hadd3} and \lineref{hadd6} respectively. From Observation \ref{obs:hoh-locate3}, we know that when $\hloct$ returns then $S_1.n_1.next = S_1.n_2$. From Line \ref{lin:hadd4} \& \ref{lin:hadd5} of $\hadd$ method, $(S_1.n_1.next = S_1.n_3)$ $\land$ $(S_1.n_3.next = S_1.n_2)$. It is to be noted that (From Observation \ref{obs:hoh-locate2}), $S_1.n_1$ \& $S_1.n_2$ are locked, hence no other thread can change the $next$ field. Also from Observation \ref{lem:hoh-node-val}, a node's key field does not change after initialization. Before executing Line \ref{lin:hadd6}, $S_1.n_1$ is reachable from $\hhead$. After Line \ref{lin:hadd6}, node $S_2.n_3$ is also reachable from $S_1.n_1$. Thus, we know that $S_2.n_3$ is also reachable from $\hhead$. Formally,	
$(S_2.\hhead \rightarrow^* S_2.n_1) \land (S_2.n_1 \rightarrow S_2.n_3) \implies (S_2.\hhead \rightarrow^* S_2.n_3)$.
\item \textbf{Line \ref{lin:hrem5} of $\hrem$ method:} Let $S_1$ and $S_2$ be the states after the execution of \lineref{hrem3} and \lineref{hrem5} respectively. By observing the code, we notice that the Line \ref{lin:hrem5} (next field changing event) can be executed only after the $\hloct$ method returns. From Observation \ref{obs:hoh-locate2}, we know that when $\hloct$ returns then $S_1.n_1 \& S_1.n_2$ are locked and $S_1.n_1$ is reachable from $\hhead$ and from Line \ref{lin:hrem5} of $\hrem$ method $S_1.n_1.next$ $=$ $S_1.n_2.next$. As $S_1.n_1$ \& $S_1.n_2$ are locked, no other thread can change $S_2.n_1.next$ and $S_2.n_2.next$. Also from Observation \ref{lem:hoh-node-val}, a node's key does not change after initialization. If $S_2.n_2.next$ is reachable from \emph{\hhead}, then it continues to remain reachable. So this event does not violate the lemma.
\end{enumerate}
Hence eventually, $\langle \forall S_2, n: (n \in \nodes{S_2}) \implies (S_2.\hhead \rightarrow^* S_2.n) \rangle$. 
\end{proof}

\begin{lemma}
\label{lem:hoh-change-abs}
	Only the events $write(n_1.next, n_3)$ in \ref{lin:hadd6} of \hadd method and $write(n_1.next, n_2.next)$ in \ref{lin:hrem5} of \hrem method can change the $\abs$.
\end{lemma}
\begin{proof}
It is to be noted that the $\hloct$ and $\hcon$ methods do not have any update events. By observing the code, it appears that the following (write) events of the $\hadd$ and $\hrem$ method can change the $\abs$:

\begin{enumerate}
\item \textbf{Line \ref{lin:hadd4} \& \ref{lin:hadd5} of $\hadd$ method:} 
In Algorithm~\ref{alg:hadd}, let $S_1.\abs$ be the initial state of the $\abs$, such that we know from Line \ref{lin:hadd3} that $key$ $\notin$ $S_1.\abs$. Line \ref{lin:hadd4} of the $\hadd$ method creates a node $n_3$ with value $key$, i.e. $n_3.val = key$. Now, Line \ref{lin:hadd5} sets $S_1.n_3.next$ $=$ $S_1.n_2$. Since this event does not change the next field of any node reachable from the $\hhead$ of the list, hence from Definition \ref{def:habs}, $S_1.\abs$ remains unchanged after these events.


\item \textbf{Line \ref{lin:hadd6} of $\hadd$ method:} 
Let $S_1$ and $S_2$ be the states after the \lineref{hadd3} and \lineref{hadd6} respectively. At line \ref{lin:hadd3}, $true$ evaluation of the condition leads to the execution of $S_1.n_1.next = S_1.n_3$ at Line \ref{lin:hadd6}. Also, $S_1.n_1$ and $S_1.n_2$ are locked, therefore from Observation \ref{obs:hoh-locate}, $\hhead$ $\rightarrow^*$ $S_1.n_1$. From line \ref{lin:hadd5} \& \ref{lin:hadd6} we get: $S_1.n_1$ $\rightarrow$ $S_1.n_3$ $\rightarrow$ $S_1.n_2$. Hence, $\head$ $\rightarrow$ $S_1.n_1$ $\rightarrow$ $S_1.n_3$ $\rightarrow$ $S_1.n_2$ follows. We have $(\head$ $\rightarrow$ $S_2.n_3)$. Thus from Definition \ref{def:habs}, $S_1.\abs$ changes to $S_2.\abs$ $=$ $S_1.\abs$ $\cup$ $n_3$.


\item \textbf{Line \ref{lin:hrem5} of $\hrem$ method:} Let $S_1$ be the state after the \lineref{hrem5}. By observing the code, we notice that the state before execution of Line \ref{lin:hrem5} satisfies that $key$ $\in$ $S_1.\abs$. After execution of line \ref{lin:hrem5}, $\abs$ changes such that $key$ $\notin$ $S_1.\abs$. In Line \ref{lin:hrem5} $S_1.n_1.next$ is set to $S_1.n_2.next$, $S_1.\abs$ remains unchanged follows from Definition \ref{def:habs}.

\end{enumerate}
Hence, only the events $write(n_1.next, n_3)$ in \ref{lin:hadd6} of \hadd method and $write(n_1.next, n_2.next)$ in \ref{lin:hrem5} of \hrem method can change the $\abs$. 
\end{proof}

\begin{corollary}
\label{cor:hoh-change-abs}
Both these events $write(n_1.next, n_3)$ in \ref{lin:hadd6} of \hadd method and $write(n_1.next, n_2.next)$ in \ref{lin:hrem5} of \hrem method can change the $\abs$ are also be the Linearization Points(LPs) of the respective methods. 
\end{corollary}

\begin{observation}
\label{obs:hoh-seq-spec}
Consider a sequential history $\spl{S}$. Let $S$ be a global state in $\spl{S}.allStates$ before the execution of the \mth and $S'$ be a global state just after the return of the \mth $(S \sqsubset S')$. Then we have the sequential specification of all methods as follows, 

	\begin{enumerate}[label=\ref{obs:hoh-seq-spec}.\arabic*]
    \item \label{obs:hoh-addT-seq}
    For a given key, suppose node(key) $\notin$ S.\abs. In this state, suppose \hadd(key) \mth is (sequentially) executed. Then the \hadd \mth will return true and node(key) will be present in $S'.\abs$.
    Formally, $\langle \forall S: (\node(key) \notin S.\abs) \xRightarrow[]{ seq\text{-}add} \spl{S}.\hadd(key, true) \land (S \sqsubset S') \land (\node(key) \in S'.\abs) \rangle$.
    

    
	\item \label{obs:hoh-addF-seq}
	For a given key, suppose node(key) $\in$ S.\abs. In this state, suppose \hadd(key) \mth is (sequentially) executed. Then the \hadd \mth will return false and node(key) will continue to be present in $S'.\abs$.  
	 Formally, $\langle \forall S: (\node(key) \in S.\abs) \xRightarrow[]{seq\text{-}add} \spl{S}.\hadd(key, false) \land (S \sqsubset S') \land (\node(key) \in S'.\abs)\rangle$.
	 
    
    \item \label{obs:hoh-removeT-seq}
   For a given key, suppose node(key) $\in$ S.\abs. In this state, suppose \hrem(key) \mth is (sequentially) executed. Then the \hrem \mth will return true and node(key) will not be present in $S'.\abs$. 
     Formally, $\langle \forall S: (\node(key) \in S.\abs) \xRightarrow[]{seq\text{-}remove} \spl{S}.\hrem(key, true) \land (S \sqsubset S') \land (\node(key) \notin S'.\abs) \rangle$.
    
    
    \item 	\label{obs:hoh-removeF-seq}
  For a given key, suppose node(key) $\notin$ S.\abs. In this state, suppose \hrem(key) \mth is (sequentially) executed. Then the \hrem \mth will return false and node(key) will continue to be not present in $S'.\abs$.  Formally, $\langle \forall S: (\node(key) \notin S.\abs) \xRightarrow[]{seq\text{-}remove} \spl{S}.\hrem(key, false) \land (S \sqsubset S') \land (\node(key) \notin S'.\abs)\rangle$.
    
    \item 	\label{obs:hoh-containT-seq}
    For a given key, suppose node(key) $\in$ S.\abs. In this state, suppose \hcon(key) \mth is (sequentially) executed. Then the \hcon \mth will return true and node(key) will continue to be present in $S'.\abs$. 
   Formally, $\langle \forall S: (\node(key) \in S.\abs) \xRightarrow[]{seq\text{-}contains} \spl{S}.\hcon(key, true) \land (S \sqsubset S') \land (\node(key) \in S'.\abs) \rangle$. 
    
\item \label{obs:hoh-containF-seq}
    For a given key, suppose node(key) $\notin$ S.\abs. In this state, suppose \hcon(key) \mth is (sequentially) executed. Then the \hcon \mth will return false and node(key) will continue to be not present in $S'.\abs$.  
     Formally, $\langle \forall S: (\node(key) \notin S.\abs) \xRightarrow[]{seq\text{-}contains} \spl{S}.\hcon(key, false) \land (S \sqsubset S') \land (\node(key) \notin S'.\abs) \rangle$. 
\end{enumerate} 

\end{observation}

\begin{lemma}
\label{lem:hoh-addT-conc}
If some \hadd method returns true in $E^H$ then
	\begin{enumerate}[label=\ref{lem:hoh-addT-conc}.\arabic*]
    \item \label{lem:hoh-addT-conc-pre}
 The $\hnode(key)$ is not present in the pre-state of $LP$ event of the method. Formally, $\langle \hadd(key, true) \Longrightarrow (\hnode(key) \notin (\prees{\hadd(key, true)}) \rangle$. 
 \item \label{lem:hoh-addT-conc-post}
  The $\hnode(key)$ is present in the post-state of $LP$ event of the method. Formally, \\$\langle \hadd(key, true) \Longrightarrow (\hnode(key) \in  (\postes{\hadd(key, true)}) \rangle$. 
 \end{enumerate}
 
\end{lemma}
\begin{proof} 
\noindent
\begin{itemize}
\item \textbf{\ref{lem:hoh-addT-conc-pre}}: 
From Line \ref{lin:hadd2}, when $\hloct$ returns we know that from the Observation \ref{obs:hoh-locate}, nodes $n_1$ and $n_2$ are locked and ($n_1, n_2$ $\in$ $\hnodes{S})$. Also, $n_1.val$ $<$ $key$ $\leq$ $n_2.val$ from Lemma \ref{lem:hoh-loc-ret}. Now in Line \ref{lin:hadd3}, $n_2.val$ $\neq$ $key$ is evaluated to \emph{true}. Also from Corollary \ref{cor:hoh-val-abs}, we conclude that $\node(key)$ not in the state after \emph{\hloct} returns. And from Observation \ref{lem:hoh-node-val}, no node changes its key value after initialization. So, $\node(key)$ $\notin$ $S.\abs$, where $S$ is the pre-state of the \lp event of the method. Hence, $\langle \hadd(key, true) \Longrightarrow (\hnode(key) \notin \\ (\prees{\hadd(key, true)}) \rangle$.
\item \textbf{\ref{lem:hoh-addT-conc-post}}: From the Lemma \ref{lem:hoh-addT-conc-pre} we get that \emph{\node(key)} is not present in the pre-state of the \lp event. From Lemma \ref{lem:hoh-change-abs}, it is known that only \lp event can change the $S.\abs$. Now after execution of the \lp event i.e. $write(n_1.next, n_3)$ in the Line \ref{lin:hadd6}, $\node(key)$ $\in$ $S'.\abs$, where $S'$ is the post-state of the \lp event of the method. Hence, $\langle \hadd(key, true) \Longrightarrow (\node(key) \in (\postes{\hadd(key, true)}) \rangle$.

\end{itemize}
\end{proof}
\vspace{-3mm}
\begin{lemma}
\label{lem:hoh-addF-conc}
If some \hadd method returns false in $E^H$ then
	\begin{enumerate}[label=\ref{lem:hoh-addF-conc}.\arabic*]
    \item \label{lem:hoh-addF-conc-pre}
 The $\hnode(key)$ is present in the pre-state of $LP$ event of the method. Formally,\\ $\langle \hadd(key, false) \Longrightarrow (\hnode(key) \in  (\prees{\hadd(key, false)}) \rangle$. 
 \item \label{lem:hoh-addF-conc-post}
  The $\hnode(key)$ is present in the post-state of $LP$ event of the method. Formally,\\ $\langle \hadd(key, false) \Longrightarrow (\hnode(key) \in  (\postes{\hadd(key, false)}) \rangle$. 
 \end{enumerate}
\end{lemma}
\begin{proof} 
\noindent
\begin{itemize}
\item \textbf{\ref{lem:hoh-addF-conc-pre}}: 
From Line \ref{lin:hadd2}, when $\hloct$ returns we know that from the Observation \ref{obs:hoh-locate}, nodes $n_1$ and $n_2$ are locked and ($n_1, n_2$ $\in$ $\hnodes{S})$. Also, $n_1.val$ $<$ $key$ $\leq$ $n_2.val$ from Lemma \ref{lem:hoh-loc-ret}. Now in Line \ref{lin:hadd3}, $n_2.val$ $\neq$ $key$ is evaluated to \emph{false}, means \hnode(key) present. Also from Corollary \ref{cor:hoh-val-abs}, we conclude that $\node(key)$ not in the state after \emph{\hloct} returns. And from Observation \ref{lem:hoh-node-val}, no node changes its key value after initialization. So, $\node(key)$ $\in$ $S.\abs$, where $S$ is the pre-state of the \lp event of the method. Hence, $\langle \hadd(key, false) \Longrightarrow (\hnode(key) \in  (\prees{\hadd(key, false)}) \rangle$.
\item \textbf{\ref{lem:hoh-addF-conc-post}}: 

From the Lemma \ref{lem:hoh-addF-conc-pre} we get that \emph{\node(key)} is present in the pre-state of the \lp event. This \lp event  $n_2.val$ $\neq$ $key$ in Line \ref{lin:hadd3} does not change the $S.\abs$, Now after execution of the \lp event the $\node(key)$ also present in the $S'.\abs$, where $S'$ is the post-state of the \lp event of the method. Hence, $\langle \hadd(key, false) \Longrightarrow (\node(key) \in \\ (\postes{\hadd(key, false)}) \rangle$.

\end{itemize}
\end{proof}
\vspace{-3mm}
\begin{lemma}
\label{lem:hoh-removeT-conc}
If some \hrem method returns true in $E^H$ then
	\begin{enumerate}[label=\ref{lem:hoh-removeT-conc}.\arabic*]
    \item \label{lem:hoh-removeT-conc-pre}
 The $\hnode(key)$ is present in the pre-state of $LP$ event of the method. Formally,\\ $\langle \hrem(key, true) \Longrightarrow (\hnode(key) \in (\prees{\hrem(key, true)}) \rangle$. 
 \item \label{lem:hoh-removeT-conc-post}
  The $\hnode(key)$ is not present in the post-state of $LP$ event of the method. Formally, $\langle \hadd(key, true) \Longrightarrow (\hnode(key) \notin  (\postes{\hadd(key, true)}) \rangle$. 
 \end{enumerate}
\end{lemma}
\begin{proof} 
\noindent
\begin{itemize}
\item \textbf{\ref{lem:hoh-removeT-conc-pre}}: 
From Line \ref{lin:hrem2}, when $\hloct$ returns we know that from the Observation \ref{obs:hoh-locate}, nodes $n_1$ and $n_2$ are locked and ($n_1, n_2$ $\in$ $\hnodes{S})$. Also, $n_1.val$ $<$ $key$ $\leq$ $n_2.val$ from Lemma \ref{lem:hoh-loc-ret}. Now in Line \ref{lin:hrem3}, $n_2.val$ $=$ $key$ is evaluated to \emph{true}, means \hnode(key) is present. So, before execution of the \lp event $write(n_1.next, n_2.next)$ in the Line \ref{lin:hrem5} \hnode(key) is also present in the $S.\abs$ and from the Observation \ref{lem:hoh-node-val}), no node changes its key value after initialization. So, $\node(key)$ $\in$ $S.\abs$, where $S$ is the pre-state of the \lp event of the method. Hence, $\langle \hrem(key, true) \Longrightarrow (\hnode(key) \in \\ (\prees{\hrem(key, true)}) \rangle$.
\item \textbf{\ref{lem:hoh-removeT-conc-post}}:
From the Lemma \ref{lem:hoh-removeT-conc-pre} we get that \emph{\node(key)} is present in the pre-state of the \lp event. This \lp event $write(n_1\\.next, n_2.next)$ in the \lineref{hrem5} changes the $S.\abs$. Now after execution of the \lp event the $\node(key)$ will not present in the $S'.\abs$, where $S'$ is the post-state of the \lp event of the method. Hence, $\langle \hrem(key, true) \Longrightarrow \\ (\node(key) \notin (\postes{\hrem(key, true)}) \rangle$.
\end{itemize}
\end{proof}

\begin{lemma}
\label{lem:hoh-removeF-conc}
If some \hrem method returns false in $E^H$ then
	\begin{enumerate}[label=\ref{lem:hoh-removeF-conc}.\arabic*]
    \item \label{lem:hoh-removeF-conc-pre}
 The $\hnode(key)$ is not present in the pre-state of $LP$ event of the method. Formally, $\langle \hrem(key, false) \Longrightarrow (\hnode(key) \notin  (\prees{\hrem(key, false)}) \rangle$. 
 \item \label{lem:hoh-removeF-conc-post}
  The $\hnode(key)$ is not present in the post-state of $LP$ event of the method. Formally, $\langle \hrem(key, false) \Longrightarrow (\hnode(key) \notin  (\postes{\hadd(key, false)}) \rangle$. 
 \end{enumerate}
\end{lemma}
\begin{proof} 
\noindent
\begin{itemize}
\item \textbf{\ref{lem:hoh-removeF-conc-pre}}: 
From Line \ref{lin:hrem2}, when $\hloct$ returns we know that from the Observation \ref{obs:hoh-locate}, nodes $n_1$ and $n_2$ are locked and ($n_1, n_2$ $\in$ $\hnodes{S})$. Also, $n_1.val$ $<$ $key$ $\leq$ $n_2.val$ from Lemma \ref{lem:hoh-loc-ret}. Now in Line \ref{lin:hrem3}, $n_2.val$ $=$ $key$ ( the \lp event $read(n_2.val)$) is evaluated to \emph{false}, means \hnode(key) is not present. So, before execution of the \lp the \hnode(key) is not present in the $S.\abs$, where $S$ is the pre-state of the \lp event of the method. Hence, $\langle \hrem(key, false) \Longrightarrow (\hnode(key) \in (\prees{\hrem(key, false)}) \rangle$.
\item \textbf{\ref{lem:hoh-removeF-conc-post}}: 
From the Lemma \ref{lem:hoh-removeF-conc-pre} we get that \emph{\node(key)} is not present in the pre-state of the \lp event. This \lp event \\ $(read(n_2.val) = key)$ in the \lineref{rem3} does not change the $S.\abs$. Now after execution of the \lp event the $\node(key)$ will not present in the $S'.\abs$, where $S'$ is the post-state of the \lp event of the method. Hence, $\langle \hrem(key, falase) \Longrightarrow (\node(key) \notin  (\postes{\hrem(key, false)}) \rangle$.

\end{itemize}
\end{proof}

\begin{lemma}
\label{lem:hoh-containsT-conc}
If some \hcon method returns true in $E^H$ then
	\begin{enumerate}[label=\ref{lem:hoh-containsT-conc}.\arabic*]
    \item \label{lem:hoh-containsT-conc-pre}
 The $\hnode(key)$ is present in the pre-state of $\lp$ event of the method. Formally,\\ $\langle \hcon(key, true) \Longrightarrow (\hnode(key) \in (\prees{\hcon(key, true)}) \rangle$. 
 \item \label{lem:hoh-containsT-conc-post}
  The $\hnode(key)$ is present in the post-state of $\lp$ event of the method. Formally, \\$\langle \hcon(key, true) \Longrightarrow (\hnode(key) \in  (\postes{\hcon(key, true)}) \rangle$. 
 \end{enumerate}
\end{lemma}
\begin{proof} 
\noindent
\begin{itemize}
\item \textbf{\ref{lem:hoh-containsT-conc-pre}}: 
From Line \ref{lin:hcon2}, when $\hloct$ returns we know from the Observation \ref{obs:hoh-locate} that, nodes $n_1$ and $n_2$ are locked and ($n_1, n_2$ $\in$ $\hnodes{S})$. Also, $n_1.val$ $<$ $key$ $\leq$ $n_2.val$ from Lemma \ref{lem:hoh-loc-ret}. Now in Line \ref{lin:hcon3}, $n_2.val$ $=$ $key$ (the \lp event $read(n_2.val)$) is evaluated to \emph{true} and this \lp event does not change the $S.\abs$. From Observation \ref{lem:hoh-node-val}, no node changes its key value after initialization. So, $\node(key)$ $\in$ $S.\abs$, where $S$ is the pre-state of the \lp event of the method. Hence, $\langle \hcon(key, true) \Longrightarrow (\hnode(key) \in \\ (\prees{\hcon(key, true)}) \rangle$.
\item \textbf{\ref{lem:hoh-containsT-conc-post}}:
From the Lemma \ref{lem:hoh-containsT-conc-pre} we get that \emph{\node(key)} is present in the pre-state of the \lp event. This \lp event $(read(n_2.val) = key)$ in the \lineref{hcon3} does not change the $S.\abs$. Now after execution of the \lp event the $\node(key)$ will be present in the $S'.\abs$, where $S'$ is the post-state of the \lp event of the method. Hence, $\langle \hcon(key, true) \Longrightarrow (\node(key) \notin (\postes{\hcon(key, true)}) \rangle$.

\end{itemize}
\end{proof}

\begin{lemma}
\label{lem:hoh-containsF-conc}
If some \hcon method returns false in $E^H$ then
	\begin{enumerate}[label=\ref{lem:hoh-containsF-conc}.\arabic*]
    \item \label{lem:hoh-containsF-conc-pre}
 The $\hnode(key)$ is not present in the pre-state of $\lp$ event of the method. Formally, $\langle \hcon(key, false) \Longrightarrow (\hnode(key) \notin (\prees{\hcon(key, false)}) \rangle$. 
 \item \label{lem:hoh-containsF-conc-post}
  The $\hnode(key)$ is not present in the post-state of $\lp$ event of the method. Formally, $\langle \hcon(key, false) \Longrightarrow (\hnode(key) \notin  (\postes{\hcon(key, false)}) \rangle$. 
 \end{enumerate}
\end{lemma}
\begin{proof} 
\noindent
Similar argument as Lemma \ref{lem:hoh-containsT-conc}.
\end{proof}

\begin{lemma}

\label{lem:hoh-pre-ret}
\textbf{hoh-locking-list Specific Equivalence:}
Consider a concurrent history $H$ and a sequential history $\spl{S}$. Let $m_x, m_y$ be \mth{s} in $H$ and $\spl{S}$ respectively. Suppose the following are true (1) The \abs in the pre-state of $m_x$'s \lp in $H$ is the same as the \abs in the pre-state of $m_y$ in $\spl{S}$;  (2) The \inv events of $m_x$ and $m_y$ are the same. Then (1) the $\rsp$ event of $m_x$ in $H$ must be same as $\rsp$ event of $m_y$ in $\spl{S}$; (2) The \abs in the post-state of $m_x$'s \lp in $H$ must be the same as the \abs in the post-state of $m_y$ in $\spl{S}$. Formally, $\langle \forall m_x \in \mths{E^{H}}, \forall m_y \in \mths{E^{\spl{S}}}: (\prees{x} = \prems{y}) \wedge (\inves{x} = \invms{y}) \Longrightarrow (\postes{x} = \\ \postms{y}) \wedge (\retes{x} = \retms{y}) \rangle$.


\end{lemma}
\begin{proof}
 Let us prove by contradiction.
So we assume that, 
\begin{equation}
 \label{lab:eq:pre-inv-resp-hoh}
     \begin{split}
         \langle (\prees{m_x} = \prems{m_y}) \wedge \\ (\inves{m_x} =  \invms{m_y}) \Longrightarrow  (\retes{m_x} \neq \retms{m_y}) \rangle
     \end{split}
 \end{equation}

We have the following cases that $\inves{m_x}$ is invocation of either of these methods:

\begin{enumerate}
\item \textbf{$m_x.inv$ is \hadd(key) Method:}
\begin{itemize}
\item \textbf{$m_x.resp$ = true:} Given that the method $m_x.resp$ which is \emph{\hadd(key)} returns $true$, we know that from the Lemma \ref{lem:hoh-addT-conc}, \textit{\node(key)} $\notin$ $\prees{\add(key, true)}$. But since from assumption equation \ref{lab:eq:pre-inv-resp-hoh}, $(\retes{m_x} \neq \retms{m_y})$, $\retms{m_y}$ is false. However, from the Observation \ref{obs:hoh-addT-seq}, if $\node(key)$ $\notin$ pre-state of $\lp$ of $\hadd$ method, then the $\hadd(key, true)$ method must return $true$ in $E^{\spl{S}}$. This is a contradiction.

\item \textbf{$m_x.resp$ = false:} Given that the method $m_x.resp$ which is \emph{\hadd(key)} returns $false$, we know that from the Lemma \ref{lem:hoh-addF-conc}, \textit{\node(key)} $\in$\\ $\prees{\hadd(key, false)}$. But since from assumption in equation \ref{lab:eq:pre-inv-resp-hoh}, $(\retes{m_x} \neq \retms{m_y})$, $\retms{m_y}$ is false. However, from the Observation \ref{obs:hoh-addF-seq}, if $\node(key)$ $\in$ pre-state of $\lp$ of $\hadd$ method, then the $\hadd(key, false)$ method must return $false$ in $E^{\spl{S}}$. This is a contradiction.


\end{itemize}

\item \textbf{$m_x.inv$ is \hrem(key) Method:}
\begin{itemize}
\item \textbf{$m_x.resp$ = true:} Given that the method $m_x.resp$ which is \emph{\hrem(key)} returns $true$, we know that from the Lemma \ref{lem:hoh-removeT-conc}, \textit{\node(key)} $\in$ \\ $\prees{\hrem(key, true)}$. But since from assumption in equation \ref{lab:eq:pre-inv-resp-hoh}, $(\retes{m_x} \neq \retms{m_y})$, $\retms{m_y}$ is false. However, from the Observation \ref{obs:hoh-removeT-seq}, if $\node(key)$ $\notin$ pre-state of $\lp$ of $\hrem$ method, then the $\hrem(key, true)$ method must return $true$ in $E^{\spl{S}}$. This is a contradiction.

\item \textbf{$m_x.resp$ = false:} Given that the method $m_x.resp$ which is \emph{\hrem(key)} returns $false$, we know that from the Lemma \ref{lem:hoh-removeF-conc}, \textit{\node(key)} $\notin$ \\ $\prees{\hrem(key, false)}$. But since from assumption in equation \ref{lab:eq:pre-inv-resp-hoh}, $(\retes{m_x} \neq \retms{m_y})$, $\retms{m_y}$ is false. However, from the Observation \ref{obs:hoh-removeF-seq}, if $\node(key)$ $\in$ pre-state of $\lp$ of $\hrem$ method, then the $\hrem(key, false)$ method must return $false$ in $E^{\spl{S}}$. This is a contradiction.
\end{itemize}

\item \textbf{$m_x.inv$ is \hcon(key) Method:}
\begin{itemize}
\item \textbf{$m_x.resp$ = true:} Given that the method $m_x.resp$ which is \emph{\hcon(key)} returns $true$, we know that from the Lemma \ref{lem:hoh-containsT-conc}, \textit{\node(key)} $\in$ \\ $\prees{\hcon(key, true)}$. But since from assumption in equation \ref{lab:eq:pre-inv-resp-hoh}, $(\retes{m_x} \neq \retms{m_y})$, $\retms{m_y}$ is false. However, from the Observation \ref{obs:hoh-containT-seq}, if $\node(key)$ $\in$ pre-state of $\lp$ of $\hcon$ method, then the $\hcon(key, true)$ method must return $true$ in $E^{\spl{S}}$. This is a contradiction.

\item \textbf{$m_x.resp$ = false:} Given that the method $m_x.resp$ which is \emph{\hcon(key)} returns $false$, we know that from the Lemma \ref{lem:hoh-containsF-conc}, \textit{\node(key)} $\notin$ \\ $\prees{\hcon(key, false)}$. But since from assumption in equation \ref{lab:eq:pre-inv-resp-hoh}, $(\retes{m_x} \neq \retms{m_y})$, $\retms{m_y}$ is false. However, from the Observation \ref{obs:hoh-containF-seq}, if $\node(key)$ $\notin$ pre-state of $\lp$ of $\hcon$ method, then the $\hcon(key, false)$ method must return $false$ in $E^{\spl{S}}$. This is a contradiction.



\end{itemize}
\end{enumerate}
Thus we conclude that the $\rsp$ event of $m_x$ in $H$ must be same as $\rsp$ event of $m_y$ in $\spl{S}$. Formally, $\langle \retes{m_x} = \retms{m_y} \rangle$. 

\end{proof}

\begin{lemma}
\label{lem:hist-lin-hoh}
All histories ${H}$ generated by the \hoh  are \lble.
\end{lemma}
\vspace{-3mm}
\begin{proof}
 Proof follows based on the  \lemref{hoh-pre-ret}, \lemref{conce-seqe-pre} and  \lemref{ret-lin}.
\end{proof}

\section{Discussion \& Conclusion}
\label{sec:conc}

\cds{s} offer great performance benefits over their sequential counterparts. But one of the greatest challenges with \cds{s} is developing correct structures and then proving their correctness either through automatic verification or through hand-written proofs \cite{Derrick-fm2011}. We believe that the techniques which help prove correctness of \cds{s} can also guide in developing new \cds{s}.

Several techniques have been proposed for proving \lbty - a \cc for concurrent objects. But \lp{s} continue to remain most popular way of illustrating correctness of \cds among practitioners since it is seems intuitive and constructive. One of the main challenges with the \lp based approach is to identify the correct \lp{s} of a \cds. Identifying the correct \lp{s} can be deceptively wrong in many cases. In fact in many cases, the \lp identified or even worse the \cds could be wrong. 

Considering the complexity of developing a \cds and verifying its correctness, we address the most basic problem of this domain in this paper: given the set of \lp{s} of a \cds, how to show its correctness? We assume that we are given a \cds and its \lp{s}. We have developed a hand-crafted technique of proving correctness of the \cds{s} by validating it \lp{s}. We believe that our technique can be applied to prove the correctness of several commonly used \cds{s} developed in literature such as Lock-free Linked based Sets \cite{Valoispodc1995}, \lazy \cite{Heller-PPL2007, MauriceNir}, Skiplists \cite{Levopodis2006} etc. Our technique will also work for \cds{s} in which the \lp{s} of a \mth might lie outside the \mth such as \lazy. To show the efficacy of this technique, we show the correctness of \lazy and hand-over-hand locking list (\emph{\hoh}) \cite{Bayerai1977, MauriceNir}. 


As a part of our technique, we have identified a \cdse (\defref{pre-resp}). We show that any \cds for which \cdse is true and satisfies our assumptions on the \lp{s}, is \lble. Thus, we would like to view \cdse as a generic template. It  has to be proved taking into account the properties and invariants of the underlying \cds. In \secref{ds-proofs}, we demonstrate this technique by giving the correctness proof of \cdse for \lazy and of \hoh. 


In \secref{gen-proof}, we postulated that the hand-crafted mechanism of proving the \cdse for a given \cds might bring out errors in the \lp{s} proposed if they are incorrect. Further, we also theorized that this technique might give new insights for designing new \cds{s}. But the actual details of these can be accomplished are still not clear. Ideally, a programmer should have a set of \emph{design patterns} using which s/he would be able to develop correct \cds which are also efficient. As observed earlier, this has been acknowledged as a very complicated problem. We believe that we have just scratched the surface of this problem in this paper. We plan to explore further in this direction as a part of future work.

To this end, Transactional Memory Systems \cite{HerlMoss:1993:SigArch} or TMs  have been proposed as an alternative to address this challenge of designing efficient concurrent structures. But the design of efficient \cds using TMs would again require the programmer to designate portions of code as transactions. Not doing this properly could again lead to loss in efficiency and/or correctness. Hence, we believe that the TMs can help with this objective although they may not be the final solution. As a part of our future work, we also plan to explore how TMs can help us achieve the objective. 

An important point to be noted with our approach: we assumed that only \lp events change the \abds (\asmref{change-abds}). Although this is true in case of many \cds{s} considered, this is not always true. As an example consider a shared array which has an lock for each entry and is modified by multiple threads concurrently. Threads wishing to update several entries in a \lble manner can obtain locks on the relevant entries of the array using two-phase locking (2PL) and then perform the updates. In this case, one can choose any event between the last locking and the first unlocking as the \lp. But then, the \lp event is not where all the updates to the shared entries of the array takes place. So with this kind of 2PL usage, our technique will not directly work. In that case, we believe that we have to consider the notion of \emph{Linearization Blocks} instead of Linearization Points. We plan to explore this notion in future. On the other hand, we believe that our technique will work for those \cds{s} which has at least one wait-free \mth (like the contains \mth in the case of lazy-list).

\bibliographystyle{plain}
\bibliography{biblio}



\end{document}